\documentclass[11pt]{article}
\usepackage[margin = 1in]{geometry}

\usepackage{natbib}
\usepackage{mathtools}
\usepackage{cancel}
\usepackage{mathrsfs}
\usepackage{float}

\newcount\Comments  %

\newcommand{\kibitz}[2]{\ifnum\Comments=1{\color{#1}{#2}}\fi}

\DeclareMathOperator{\Hessian}{Hess}

\newcount\TODO
\newcommand{\kibitzTODO}[2]{\ifnum\TODO=1{\color{#1}{#2}}\fi}

\newcount\Draft 
\ifnum\Draft=1
\fi

\usepackage{url}
\usepackage{etex}
\usepackage{cleveref}
\usepackage{latexsym}
\usepackage{amsfonts}
\usepackage{bbm}
\usepackage{ifthen}
\usepackage{enumerate}
\usepackage{psfrag}
\usepackage{floatflt}
\usepackage{verbatim}
\usepackage{dsfont}
\usepackage{scalefnt}
\usepackage{epsfig,epstopdf}
\usepackage{calc}
\usepackage{graphicx}
\usepackage{subcaption}
\usepackage{multicol}
\usepackage{tikz}
\usetikzlibrary{shapes.multipart,arrows,automata}
\usetikzlibrary{positioning,arrows.meta,bending}
\usetikzlibrary{patterns,calc,intersections,through,backgrounds,decorations.pathreplacing}
\usepackage{xstring}
\usepackage{pgfplots}

\usepackage{amsthm}		%
\usepackage{thmtools, thm-restate}
\usepackage{algorithmic}
\usepackage[linesnumbered,ruled,vlined]{algorithm2e}

\usepackage{color, xcolor}
\colorlet{darkblue}{blue!40!black}
\definecolor{auburn}{rgb}{0.43, 0.21, 0.1}
\definecolor{orange}{rgb}{1, 0.5, 0}
\definecolor{lightblue}{rgb}{0.1176, 0.5647, 1}

\usepackage{listings}

\lstdefinestyle{mycodestyle}
{
basicstyle = \ttfamily,
mathescape=true,
breaklines=true,
}

\newcount\ColCMT  %

\newcommand{\colcmts}[2]{\ifnum\ColCMT=1{\color{#1}{#2}}\fi}

\theoremstyle{plain}
\newtheorem{theorem}{Theorem}

\theoremstyle{definition}
\newtheorem{definition}{Definition} 
\newtheorem{assumption}{Assumption} 
\newtheorem{example}{Example}
\newtheorem{lemma}{Lemma}
\newtheorem{proposition}{Proposition}

\newtheorem{corollary}{Corollary}

\theoremstyle{definition}
\newtheorem{remark}{Remark}

\DeclareMathOperator{\var}{var}

\DeclareMathOperator{\Binom}{Binomial}

\DeclareMathOperator{\Bias}{Bias}

\newcommand{\Thetaseg}{\hat \Theta_y^{(\textrm{\normalfont{seg}})}}
\newcommand{\ThetasegG}{\hat \Theta_{y}^{(\textrm{\normalfont{g-seg}})}}

\newcommand{\ThetasegP}{\hat \Theta_{Y_p}^{(\textrm{\normalfont{seg}})}}
\newcommand{\Thetasegopt}{\hat \Theta_y^{\ast(\textrm{\normalfont{seg}})}}
\newcommand{\ThetasegoptG}{\hat \Theta_y^{\ast(\textrm{\normalfont{g-seg}})}}

\newcommand{\Thetasegpbigy}{\hat \Theta_{Y_p}^{(\textrm{\normalfont{seg}})}}

\newcommand{\Thetatrp}{\hat \Theta_y^{(\textrm{\normalfont{route}})}}
\newcommand{\Thetatrpopt}{\hat \Theta_y^{\ast(\textrm{\normalfont{route}})}}
\newcommand{\Thetatrpoptp}{\hat \Theta_{y_p}^{\ast(\textrm{\normalfont{route}})}}
\newcommand{\ThetatrpoptP}{\hat \Theta_{Y_p}^{\ast(\textrm{\normalfont{route}})}}

\newcommand{\bb}[1]{\mathbb{#1}}
\newcommand{\1}{\mathbf 1}

\newcommand{\diag}{\mathop{\mathrm{diag}}}

\usepackage{mhequ}
\usepackage{amsmath,amssymb}
\usepackage{booktabs}

\newcommand{\be}{\begin{equs}}
\newcommand{\ee}{\end{equs}}

\newcommand{\Thetasegprime}{\hat{\Theta}_y^{\prime(\textrm{\normalfont{seg}})}}

\pgfplotsset{
    discard if not/.style 2 args={
        x filter/.code={
            \edef\tempa{\thisrow{#1}}
            \edef\tempb{#2}
            \ifx\tempa\tempb
            \else
                
            \fi
        }
    }
}

\title{Efficiency of ETA Prediction} %

\author{
Chiwei Yan%
\thanks{Department of Industrial Engineering and Operations Research, University of California, Berkeley.}
\and
James Johndrow%
\thanks{Department of Statistics and Data Science, Wharton School, University of Pennsylvania.}
\and Dawn Woodard%
\thanks{LinkedIn, work done while at Uber Technologies.}
\and Yanwei Sun\thanks{Department of Analytics, Marketing and Operations, Imperial College Business School.}
\bigskip
}

\date{\today}

\begin{document}

\maketitle

\begin{abstract}
	Modern mobile applications such as navigation services and ride-sharing platforms rely heavily on geospatial technologies, most critically predictions of the time required for a vehicle to traverse a particular route, or the so-called estimated time of arrival (ETA). There are various methods used in practice, which differ in terms of the geographic granularity at which the predictive model is trained --- e.g., segment-based methods predict travel time at the level of road segments (or a combination of several adjacent road segments) and then aggregate across the route, whereas route-based methods use generic information about the trip, such as origin and destination, to predict travel time. Though various forms of these methods have been developed, there has been no rigorous theoretical comparison regarding their accuracies, and empirical studies have, in many cases, drawn opposite conclusions. 
		We provide the first theoretical analysis of the predictive accuracy of various ETA prediction methods and argue that maintaining a segment-level architecture in predicting travel time is often of first-order importance. Our work highlights that the accuracy of ETA prediction is driven not just by the sophistication of the model but also by the spatial granularity at which those methods are applied.
\end{abstract}

\section{Introduction} \label{sec:intro}

	Geospatial (maps) technologies underlie a broad spectrum of modern mobile applications. For example, consumer-facing navigation applications (such as Google Maps and Waze) provide recommended routes along with associated times, as well as turn-by-turn navigation along those routes.  Geospatial technologies are also the foundation of decision systems for  ride-sharing (such as Uber, Lyft, Didi Chuxing, and Ola) and delivery platforms (such as Uber Eats and Doordash).  For example, riders on these platforms are presented with estimated pickup time and time to arrival, and drivers are provided with turn-by-turn navigation. Matching and pricing decisions  on these platforms also heavily rely on mapping inputs to optimize efficiency and reliability \citep{yan2020dynamic}.

	An important geospatial technology is the prediction of the time required for a driver (or biker or pedestrian) to travel a particular route in the road network or the so-called estimated time of arrival (ETA). %
	Modern methods leverage location data traces from past vehicle trips in the road network, so-called ``floating-car'' data, typically gathered (with permission) from users of a particular application, such as a consumer-facing navigation service. Location traces from driver trips in the road network are processed by a ``map-matching'' algorithm to obtain travel time observations on each road segment along the driver's trajectory \citep{quddus2007current}. This data provides detailed information about traffic and travel speed patterns throughout the road network and along individual routes, and so is the foundation of modern methods for ETA prediction at scale. A feature that distinguishes different prediction methods is \emph{the level of geographic granularity} at which the predictive model is trained. The geographic unit can be a single road segment, a combination of multiple adjacent and connecting segments (aka ``super-segment'', see \citealt{derrow2021eta} for the implementation in Google Maps), or the entire trip. To be more specific, \emph{segment-based} methods rely heavily on the underlying road network and predict travel time at the level of road segments or super-segments, and then aggregate across the route (see, e.g., \citealt{hofleitner2012learning} and \citealt{jenelius2013travel}). On the other hand, with a large and growing amount of trip data being collected by firms such as ride-sharing platforms, a class of more recently proposed \emph{route-based} methods hinge less on the road network and use generic information about the origin, destination, departure time and sometimes route characteristics to predict travel time. This started with the \emph{k}-nearest neighbors approach proposed in \cite{wang2016simple}, where the prediction of travel time on a new route was done using travel times of historical trips that have similar origins, destinations and departure times as those of the predicting route. Then a number of neural-network based approaches were developed for route-based prediction \citep{jindal2017unified,li2018multi,yuan2020effective}, including work from Didi Chuxing, a major ride-sharing provider.

	Though many variations of these methods have been proposed in the literature and used in practice, there is a very limited theoretical understanding of the accuracy of these methods.
	Most work in this space is empirical, and these empirical studies have, in many cases, drawn opposite conclusions (see, e.g., \citealt{wang2018learning,yuan2020effective,derrow2021eta,wang2016simple}). Indeed, the comparison is not trivial. Segment-based methods have the advantage of a larger sample size as there are more individual traversals on a segment level. However, the estimation can accumulate errors due to aggregating over road segments. On the other hand, route-based methods can have the advantage of absorbing errors among segment travel times, but it is often at the cost of a smaller sample size. Part of the confusion in the empirical analyses stems from some papers assuming that the route that the driver will take is known, whereas other papers assume that the route is uncertain (and so must be estimated or ignored); the latter naturally disadvantages segment-based methods. However, even papers that analyze cases where the route is known sometimes conclude that route-based methods are superior \citep{wang2016simple}.
 Due to the uncertainty about the best approach to take, several recent papers have tried combining segment-based and route-based methods into a single model \citep{wangfuye,hu2022deepreta}, or using methods that model travel time at the level of the ``super-segment'' (a sequence of segments) \citep{wang2014travel,derrow2021eta}.

	To fill this gap in understanding, we conduct rigorous analyses comparing segment-based and route-based methods in terms of their predictive accuracy as a function of the training data sample size, i.e., in terms of statistical efficiency. We now give a brief summary of our framework, analysis, and major results.
	
	\vspace{1mm}
	
	\noindent\textbf{Framework.}  We consider a road network consisting of a set of road segments (directed edges) between intersections (vertices). The training data consists of individual trips traversed at different times, each of which is along a route that consists of a sequence of adjacent road segments. The travel time on each road segment of that route is observed, and the total travel time of a trip is the sum of the observed segment travel times along the route. Our goal is to predict the total travel time on a new trip, given the dataset of historical trips. We ask questions such as: What is the most accurate ETA prediction method? and, How do various methods used in practice compare? To precisely answer these questions, we construct a data-generating process based on a general mean function that depends flexibly on features like distance and time of day, and that also incorporates parameters associated with the idiosyncratic travel speed effects of individual road segments. Travel times are allowed to be correlated across the road segments of a trip.
 Under this data-generating process, we first explicitly characterize the optimal predictor that has the lowest predictive mean squared error. The optimal predictor, though having the best accuracy, is computationally intractable to implement in real-time mapping services and requires full knowledge of the segment travel times' covariance structure which can be hard to obtain in practice. %
	This calls for the need to understand the accuracy of simpler and more practical methods. We formally define a family of segment-based methods and route-based methods that resemble many practical methods proposed in the literature and used in practice.
	
	\vspace{1mm}
	
	\noindent\textbf{Analysis and Results.} We start with a finite-sample setting where a set of historical trips are given on an arbitrary road network. When segment travel times are \emph{non-negatively} correlated over the network, we show that the predictive mean squared error of the optimal segment-based method, where prediction is made on each individual road segment and then aggregated over the predicting route, is always lower than those of a wide range of route-based methods. 
 We then extend our analysis to an asymptotic setting where the number of trip observations grows with the size of the road network, and trip routes are sampled randomly from a generic route distribution. We show that a very simple class of segment-based methods with minimum information requirement %
	can asymptotically dominate popular route-based methods. Furthermore, under a broad range of trip-generating processes on a grid network, we show that this class of simple segment-based methods is at least as good (up to a logarithmic factor) as any possible predictor. In other words, segment-based methods are asymptotically optimal up to a logarithmic factor. Numerical experiments based on realistic parameters reveal that the accuracy of the segment-based methods is extremely competitive --- the error of the segment-based method is often very close to that of the optimal method, even not at the asymptotic limit. %

 Our analysis is greatly facilitated by the fact that minimizing the predictive mean squared error is mathematically equivalent to minimizing the estimation error for the expected travel time  (i.e., the conditional mean of travel time given the input features). This allows us to focus our analysis on the accuracy of estimating the expected travel time  (rather than, for example, the variance of travel time), which is summable across road segments. We believe that focusing on the accuracy of estimation of the expected travel time is reasonable because modern navigation applications have chosen to focus mainly on communicating the expected travel time to users, ignoring the variance (which is harder to communicate and less interpretable for many users).\footnote{As far as we know, currently there is only one case where a navigation provider gives a measure of variability of travel time to the user. This is for the web interface for Google Maps, in the case where the user inputs a future time for departure/arrival. In this case, Google Maps provides a range (interval prediction) for travel time. To our knowledge this information isn’t provided in the Google Maps app, or by any other common mapping providers like Apple Maps.} 
	
	\vspace{1mm}
	
	In short, our paper makes a contribution to the literature and practice of the ETA prediction problem by providing important theoretical underpinnings. Through extensive analyses, our paper argues that maintaining a segment-level architecture in predicting travel time is often of first-order importance. This gives important practical guidance to mapping services as they improve their underlying predictive models. The remainder of the paper is organized as follows. %
	In Section \ref{sec:finite}, we introduce the model setup and conduct finite sample analysis, meaning analysis for a given set of historical trips. In Section \ref{sec:asymptotic}, we analyze an asymptotic setting in which the number of trip observations grows with the road network size, and trip observations are sampled randomly from a route distribution. %
	We conclude with a brief discussion in Section \ref{sec:dis}. All proofs and various auxiliary results are presented in the supplement. A companion Jupyter notebook can be found at \url{https://github.com/yanchiwei/eta/blob/main/examples.ipynb} to reproduce all the examples presented in the paper.

	\section{Model and Finite Sample Analysis}
	\label{sec:finite}
	
	We first consider a standard travel time setting, where we are given $N$ historical trips on an arbitrary road network $(\mathcal{V}, \mathcal{S})$ where $\mathcal{V}$ is a vertex set and $\mathcal{S}$ is an edge (road segment) set. Let $y_1, \ldots, y_N$ be the routes for each trip, and $[N] := \{1,\ldots, N\}$, so that $y_{[N]}$ is the set of routes. Put $|y_n|$ as the number of segments on route $n$. Each route consists of a sequence of distinct road segments $s\in \mathcal{S}$. Most simply, think of a road segment as the primitive used in the standard representation of road graphs, i.e., a directed section of roadway that is uninterrupted by intersections and has constant values for features like the number of lanes and speed limit. A more sophisticated representation of road graphs also fits into our framework: one which incorporates turn effects by defining a road segment $s$ to be a section of roadway (with constant feature values) that is followed by a specific turn direction. For example, segment $s$ can represent a particular directed section of highway that is followed by the turn onto an exit ramp, and the next segment $s'$ in the route could be the exit ramp that is followed by a left turn onto a minor road (see e.g., Section 4.1 of \cite{delling2017customizable}). Let $T_{n,s}$ be the travel time on segment $s\in y_n$ for the $n^{\textrm{th}}$ observed trip, and denote the $n^{\textrm{th}}$ trip by $\mathcal{T}_n = \{y_n, \{T_{n,s}\}_{s\in y_n}\}$. %

	\subsection{Generative Process}
	
	We first discuss the generative process that we assume for travel times. In practice, the segment travel time $T_{n,s}$ and the route travel time  $\sum_{s\in y_n}T_{n,s}$ are affected by the set of observed features $V_s$ of the road segments such as the number of lanes, speed limit, segment length, and road classification (local road, highway, arterial, etc.). The travel times are also affected by a set of trip-level characteristics $W_n$, such as time of week and weather conditions. In addition, there are unobserved idiosyncratic characteristics of the road segments that affect their travel times. For example, some segments have bad traffic conditions, a poor layout of the lanes, road constructions, or a slow traffic light, which the mapping services typically don't observe directly outside of the location trace data. Following this physical understanding, we assume that the segment travel times $T_{n,s}=g(\theta_s,V_s,W_n)+\varepsilon_{n,s}$, where $g(\theta_s, V_s, W_n)$ is the true mean with some function $g(\cdot)$, $\theta_s$ is an unobserved feature vector for each road segment $s$, and $\varepsilon_{n,s}$ is the error term with mean $0$. Let $\theta := [\theta_s]_{s\in\mathcal{S}}$. %
    The mean of the travel time on route $y_n$ is then $\sum_{s\in y_n} g(\theta_s, V_s, W_n)$. For mathematical tractability, we analyze a simplified generative model that has an additive structure, i.e., we assume that $g(\theta_s, V_s, W_n) = \theta_s + h(V_s, W_n)$ for some function $h(\cdot)$ and for $\theta_s$ a scalar that can capture the fact that a particular road segment $s$ has faster or slower average travel time. This generative model, while simple, captures the most foundational characteristics of typical traffic data, specifically a mean structure that depends in a potentially nonlinear way on $W_n$ and $V_s$, as well as idiosyncratic travel time effects at the level of the road segment. Such additive models are common in the statistics literature, where they are called mixed-effects models \cite{pinheiro2006mixed}.  
    These discussions lead us to the following assumptions regarding the generative process of $T_{n,s}$. %

	\begin{assumption}
		\label{assump:1}
		We make the following assumptions about $T_{n,s}$,
		\begin{enumerate}
 			\item $T_{n,s} = \theta_s + h(V_s, W_n) + \varepsilon_{n,s}$ for some function $h$ of the input features $V_s,W_n$, and for $\theta_s$ a scalar capturing road segment travel time effects.
			\item For every trip $n$, the errors $\{\varepsilon_{n,s}\}_{s\in y_n}$ are drawn from a joint distribution with mean $0$ for all $\varepsilon_{n,s}$ and covariances $\{\sigma_{s,t}\}_{s, t \in y_n}$, where $\sigma_{s,s} = \sigma_s^2$ is the variance of the error term on segment $s$. %
			\item For any $n\neq n'$ and any $s\in y_n, t\in y_{n'}$, $\varepsilon_{n,s}$ and $\varepsilon_{n',t}$ are independent.   
		\end{enumerate}
	\end{assumption}

    The first and second assumptions are directly motivated by the discussions above. Note that we do not impose any distributional assumptions other than specifying the means and covariances of the travel times $T_{n,s}$.
	The third assumption says that conditional on all the segment-level and trip-level effects, the travel times on different trips are independent. This is a natural assumption, since much of the observed correlation across trips is due to time of week and other covariates. Conditional on those relevant covariates, it is much more reasonable to assume independence. Empirical evidence also shows that intra-trip correlation is much stronger than inter-trip correlation within similar time of week (see Figure 5 in \citealt{woodard2017predicting}).

	\subsection{Travel Time Estimators}
	
	For a new $(N+1)^{\textrm{th}}$ trip $\mathcal{T}_{N+1} = \{y_{N+1}, \{T_{N+1,s}\}_{s\in y_{N+1}}\}$ with segment-level feature sets $\{V_s\}_{s\in y_{N+1}}$ and route-level feature set $W_{N+1}$, the goal is to come up with an estimator $\hat{\Theta}_{\mathcal{T}_{N+1}}$, a function of the $N$ historical trips, $\{\mathcal{T}_n\}_{n\in[N]}$, for the total travel time $\sum_{s\in y_{N+1}} T_{N+1,s}$ that minimizes the following \emph{predictive  mean squared error} where the expectations are taken over $\{T_{n,s}\}_{n\in [N+1], s\in y_n}$ conditional on $h$ and on $\{\theta_s, V_s, W_n\}_{s\in y_n, n\in[N+1]}$. We drop the explicit conditioning in the following expectations for notation brevity. 
	\be
	&\bb E\Bigg[\bigg( \hat{\Theta}_{\mathcal{T}_{N+1}} - \sum_{s\in y_{N+1}} T_{N+1,s}\bigg)^2 %
	\Bigg] \\ 
	=&\bb E\Bigg[\bigg( \hat{\Theta}_{\mathcal{T}_{N+1}} - \sum_{s\in y_{N+1}}\Big(\theta_s + h(V_s, W_{N+1})\Big) + \sum_{s\in y_{N+1}}\Big(\theta_s + h(V_s, W_{N+1})\Big) -  \sum_{s\in y_{N+1}} T_{N+1,s}\bigg)^2 \Bigg] \\ 
	=&\bb E\Bigg[\bigg( \hat{\Theta}_{\mathcal{T}_{N+1}} - \sum_{s\in y_{N+1}}\Big(\theta_s + h(V_s, W_{N+1})\Big)\bigg)^2\Bigg] + \bb E\Bigg[ \bigg(\sum_{s\in y_{N+1}}\Big(\theta_s + h(V_s, W_{N+1})\Big) -  \sum_{s\in y_{N+1}} T_{N+1,s}\bigg)^2 \Bigg].
	\ee
	The last equality holds because $\hat{\Theta}_{\mathcal{T}_{N+1}}$ (a function of $\{T_{n,s}\}_{s\in y_n, n\in[N]}$) and $\sum_{s\in y_{N+1}} T_{N+1,s}$ are independent conditional on $\{\theta_s, V_s, W_n\}_{s\in y_n, n\in[N+1]}$ by the third part in Assumption~\ref{assump:1}, and moreover $\bb E\big[\sum_{s\in y_{N+1}} T_{N+1,s}\big] = \sum_{s\in y_{N+1}}\big(\theta_s + h(V_s, W_{N+1})\big)$. Now notice that the second term in the last equality does not depend on $\hat{\Theta}_{\mathcal{T}_{N+1}}$. 
This implies that the estimator $\hat{\Theta}_{\mathcal{T}_{N+1}}$ that minimizes the predictive error is the same one that minimizes the squared error for 
estimating the mean $\sum_{s\in y_{N+1}}\left(\theta_s + h(V_s, W_{N+1})\right)$.

Segment-based approaches typically directly estimate segment-level embeddings $\theta_s$ as well as a mean function $g(\cdot)$, and then they sum up the estimated travel times across the segments of the route. %
For example, the production ETA model in Google Maps at the time of the publication of \cite{rss,derrow2021eta} was based on a linear regression model for $g(\cdot)$ with features $V_s$ that include length, road class, and real-time and historical average travel speed; it also included learnable embedding vectors $\theta_s$ for each road segment to capture idiosyncratic effects. For tractability, we analyze a class of segment-based models that fit into the additive framework where $T_{n,s} = \theta_s + h(V_s, W_n) + \varepsilon_{n,s}$ for a scalar $\theta_s$.

Route-based methods, unlike segment-based methods, fit a model for the whole trip travel time $\sum_{s\in  y_n} T_{n,s}$. Typically they include some embeddings at the level of the origin and destination, or at the level of origin-destination pair.  They also typically use trip-level features $W_n$, as well as route-level features $\tilde{V}_n$ that are created by aggregating over the segments of the route, such as route length.  An example is a deep neural network used by Uber \citep{hu2022deepreta}, which includes an embedding $\Theta_{o,d}$ for the origin-destination pair, as well as trip-level features $W_n$ that include time of week, and route features $\tilde{V}_n$ that include route length and aggregated inputs related to real-time traffic conditions. In our simplified setting, a route-based method corresponds to fitting a model $\Theta_{y_n} + f(\tilde{V}_n, W_n)$, where $\Theta_{y_n}$ is an embedding that approximates $\sum_{s\in y_{n}}\theta_s$ and $f(\tilde{V}_n, W_n)$ is a function that approximates $ \sum_{s\in y_{n}} h(V_s,W_n)$ by using a feature transformation vector $\tilde{V}_n = \phi(\{V_s\}_{s\in y_n})$. For example, if $V_s$ is the length of segment $s$, $\tilde{V}_n$ can be the total length of the route $y_n$ (see, e.g., the linear regression model based on trip distance on page 8 of \citealt{wang2016simple}). 

In most of the travel time models described in the literature, the parameters in the function $h(\cdot)$ or $f(\cdot)$ don't need to scale with the number of road segments of the network (because the set of features, which include things like segment/route length and weather conditions, is fixed). The number of parameters $\{\theta_s\}_{s\in\mathcal{S}}$, by contrast, scales proportionally with the number of road segments. Since typical road networks, for example for large metropolitan regions, have hundreds of thousands or even millions of road segments, the number of parameters in the set $\{\theta_s\}_{s\in\mathcal{S}}$ tends to dominate the parameter size needed to robustly model $h(\cdot)$ or $f(\cdot)$. For example, the linear regression production model described in \cite{rss,derrow2021eta} estimates the parameters of a regression model with fixed input and output dimension. This model also includes road segment embeddings, the number of which scales proportionally with the size of the road graph.  

As a result, most of the error in the estimation of $\left(\theta_s + h(V_s, W_n)\right)$ in a segment-based method typically comes from the error in the estimation of $\theta_s$. A similar effect occurs in route-based methods: so long as $\phi(\cdot)$ is chosen in such a way that $f(\tilde{V}_n, W_n)$ is a good approximation to $ \sum_{s\in y_{n}} h(V_s,W_n)$, typically it is much easier to estimate $f(\cdot)$ than to estimate $\Theta_y$. If $\phi(\cdot)$ is chosen in such a way that $f(\tilde{V}_n, W_n)$ is {\it not} a good approximation to $ \sum_{s\in y_{n}} h(V_s,W_n)$, then the accuracy of the route-based method is degraded. We assume that this is not the case, which gives the benefit of the doubt to the route-based method. %

Based on the discussions above, we will thus focus on estimating the accumulation of segment random effects on a route, $\sum_{s\in y_n} \theta_s$. We denote by $T'_{n,s}= T_{n,s} - h(V_s, W_n)$ the adjusted observed segment travel time with mean $\theta_s$. %

\begin{assumption}
   We assume that the function $h(\cdot)$ is known (approximating a situation where $h(\cdot)$ is much easier to estimate than $\{\theta_s\}_{s\in\mathcal{S}}$).
\end{assumption}

	To compare the predictive accuracy of different estimators, we introduce the  \emph{integrated risk}, a Bayesian statistical concept capturing the accuracy of the travel time prediction by integrating the risk (in our case, the predictive mean squared error) over the prior distribution of the unknown parameters. This is also known as an ``average-case''  analysis of accuracy (versus for example a worst-case analysis). As we shall see later, focusing on such an average-case analysis also helps us to reach more general conclusions regarding the comparisons of these estimators. In particular, we impose the following assumption on the prior distribution.

    \begin{assumption}
    \label{assump:2}
       We assume that $\{\theta_s\}_{s\in\mathcal{S}}$ are drawn i.i.d.\ from a population distribution with mean $\mu$ and variance $\tau^2$.%
    \end{assumption}

The choice of i.i.d.\ population distribution is for notation brevity, and our results can be generalized to non-i.i.d.\ population distribution to capture, for example, congestion patterns across road networks. The integrated risk of estimator $\hat{\Theta}_y$ for a new route $y$ given historical routes $y_{[N]}$, which we call $R(\hat \Theta_y \,\vert\,  y_{[N]})$, is defined to be the expectation of the squared difference between the true total mean travel time $\Theta_y:=\sum_{s\in y}\theta_s$ and the estimated total mean travel time $\hat{\Theta}_y$.  This expectation is taken with respect to (i) the observed adjusted segment travel times $\{T'_{n,s}\}_{n\in[N], s\in y_n}$ and (ii) the population distribution over the parameters $\{\theta_s\}_{s\in y}$, conditional on the historical route observations $y_{[N]}$: %
	\be
	\label{eq:integrated_risk}
	R\left(\hat \Theta_y \,\middle\vert\,  y_{[N]} \right) := \mathbb{E}\left[\left(\hat \Theta_y-\Theta_y\right)^2 \,\middle\vert\,  y_{[N]}\right].%
	\ee

 We now illustrate how travel time on a route can be predicted using different examples of estimators.

	\begin{example}[\textsc{Travel Time Estimators}]
		\label{ex:estimators}
		Consider the following $3\times 3$ grid in \Cref{fig:grid} where there are six historical trips $\{\mathcal{T}_n\}_{n\in\{1,\cdots,6\}}$ whose routes are displayed in the figure. Let segments $s_1$ and $s_2$ denote the ones that traverse $(1,0)\rightarrow(1,1)$ and $(1,1)\rightarrow (1,2)$, respectively. Suppose that we want to predict the travel time on route $y$ which traverses through $(1,0)\rightarrow(1,1)\rightarrow(1,2)$ ($s_1$ and $s_2$). Among the space of all possible estimators which are functions of historical data, the following are a few simple estimators that capture characteristics of popular estimators used in the literature and practice. %
		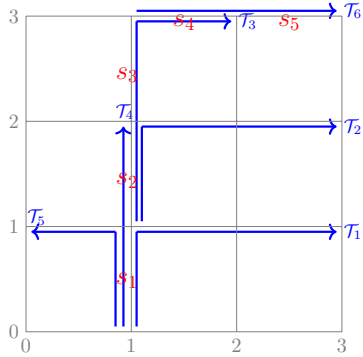
\begin{figure}[h]
			\centering
			\scalebox{0.7}{\begin{tikzpicture}
\foreach \i[evaluate=\i as \evali using int(\i/2)] in {0,2,4,6} {
    \draw [very thin,gray] node [below] at (\i,0) {$\evali$}; %
}
\foreach \i[evaluate=\i as \evali using int(\i/2)] in {0,2,4,6} {
    \draw [very thin,gray] node [left] at (0,\i) {$\evali$}; %
}

\draw[very thick, blue]   (2.1,0.1) -- (2.1,1.9);
\draw[very thick, ->, blue]   (2.1,1.9) -- (5.9,1.9);
\draw node [right, blue] at (5.9,1.9) {\small{$\mathcal{T}_1$}};

\draw[very thick, blue]   (2.2,2.1) -- (2.2,3.9);
\draw[very thick, ->, blue]   (2.2,3.9) -- (5.9,3.9);

\draw node [right, blue] at (5.9,3.9) {\small{$\mathcal{T}_2$}};

\draw[very thick, blue]   (2.1,2.1) -- (2.1,5.9);
\draw[very thick, ->, blue]   (2.1,5.9) -- (3.9,5.9);

\draw node [right, blue] at (3.9,5.9) {\small{$\mathcal{T}_3$}};

\draw node [right, blue] at (5.9,6.1) {\small{$\mathcal{T}_6$}};

\draw[very thick, ->, blue]   (1.85,0.1) -- (1.85,3.9);

\draw[very thick, -, blue]   (1.7,0.1) -- (1.7,1.9);

\draw[very thick, ->, blue]   (1.7,1.9) -- (0.1,1.9);

\draw[very thick, ->, blue]   (2.1,6.1) -- (5.9,6.1);

\draw node [right, blue, above] at (1.88,3.9) {\small{$\mathcal{T}_4$}};

\draw node [right, red, above] at (1.92,0.7) {\Large{$s_1$}};

\draw node [right, red, above] at (1.92,2.6) {\Large{$s_2$}};

\draw node [right, red, above] at (1.92,4.6) {\Large{$s_3$}};

\draw node [right, red, above] at (3.0,5.6) {\Large{$s_4$}};

\draw node [right, red, above] at (5.0,5.6) {\Large{$s_5$}};

\draw node [right, blue, above] at (0.2,1.9) {\small{$\mathcal{T}_5$}};

\draw[step=2cm,gray,very thin] (0,0) grid (6,6);
\end{tikzpicture}}
			\caption{A $3\times 3$ grid example.}
			\label{fig:grid}
		\end{figure} 
		\smallskip
		\begin{enumerate}
			\item \emph{Segment-based estimator}: We estimate travel times on each segment using individual segment traversal data and then aggregate them across the route. In this example, $s_1$ has three traversals (from trips $\mathcal{T}_1, \mathcal{T}_4$ and $\mathcal{T}_5$), $s_2$ has three traversals (from trips $\mathcal{T}_2, \mathcal{T}_3$ and $\mathcal{T}_4$). If we use a simple estimator by taking the average of historical travel times, 
			\be
			\hat{\Theta}_y = \frac{T^\prime_{1,s_1} + T^\prime_{4, s_1} + T^\prime_{5, s_1}}{3} + \frac{T^\prime_{2,s_2} + T^\prime_{3, s_2} + T^\prime_{4, s_2}}{3},
			\ee
			where each term is the estimation of the segment travel times on $s_1$ and $s_2$ respectively.
			
			\smallskip
			
			\item \emph{Generalized segment-based estimator}: Combining individual segments into \emph{super-segments}, we can generalize the previously defined segment-based estimator. For example, we can define $s_1$ and $s_2$ together as one super-segment, and this super-segment has one traversal (from trips $\mathcal{T}_4$). This yields
			\be
			\hat{\Theta}_y = T^\prime_{4,s_1} +T^\prime_{4,s_2}.
			\ee
			
			\smallskip
			
			\item \emph{Route-based estimator}: Instead of aggregating over segments or super-segments, we directly use total travel times on historical routes that are similar to route $y$ for estimation. %
			For example, we can average over travel times of routes that share similar origin and destination as $y$ --- both the origin and the destination of the route are at most one segment away from those of $y$. This includes trips $\mathcal{T}_4$ and $\mathcal{T}_5$. This gives,
			\be
			\hat{\Theta}_y = \frac{\sum_{s\in y_4}T^\prime_{4,s} + \sum_{s\in y_5}T^\prime_{5,s}}{2}.
			\ee
			
			Note that for route-based estimators, it is possible to use historical traversal data on segments that are \emph{not} included in $y$ to predict the total travel time of $y$. %
			
		\end{enumerate}
	\end{example}
	
	Variations of the abovementioned estimators have been practiced extensively, but there is a lack of formal analysis to compare their accuracies. Such comparison is not straightforward because these estimators differ in sample sizes and the way historical data is used, and these differences can have non-trivial effects on accuracies, as we will illustrate later.
	
	\subsubsection{The Optimal Estimator under Normality Assumptions}
	
	Before comparing different estimators, it is always helpful to understand the perfect benchmark first --- the optimal estimator $\hat \Theta^\ast_y$ that has the minimal integrated risk $R(\hat \Theta^\ast_y \,\vert\,  y_{[N]})$. %
    Characterizing the optimal estimator under general distributions does not always admit closed or tractable forms \citep{gelman2013bayesian}. Here we give the form of the optimal estimator under an additional assumption that the error terms $\varepsilon_{n,s}$ and $\theta_s$ are jointly Gaussian distributed. %
    Defining $M=\sum_{n=1}^N |y_n|$, let $Z\in\mathbb R_{\ge0}^{M\times 1}$ be the vector consisting of the concatenated travel times across all the trips so that $Z_i$ is the travel time for a single segment on a single trip. Similarly, let $\mathcal{E}\in\mathbb R^{M\times 1}$ be the corresponding vector of error terms.
    Let $u_i=s$ if the $i^{\textrm{th}}$ entry in $Z$ and $\mathcal{E}$ is a travel time and its corresponding error term on segment $s$. Let $w_i = n$ if the $i^{\textrm{th}}$ entry in $Z$ and $\mathcal{E}$ is a travel time and its corresponding error term from trip $n$.
    Thus, $Z_i = T^\prime_{w_i, u_i},\:\forall i\in\{1,\cdots, M\}$. 
 Let $U\in\{0,1\}^{M\times|\mathcal{S}|}$ be a matrix with entries $U_{i,s}=\1\{u_i=s\}$.
 This yields that $Z = U \theta + \mathcal{E}$.
 
 Define $\Phi\in\mathbb R^{M\times M}$ to be a matrix whose entries are
	\be
	\Phi_{i,j} = \begin{cases} \sigma_{u_i,u_j}, & \textrm{if } w_i = w_j, \\ 0, & \text{otherwise}. \end{cases}
	\ee
	Notice that $\Phi$ is a block diagonal matrix with $N$ total blocks, where the $n^{\textrm{th}}$ diagonal block has dimension $|y_n| \times |y_n|$. Let $e\in 1^{|\mathcal{S}|\times 1}$ be a $|\mathcal{S}|$-dimensional all-ones vector and $e_y\in\{0,1\}^{|\mathcal{S}|\times 1}$ such that $e_{y,s}=1$ if $s\in y$ and $0$ otherwise. To simplify the notation, put $E_y = e_y e_y^\intercal$ and $Q = U^\intercal \Phi^{-1}U + \diag((1/\tau^2) e)$.

	\begin{theorem}[\textsc{Optimal Estimator}]
		\label{thm:optimal_estimator}
		In addition to Assumptions \ref{assump:1} and \ref{assump:2}, assume that $(\mathcal{E}, \theta)$ are jointly Gaussian distributed, the following estimator $\hat{\Theta}^\ast_y$ of travel time on route $y$ minimizes the integrated risk \eqref{eq:integrated_risk} among all possible estimators,
		\be
		\label{eq:opt_est}
		\hat{\Theta}^\ast_y = e^\intercal_yQ^{-1}\left(U^\intercal\Phi^{-1}Z+ (\mu/\tau^2)e\right).
		\ee
		Its integrated risk $R(\hat \Theta_y^\ast \mid  y_{[N]})$ based on squared error is
		\be
		\mathrm{tr}\left(\Phi^{-1}UQ^{-1}E_y Q^{-1}U^\intercal\right) + \mathrm{tr}\left(\diag\left(\tau^2e\right)\left(E_y + U^\intercal\Phi^{-1}UQ^{-1}E_y Q^{-1} U^\intercal\Phi^{-1}U - 2 E_y Q^{-1}U^\intercal\Phi^{-1}U\right)\right).
		\ee
		The first term is the expected variance, and the second term is the expected squared bias. 
	\end{theorem}
	
An interesting observation is that the minimum integrated risk does \emph{not} depend on the population mean $\mu$. The proof  relies on deriving the conditional mean $\mathbb E[\sum_{s\in y}\theta_s |\: \{\mathcal{T}_n\}_{n\in[N]}]$ of travel time on route $y$ given historical trip data. It then uses the fact that the estimator minimizing the integrated risk based on squared error is the posterior mean (see, e.g., \citealt{berger2013statistical}). As a sanity check, when there is no historical traversal at all, $\hat{\Theta}^\ast_y = |y|\mu$ which simply uses the population mean.
	
	When segment travel times are independent across the road network, the form of the optimal estimator can be greatly simplified. Let $N_s := |\{y_n: s\in y_n, n\in[N]\}|$ denote the sample size of traversals on segment $s$ in the historical data.
	
	\begin{corollary}[\textsc{Optimal Estimator under Independent Segment Travel Times}]
		\label{cor:independent_simple_form}
	In addition to Assumptions \ref{assump:1} and \ref{assump:2}, assume that $(\mathcal{E}, \theta)$ are jointly Gaussian distributed,	when $\sigma_{s,t}=0,\:\forall s\neq t\in\mathcal{S}$, the optimal estimator takes the following form,
		\be
		\hat{\Theta}_y^\ast = \sum_{s\in y}\left( \frac{\sigma_s^2}{N_s\tau^2 + \sigma_s^2}\cdot\mu + \frac{N_s\tau^2}{N_s\tau^2 + \sigma_s^2}\cdot\frac{\sum_{n: s\in y_n} T_{n,s}^\prime}{N_s}\right).
		\ee
	\end{corollary}
	
	This result says that in the independent case, the optimal estimator $\hat{\Theta}_y^\ast$ takes the form of a simple segment-based estimator, where the travel times on each segment are estimated by a weighted average of the sample mean $\sum_{n: s\in y_n} T_{n,s}^\prime/N_s$ and the population mean $\mu$. The weight depends on the sample size $N_s$ and variance parameters $\sigma^2_s$ and $\tau^2$. Intuitively, when the sample size $N_s$ or population variance $\tau^2$ is high, the sample mean gains more weight --- the optimal estimator weighs the historical observations more; on the other hand, when the variance of the segment travel time $\sigma_s^2$ is high, the optimal estimator relies more on the population information.
	
	When segment travel times are correlated over the road network, the form of the optimal estimator cannot be decomposed by segments, and computing the optimal estimator for a new route $y$ uses \emph{all} historical segment traversal data over the \emph{entire} road network, regardless of being part of route $y$ or not. Below is an example.
	
	\begin{example}
		\label{eg:opt_estimator}
		Using the same setup in \Cref{ex:estimators}, we now compute the optimal estimator for travel time on route $y$, traversing through $(1,0)\rightarrow(1,1)\rightarrow(1,2)$. Take the example where 
        $\{\theta_s\}_{s\in\mathcal{S}}$ are drawn i.i.d.\ from a normal distribution with mean $\mu=1$ and variance $\tau^2=0.2$ and the adjusted segment travel times are drawn from a multivariate normal distribution with mean $\{\theta_s\}_{s\in\mathcal{S}}$ and covariance matrix $\Sigma=[\sigma_{s,t}]_{s,t\in\mathcal{S}}=e^{-\mathscr{L}}$.
  Matrix $\mathscr{L}$ is the normalized graph Laplacian of an undirected graph where each node in the graph is a segment in the $3\times3$ grid network in Figure \ref{fig:grid}, and an edge is created if two segments are directly connected. Precisely, let $\mathscr{L}=D^{-1/2}LD^{1/2}$ where $D$ is the diagonal matrix of segment degrees, and $L=D-A$ is the graph Laplacian where $A$ is the adjacency matrix. The covariance matrix is the matrix exponential $e^{-\mathscr{L}}$ which is often called the diffusion kernel. It models spatial decay of correlation among segment travel times. For this example, the optimal estimator $\hat{\Theta}_y^\ast$ based on \eqref{eq:opt_est} takes the form
		\be
		\hat{\Theta}_y^\ast =& 0.211\cdot T^\prime_{1, (1,0)\rightarrow(1,1)}  - 0.040\cdot T^\prime_{1, (1,1)\rightarrow(2,1)} + 0.002\cdot T^\prime_{1, (2,1)\rightarrow(3,1)} \\[1mm]
		&+0.207\cdot T^\prime_{2, (1,1)\rightarrow(1,2)} - 0.003 \cdot T^\prime_{2, (1,2)\rightarrow (2,2)} + 0.002\cdot T^\prime_{2, (2,2)\rightarrow (3,2)} \\[1mm]
		&+0.210\cdot T^\prime_{3, (1,1)\rightarrow(1,2)} - 0.040\cdot T^\prime_{3, (1,2)\rightarrow(1,3)} + 0.002\cdot T^\prime_{3, (1,3)\rightarrow(2,3)} \\[1mm]
		&+0.157\cdot T^\prime_{4, (1,0)\rightarrow(1,1)} + 0.156\cdot T^\prime_{4, (1,1)\rightarrow(1,2)} \\[1mm]
		&+0.201\cdot T^\prime_{5,(1,0)\rightarrow(1,1)} - 0.010\cdot T^\prime_{5,(1,1)\rightarrow (0,1)}\\[1mm]
		& -0.001\cdot T^\prime_{6, (1,3)\rightarrow(2,3)}  + 0.000\cdot T^\prime_{6,(2,3)\rightarrow(3,3)}+0.978.
		\ee
		The optimal integrated risk $R(\hat \Theta_y^\ast \mid  y_{[N]})$ based on squared error is then $0.172$ (which is the sum of the expected variance $0.097$ and the expected squared bias $0.075$). %
	\end{example}
	
	\subsubsection{Segment-Based and Route-Based Estimators}
	
	As we mentioned above, the optimal estimator $\hat{\Theta}_y^\ast$ uses historical traversal data on segments that are not part of route $y$. Moreover, the traversal data on different segments cannot be easily aggregated to obtain the optimal estimator, since the weight of each historical observation is non-trivially determined by the correlation structure. This makes the optimal estimator intractable to implement for large road networks, which typically consist of millions of road segments. Moreover, as we discussed in the previous subsection, the form of the optimal estimator depends on distributional assumptions, which can be hard to obtain in practice.
	
	On the other hand, although being sub-optimal in general, the estimators mentioned in \Cref{ex:estimators} are simple in nature and only use historical traversal data that is directly relevant to route $y$; approaches like these are generally regarded as efficient and scalable methods and are practiced widely in mapping services. In addition, as we will show later, optimal estimators within these classes can be developed without distributional assumptions. It is thus of both practical and theoretical interests to compare their relative performance and benchmark them against the best estimator possible. To do so, generalizing the discussions in \Cref{ex:estimators} and \Cref{cor:independent_simple_form}, we first introduce a formal definition of the segment-based estimators.

	\begin{definition}[\textsc{Segment-Based Estimator}]
		\label{dfn:seg}
		A segment-based estimator $\Thetaseg$ takes the form 
		\begin{align*}
			\Thetaseg &:= \sum_{s \in y} \hat \theta_s,\qquad
			\hat \theta_s := (1-\phi_s(N_s))\mu + \phi_s(N_s)\frac{\sum_{n: s\in y_n}T'_{n,s}}{N_s},
		\end{align*}
		for some $\phi_s: \bb Z_{\ge0}\mapsto\mathbb R$ such that %
		$\phi_s(0)=0$ for all $s\in y$, and define $0/0=0$.
	\end{definition}

	In other words, $\hat{\Theta}_y^{(\textrm{seg})}$ is the summation of segment-level estimators $\hat{\theta}_s$ that are constructed using a weighted average of the sample mean and the mean of the population distribution, where the weights $\{\phi_s(N_s)\}_{s\in y}$ are sample size dependent. %
	One would typically expect the weights $\{\phi_s(N_s)\}_{s\in y}\in[0,1]$ to converge to $1$ as the sample size $N_s$ grows to infinity. Such behavior ensures consistency of the estimator, i.e., $\hat{\theta}_s$ converges to $\theta_s$ almost surely as $N_s$ tends to infinity. %
	As an example, the estimator in \Cref{cor:independent_simple_form} takes the form of $\phi_s(N_s) = \tau^2 N_s/(\tau^2 N_s + \sigma^2_{s})$. 
 Although the segment-based estimator in Definition~\ref{dfn:seg}
 appears quite simple, variants on this estimator are used widely as a foundational component of commercial ETA prediction systems. In particular, either the time or speed is averaged for both (a) traversals of the road segment in the last few minutes; (b) traversals of the road segment from the same time of week in previous weeks. These are combined using either simple fallback logic (if sufficient recent traversals are available, then use their average, otherwise use the historical average) or a machine learning model trained with those features as inputs (see Section~3.1.3 and 4.1.2 of \citealt{derrow2021eta}). 
	
	Generalizing \Cref{dfn:seg} by combining individual segments into super-segments, we have the following definition of the generalized segment-based estimators. A super-segment $S$ is a set of (usually connecting) distinct segments. Let $\mathscr{S}_y$ denote a set of super-segments that constitutes route $y$, i.e., $\cup_{S\in \mathscr{S}_y} S = y$ and $S\cap T =\emptyset$ for any $S\neq T \in\mathscr{S}_y$. With a slight abuse of notation, let $N_S := |\{y_n: S\subset y_n, n\in[N]\}|$ be the sample size of traversals on super-segment $S$ in the historical data. 
	
	\begin{definition}[\textsc{Generalized Segment-Based Estimator}]
		\label{dfn:gseg}
		A generalized segment-based estimator $\ThetasegG$ takes the form 
		\begin{align*}
			\ThetasegG &:= \sum_{S \in \mathscr{S}_y} \hat \theta_S, \qquad\hat \theta_S := (1-\phi_S(N_S))|S|\mu + \phi_S(N_S)\frac{\sum_{n: S\subset y_n}\sum_{s\in S} T'_{n,s}}{N_S},
		\end{align*}
		for some $\mathscr{S}_y$, a set of super-segments constituting route $y$ and some $\phi_S: \bb Z_{\ge0}\mapsto\mathbb R$ such that %
		$\phi_S(0)=0$ for all $S\in \mathscr{S}_y$, and define $0/0=0$.
	\end{definition}
	
	Similarly to the segment-based estimator, $\{\phi_S(N_S)\}_{S\in \mathscr{S}_y}\in[0,1]$ are expected to converge to $1$ as the sample size $N_S$ approaches infinity. When $\mathscr{S}_y=\{\{s\}\}_{s\in y}$, the generalized segment-based estimator based on $\mathscr{S}_y$ reduces to the segment-based estimator in \Cref{dfn:seg}. A generalized segment-based travel time estimation method is described in \cite{derrow2021eta} for Google Maps.

\smallskip

     We finally define a family of \emph{route-based estimators} which uses route-level traversal data to estimate the travel time on a new route $y$. Let $\delta(y)\subset y_{[N]}$ denote a subset of historical routes which represents the neighborhood of route $y$. For example, $\delta(y)$ can be historical routes that share the same or similar origin and destination (but possibly with a different sequence of segments) as those of $y$ (see \citealt{wang2016simple}). These neighboring routes are representative observations to estimate travel time on route $y$. Let $M_{\delta(y)} = \sum_{n=1}^N \1\{y_n \in \delta(y)\}$ be the sample size of route $y$'s neighborhood, and $|y|$ be the number of segments traversed on route $y$.

	\begin{definition}[\textsc{Route-Based Estimator}]
		\label{dfn:route}
		A route-based estimator $\Thetatrp$ takes the form
		\be
		\Thetatrp&:=(1-\phi_{\delta(y)}(M_{\delta(y)}))|y|\mu + \phi_{\delta(y)}(M_{\delta(y)})\frac{\sum_{n : y_n \in \delta(y)} \sum_{s\in y_n} T'_{n,s}}{M_{\delta(y)}},
		\ee
		for some neighborhood of route $y$, $\delta(y)$, and some $\phi_{\delta(y)}: \bb Z_{\ge0}\mapsto\mathbb R$ such that %
		$\phi_y(0)=0$. Define $0/0=0$.
	\end{definition}
	
	In words, $\Thetatrp$ estimates travel time on route $y$ by a weighted average of the sample mean of all observed travel times of the historical routes in $\delta(y)$, and the population mean of travel time on route $y$, where the weight $\phi_{\delta(y)}(M_{\delta(y)})$ is sample size dependent. Such a nearest-neighbor route-based estimator is used in \cite{wang2016simple}, for example.%

    We can derive the integrated risks $R\big(\ThetasegG \:\big|\: y_{[N]}\big)$ and $R\big(\Thetatrp \:\big|\:  y_{[N]}\big)$ conditional on historical routes $y_{[N]}$. With a slight abuse of notation, let $N_s^{\delta(y)} :=  \left|\left\{y_n\in\delta(y): s\in y_n\right\}\right|$ be the number of traversals on segment $s$ from the historical routes in neighborhood $\delta(y)$. Note that $N_s^{\delta(y)}$ is defined for $s\notin y$ as well since routes in $\delta(y)$ can traverse segments that are not in $y$. Similarly, for a set of distinct segments $S$, let $N_S^{\delta(y)}:=\left|\left\{y_n\in\delta(y): S\subset y_n\right\}\right|$ denote the number of traversals on super-segments $S$ from the historical routes in neighborhood $\delta(y)$.
	By definition, we have $N_S^{\delta(y)} \le N_s^{\delta(y)} \le M_{\delta(y)}$, for all super-segments $S$ such that $s\in S$. Let $\mathcal{S}_{\delta(y)} = \cup_{y^\prime\in \delta(y)}y^\prime$ be the set of all segments traversed by the routes in the neighborhood $\delta(y)$. Finally, put $\bar{y}_{\delta(y)} = \sum_{y_n\in\delta(y)} |y_n| /M_{\delta(y)}$ to be the average number of segments traversed per route in the neighborhood $\delta(y)$. We now give the integrated risks $R\big(\ThetasegG \:\big|\: y_{[N]}\big)$ and $R\big(\Thetatrp \:\big|\:  y_{[N]}\big)$ conditional on historical routes $y_{[N]}$. The integrated risk of the segment-based method $R\big(\Thetaseg \:\big|\: y_{[N]}\big)$ is a special case of $R\big(\ThetasegG \:\big|\: y_{[N]}\big)$ with $\mathscr{S}_y=\{\{s\}\}_{s\in y}$. It is worth noting that these expressions are derived under \emph{no} distributional assumption.

	\begin{proposition}
		\label{prop:diff}
		Under Assumptions \ref{assump:1} and \ref{assump:2}, for any route $y$, the integrated risks, conditional on the historical routes $y_{[N]}$, are %
		\be
		\nonumber
		R\left(\ThetasegG\,\middle\vert\,  y_{[N]}\right)=&%
		\sum_{S,T\in \mathscr{S}_y}\frac{N_{S\cup T}}{N_S N_T}\phi_S(N_S)\phi_{T}(N_{T})\bigg(\sum_{s\in S, t\in T} \sigma_{s,t}\bigg) \\ 		\label{eq:seg_risk}
		&\quad+ \sum_{S\in\mathscr{S}_y}(1-\phi_S(N_S))^2|S|\tau^2, \\  
		R\left(\Thetatrp\,\middle\vert\,  y_{[N]}\right)=&%
		\left(\frac{\phi_{\delta(y)}(M_{\delta(y)})}{M_{\delta(y)}}\right)^2 \bigg(\sum_{s, t\in \mathcal{S}_{\delta(y)}} N_{s\cup t}^{\delta(y)} \sigma_{s,t}\bigg) + \left(\phi_{\delta(y)}(M_{\delta(y)})(\bar{y}_{\delta(y)} - |y|)\mu\right)^2 \\ \label{eq:trp_risk}
		+\sum_{s\in \mathcal{S}_{\delta(y)}\setminus y}&\left(\phi_{\delta(y)}(M_{\delta(y)}) \frac{N_s^{\delta(y)}}{M_{\delta(y)}}\right)^2\tau^2 + \sum_{s\in y}\left(1 - \phi_{\delta(y)}(M_{\delta(y)}) \frac{N_s^{\delta(y)}}{M_{\delta(y)}} \right)^2\tau^2.  %
		\ee
		
	\end{proposition}
	
	The first and second terms in \eqref{eq:seg_risk} correspond to the expected variance and squared bias of the generalized segment-based estimator, respectively. The expected squared bias comes from the shrinkage towards the prior mean $|S|\mu$ (can be different from the true means $\sum_{s\in S}\theta_s$), which goes down as $\phi_S(N_S)$ increases. The choice of $\phi_S(N_S)$ controls the bias-variance trade-off. Higher $\phi_S(N_S)$ (less shrinkage) leads to lower bias but introduces more variance as the estimator puts more weight on the information provided by the historical observations. Similarly, the first term in \eqref{eq:trp_risk} represents the expected variance of the route-based estimator, and the sum of the second, third, and fourth terms collectively represents the expected squared bias of the route-based estimator. Specifically, the second term represents the squared bias introduced by including routes in $\delta(y)$ that have more or fewer road segments than $y$. The third term accounts for the squared bias of using traversal data on segments that are not included in $y$. Finally, the fourth term calculates the amount of squared bias induced by the shrinkage towards the prior mean $|y|\mu$.  In addition to $\phi_{\delta(y)}(M_{\delta(y)})$, the choice of neighborhood $\delta(y)$ also plays a significant role here. If $\delta(y)$ is chosen to include only routes that are very similar to $y$, in terms of the number of segments and the set of segments they traverse, $\bar{y}_{\delta(y)}$ will be close to $|y|$, and $N_s^{\delta(y)}/M_{\delta(y)}$ will be close to $1$ for segments $s\in y$ and close to $0$ for segments $s\notin y$. This will lead to a lower bias, but potentially a higher variance as the number of samples $M_{\delta(y)}$ will be smaller.

	Based on the formulae of the integrated risks, we define the optimal generalized segment-based estimator $\ThetasegoptG$ given $\mathscr{S}_y$ as the one that minimizes the integrated risk \eqref{eq:seg_risk} by picking the best forms of $\{\phi_S(N_S)\}_{S\in \mathscr{S}_y}$. Also, given a neighborhood $\delta(y)$, the optimal route-based estimator $\Thetatrpopt$ is defined to be the one that minimizes the integrated risk \eqref{eq:trp_risk} by picking the best form of $\phi_{\delta(y)}(M_{\delta(y)})$. The next result characterizes the forms of $\ThetasegoptG$ and $\Thetatrpopt$. Note that $\Thetasegopt$ is a special case of $\ThetasegG$ with $\mathscr{S}_y=\{\{s\}\}_{s\in y}$. Again, in contrast to \Cref{thm:optimal_estimator}, these optimal estimators are characterized under no distributional assumption.

	\begin{proposition}
		\label{prop:optimal_formula}
		Under Assumptions \ref{assump:1} and \ref{assump:2}, given $\mathscr{S}_y$, the optimal generalized segment-based estimator $\ThetasegoptG$ takes the following form:
		\be
		\label{eq:opt_seg}
		\ThetasegoptG &:= \sum_{S \in \mathscr{S}_y} \hat \theta_S, \quad\hat \theta_S &:= (1-\phi^\ast_S(N_S))|S|\mu + \phi^\ast_S(N_S)\frac{\sum_{n: S\subset y_n}\sum_{s\in S} T'_{n,s}}{N_S},
		\ee
		where $\{\phi_S^\ast(N_S)\}_{S\in \mathscr{S}_y}$ uniquely solves a system of linear equations 
		\be
		\sum_{T\in \mathscr{S}_y} (N_{S\cup T}/ (N_SN_T))&\phi^\ast_T(N_T)\left(\sum_{s\in S,t\in T} \sigma_{s,t}\right) \\ 		\label{eq:lin_sys}
		&+ (\phi_S^\ast(N_S) - 1)|S|\tau^2 = 0,~\forall S\in \mathscr{S}_y.
		\ee
		
		On the other hand, the optimal route-based estimator $\Thetatrpopt$ has the following form:
		\be
		\label{eq:opt}
		\Thetatrpopt:=(1-\phi_{\delta(y)}^\ast(M_{\delta(y)}&))|y|\mu + \phi^\ast_{\delta(y)}(M_{\delta(y)})\frac{\sum_{n : y_n = y} \sum_{s\in y} T'_{n,s}}{M_{\delta(y)}}, \\ %
		\phi_{\delta(y)}^\ast(M_{\delta(y)}) =\left(\sum_{s\in y} N_s^{\delta(y)} \right)\tau^2\Bigg/ \Bigg(&\sum_{s\in \mathcal{S}_{\delta(y)}} \frac{\left(N_s^{\delta(y)}\right)^2}{M_{\delta(y)}}\tau^2 + \frac{\sum_{n: y_n\in\delta(y)}\sum_{s,t\in y_n}\sigma_{s,t}}{M_{\delta(y)}} + M_{\delta(y)}\mu^2(\bar{y}_{\delta(y)} - |y|)^2\Bigg).
		\ee
	\end{proposition}

\smallskip
 
	The following example illustrates the forms of these estimators.
	
	\begin{example}
		\label{ex:opt_seg_route_estimators}
		Using the same setup in \Cref{ex:estimators} without the normality assumptions, we first compute the optimal segment-based estimator $\Thetasegopt$ for route $y$ traversing through $(1,0)\rightarrow(1,1)\rightarrow(1,2)$. Note that the covariance matrix $e^{-\mathscr{L}}$ has all non-negative entries. According to \Cref{prop:optimal_formula} of the supplement, %
		\be
		\Thetasegopt =& \left(0.560\cdot\frac{T^\prime_{1,s_1} + T^\prime_{4, s_1} + T^\prime_{5, s_1}}{3} + 0.440\cdot\mu\right) + \left(0.562\cdot\frac{T^\prime_{2,s_2} + T^\prime_{3, s_2} + T^\prime_{4, s_2}}{3} + 0.438\cdot\mu\right),
		\ee
		with integrated risk $0.176$ (expected variance $0.099$ and expected squared bias $0.077$).
		
		The optimal generalized segment-based estimator $\ThetasegoptG$ based on $\mathscr{S}_y = \{\{y\}\}$ is with $\phi_y^\ast(N_y) =N_y|y|\tau^2\big/\big(N_y|y|\tau^2 + \sum_{s,t\in y}\sigma_{s,t}\big)$. In this example, $N_y=1$. This gives
		\be
		\ThetasegoptG = 0.267\cdot\left(T^\prime_{4,s_1} +T^\prime_{4,s_2}\right) + 0.733\cdot (2\mu),
		\ee
		with integrated risk $0.293$ (expected variance $0.078$ and expected squared bias $0.215$).

		Finally, consider the optimal route-based estimator $\Thetatrpopt$ based on $\delta(y)$ which includes all historical routes whose origin and destination are at almost one segment away from that of $y$. This includes trips $\mathcal{T}_4$ and $\mathcal{T}_5$. %
		\be
		\Thetatrpopt = 0.372\cdot \frac{\left(\sum_{s\in y} T^\prime_{4,s} + \sum_{s\in y} T^\prime_{5,s} \right)}{2} + 0.628\cdot(2\mu),
		\ee
		with integrated risk $0.288$ (expected variance $0.070$ and expected squared bias $0.218$).
	\end{example}

	Our next result investigates the comparison of the integrated risks of these estimators, under a case where $\sigma_{s,t}\ge0$ for all $s\neq t\in \mathcal{S}$, i.e., there exists only non-negative covariances in the road network. We first show that the optimal segment-based estimator $\Thetasegopt$ is more accurate than a wide variety of optimal route-based estimators $\Thetatrpopt$. %
	
	\begin{theorem}
		\label{thm:opt_bayes_comparison}
		In addition to Assumptions \ref{assump:1} and \ref{assump:2}, suppose that $\sigma_{s,t}\ge 0$ for all $s,t \in \mathcal{S}$. Let $\Thetasegopt$ be the optimal segment-based estimator and let $\Thetatrpopt$ be the optimal route-based estimator with neighborhood $\delta(y)$ such that 
		\be
		\label{eq:nb_condition}
		N_{s\cup t}N_s^{\delta(y)} N_t^{\delta(y)} \le N_{s\cup t}^{\delta(y)}N_sN_t,\:\forall s,t\in y,
		\ee
		we have,
		\be
		R\left(\Thetasegopt \,\middle\vert\,  y_{[N]}\right)  \le R\left(\Thetatrpopt \,\middle\vert\,  y_{[N]}\right),
		\ee
		for any set of historical routes $y_{[N]}$.
	\end{theorem}
	Condition \eqref{eq:nb_condition} on the neighborhood of route-based estimators $\delta(y)$ is very mild --- by re-arranging the term, we have 
	\be
	N_{s\cup t}/(N_s N_t) \le N^{\delta(y)}_{s\cup t}/(N_s^{\delta(y)} N_t^{\delta(y)}),\:\forall s,t\in y.
	\ee
	One can effectively think of $N_{s\cup t}, N_s$, and $N_t$ as the sample sizes under a neighborhood that includes \emph{all} historical routes. In other words, condition \eqref{eq:nb_condition} intuitively requires the chosen neighborhood $\delta(y)$ to concentrate more around the predicting route compared to a neighborhood which includes all routes --- the ratio $N^{\delta(y)}_{s\cup t}/(N_s^{\delta(y)} N_t^{\delta(y)})$ gets larger as the neighborhood $\delta(y)$ concentrates around $y$. This is generally expected because as the neighborhood gets smaller, routes in $\delta(y)$ become more similar to $y$. It then becomes more likely that a route in $\delta(y)$ traversing over segment $s\in y$ also traverses over segment $t\in y$. Another way to appreciate this intuition is to look at the other extreme --- the smallest possible neighborhood $\delta(y) = \{y_n: y_n = y\}$ which only includes historical routes that are \emph{exactly the same} as route $y$. Under such a neighborhood, we have $N^{\delta(y)}_{s\cup t}=N_s^{\delta(y)}=N_t^{\delta(y)}$. Now consider any neighborhood $\delta'(y) \supset \delta(y)$, we have
	\be
	\frac{N_{s\cup t}^{\delta^\prime(y)}}{N_s^{\delta^\prime(y)} N_t^{\delta^\prime(y)}} \le \frac{N_{s\cup t}^{\delta(y)}}{N_s^{\delta(y)} N_t^{\delta(y)}}
	~~\Leftrightarrow~~N_{s\cup t}^{\delta^\prime(y)} N_t^{\delta(y)} \le N_s^{\delta^\prime(y)} N_t^{\delta^\prime(y)},\:\forall s,t\in y.
	\ee
	The latter holds because $N_{s\cup t}^{\delta^\prime(y)}\le N_s^{\delta^\prime(y)}$ and $N_t^{\delta(y)}\le N_t^{\delta^\prime(y)},\:\forall s,t\in y$. In other words, enlarging the neighborhood from the smallest one $\delta(y) = \{y_n: y_n=y \}$ always decreases the ratio $N^{\delta(y)}_{s\cup t}/(N_s^{\delta(y)} N_t^{\delta(y)})$.

	We now show that the non-negative covariance assumption in \Cref{thm:opt_bayes_comparison} is critical. When travel times on different road segments can potentially be negatively correlated, we show, with the following example, that the optimal segment-based estimator can produce a strictly higher integrated risk than the optimal route-based estimator with neighborhood $\delta(y) = \{y_n: y_n = y\}$. %
	
	\begin{example}
		\label{ex:negative_cov}
		Using the same setup in \Cref{ex:estimators} without normality assumptions, we now compare the integrated risk of $\Thetasegopt$ with that of $\Thetatrpopt$ with a neighborhood $\delta(y) = \{y_n: y_n = y\}$, under negative covariances.
		Suppose $\sigma_{s_1}^2 = \sigma_{s_2}^2 = 1$ and $\sigma_{s_1,s_2} = \sigma_{s_2, s_1} = -0.9$. Let $\tau^2=1$. In this case, the optimal segment-based estimator $\Thetasegopt$ takes the form
		\be
		\Thetasegopt = & 0.811\cdot\frac{T^\prime_{1,s_1} + T^\prime_{4, s_1} + T^\prime_{5, s_1}}{3} + 0.189\cdot\mu\\
		&+ 0.811\cdot\frac{T^\prime_{2,s_2} + T^\prime_{3, s_2} + T^\prime_{4, s_2}}{3} + 0.189\cdot\mu,
		\ee
		with integrated risk $0.378$ (expected variance $0.307$ and expected squared bias $0.071$).
		
		On the other hand, the optimal route-based estimator $\Thetatrpopt$ takes the form,
		\be
		\Thetatrpopt = 0.909\cdot \left(T^\prime_{4,s_1} + T^\prime_{4, s_2}\right) + 0.091\cdot(2\mu),
		\ee
		with integrated risk $0.182$ (expected variance $0.165$ and expected squared bias $0.017$). 
	\end{example}

	The intuition behind the observation that negatively correlated segment travel time can benefit the route-based estimator is that route-level travel times can potentially \emph{absorb} the variance of segment travel times by avoiding additional aggregation. This could sometimes create an edge over the segment-based estimator even when the route-based estimator uses fewer samples. Negative correlations between the segments can occur, for instance, due to having traffic signals in the route. If one segment is slow due to a red signal, the subsequent segment can have faster travel time due to a green signal \citep{ramezani2012estimation}.

	Based on this observation in \Cref{ex:negative_cov}, it is reasonable to conjecture that when all the covariances are non-negative $\sigma_{s,t}\ge0,\:\forall s,t\in\mathcal{S}$, the optimal segment-based estimator $\Thetasegopt$ has the minimum integrated risk among \emph{all} generalized segment-based estimators $\ThetasegG$. It appears at first glance that aggregating segment travel times into super-segment travel times in this case does not help reduce the overall variance of the estimator. Surprisingly, the next example shows that this might \emph{not} be the case.
	
	\begin{example}
		\label{ex:seg_not_optimal}
		Using the same setup in \Cref{ex:estimators} without normality assumptions, we consider the optimal segment-based and generalized segment-based estimator for the travel time of a new route $y$ traversing through $(1,2)\rightarrow(1,3)\rightarrow(2,3)\rightarrow(3,3)$. We call segment $(1,2)\rightarrow(1,3)$ to be $s_3$, segment $(1,3)\rightarrow(2,3)$ to be $s_4$ and segment $(2,3)\rightarrow(3,3)$ to be $s_5$. Consider $\sigma_{s_3}^2 = 0.1$, $\sigma_{s_4}^2 = 10$, $\sigma_{s_5}^2 = 10$, $\sigma_{s_3,s_4} = 1$, and $\sigma_{s_3,s_5} = \sigma_{s_4,s_5} = 0$. Let $\tau^2=1$. One can check that this is a valid covariance matrix. We first compute the optimal segment-based estimator $\Thetasegopt$, which takes the form
		\be
		\Thetasegopt =& \left(0.866\cdot T^\prime_{3,s_3} + 0.134\cdot\mu\right)
		+ \left(0.094\cdot\frac{T^\prime_{3,s_4} + T^\prime_{6, s_4}}{2} + 0.906\cdot\mu\right)
		+ \left(0.091\cdot T^\prime_{6,s_5}+ 0.909\cdot\mu\right).
		\ee
		The integrated risk of $\Thetasegopt$ is $1.948$ (expected variance $0.284$ and expected squared bias $1.664$).

		Now the optimal generalized segment-based estimator $\ThetasegoptG$ with $\mathscr{S}_y=\{\{s_3\},\{s_4, s_5\}\}$. It takes the form
		\be
		\label{eq:gseg_form}
		\ThetasegoptG =& \left(0.909\cdot T^\prime_{3,s_3} + 0.091\cdot\mu\right) +\left(0.091\cdot\left(T^\prime_{6,s_4} + T^\prime_{6,s_5}\right) + 0.909\cdot(2\mu)\right).
		\ee
		The integrated risk of $\ThetasegoptG$ is $1.909$ (expected variance $0.248$ and expected squared bias $1.661$).  
	\end{example}
	
	The reason that the optimal segment-based estimator is not the best among all generalized segment-based estimators under non-negative covariances is quite subtle. In \Cref{ex:seg_not_optimal}, the only pair of segments that are correlated are $s_3$ and $s_4$ with a positive covariance $1$. Interestingly, merging $s_4$ and $s_5$ into a super-segment avoids increasing the variance of the estimator resulting from the positive covariance between a \emph{different} pair of segments $s_3$ and $s_4$. To see that, in the form of $\ThetasegoptG$ (equation \eqref{eq:gseg_form}), historical traversal data on segments $s_3$ and $s_4$ from trip $\mathcal{T}_3$ are not both used in the estimator because $\mathcal{T}_3$ does not traverse through all three segments $s_3$, $s_4$ and $s_5$. In other words, merging segments into super-segments sometimes breaks the dependency of two segments within a different super-segment. This is achieved by creating a higher barrier for the historical traversals on these segments within the same trip to be included in the estimator.

	Nevertheless, we have the following proposition when the optimal generalized segment-based estimator uses the entire route $y$ as a super-segment, i.e., $\mathscr{S}_y=\{\{y\}\}$, and the optimal route-based estimator $\Thetatrpopt$ uses the neighborhood $\delta(y)=\{y_n: y_n=y\}$ that contains the exact same route as $y$ in the historical data. There is a (subtle) difference between these two estimators. The former includes all traversals that go through $y$, i.e., $y$ is a sub-path of the traversals.  
 On the other hand, the latter only includes historical routes that share the exact same route as $y$ including its origin and destination.

	\begin{proposition}
		\label{cor:same_route}
		In addition to Assumptions \ref{assump:1} and \ref{assump:2}, suppose that $\sigma_{s,t}\ge 0$ for all $s,t \in y$. Let $\Thetasegopt$ be the optimal segment-based estimator, $\ThetasegoptG$ be the optimal generalized segment-based estimator with $\mathscr{S}_y=\{\{y\}\}$, and $\Thetatrpopt$ be the optimal route-based estimator with neighborhood $\delta(y)=\{y_n: y_n=y\}$,
		\be
		R\left(\Thetasegopt \,\middle\vert\,  y_{[N]}\right) \le R\left(\ThetasegoptG \,\middle\vert\,  y_{[N]}\right) \le R\left(\Thetatrpopt \,\middle\vert\,  y_{[N]}\right),
		\ee
		for any set of historical routes $y_{[N]}$.
	\end{proposition}

	We conclude this section by commenting that although \Cref{thm:opt_bayes_comparison} and \Cref{cor:same_route} give some evidence in terms of the superiority of the optimal segment-based estimator, \Cref{ex:negative_cov} and \Cref{ex:seg_not_optimal} also point out that there are cases where the comparisons are not clean. To garner more insights, in the next section, we are going to analyze an \emph{asymptotic} setting where the number of trip observations grows with the size of the road network. %
	
	\section{Asymptotic Analysis}
	\label{sec:asymptotic}
	
	In this section, we compare estimators in terms of how their integrated risks scale with the road network size. We consider an asymptotic setting where the number of trip observations grows with the size of the road network.\footnote{One can also consider the context of a fixed size road network, and analyze the efficiency as the number of trips $N\rightarrow \infty$. This is less informative because nearly all reasonable approaches have the same asymptotic rate as a function of $N$, but with very large (and meaningful) differences in their constants.} This regime is relevant in practice, since the road network of a major metropolitan area typically contains hundreds of thousands to millions of road segments, and typical commercial datasets contain tens of millions of trips in such a network \citep{li2018multi}. One benefit of such an asymptotic analysis is to compare estimators in a more tractable setting, enabling comparisons that can't be done in the finite-sample setting. 
 
 We start by pointing out that the optimal segment-based estimator $\Thetasegopt$ %
	requires inverting a $|y|\times|y|$ matrix which could be computationally intensive for real-time implementation on large-scale road networks. Moreover, it also requires explicit knowledge of the covariance structures among each pair of road segments $\sigma_{s,t}$, which can be hard to precisely estimate in practice. %
	Our goal in this section is to see if a similar (or stronger) result as \Cref{thm:opt_bayes_comparison} or \Cref{cor:same_route} holds in an asymptotic limit with a class of \emph{much simpler} segment-based estimators. These simple segment-based estimators are tractable to compute for large road networks and do not require any knowledge of the covariance structures. In addition, we aim to generalize our result to a case where there exist \emph{negative} correlations among segment travel times. Finally, the asymptotic analysis also enables us to compare the segment-based estimators with generalized segment-based estimators, which we are not able to do in the finite-sample case.

	We first introduce our asymptotic setting. Consider a road network indexed by a size $p\in\mathbb{N}$, with a set of vertices (intersections) $\mathcal{V}_p$ and a set of edges (road segments) $\mathcal S_p$. %
	An example road network is the grid network where $p$ represents the size of the grid. For a grid network with size $p$, $|\mathcal{S}_p|=|\mathcal{V}_p|\simeq p^2$. Let $\mathcal{Y}_p$ be the set of all possible routes in the road network of size $p$. We assume that any route $y\in\mathcal{Y}_p$ contains at least one road segment, $|y|\ge1$. The number of trips $N$ in the training data grows with $p$, such that $N\rightarrow \infty$ as $p\rightarrow \infty$, although this is not strictly required for any of the following results. In addition,
	\begin{assumption}
		\label{assumption:3}
		Assume that:
		\begin{enumerate}
			\item For each road network with size $p$, the historical routes $Y_{p,[N]}$ in the training data as well as the predicting route $Y_p$ are drawn independently according to some probability distribution $\mu_p$ over $\mathcal{Y}_p$. 
			\item The covariance matrix $\Sigma_p = [\sigma_{s,t}]_{s,t\in\mathcal{S}_p}$ and its corresponding precision matrix $\Psi_p = \Sigma^{-1}_p = [\psi_{s,t}]_{s,t\in\mathcal{S}_p}$ under network size $p$ satisfy $\sum_{t\in\mathcal{S}_p}|\sigma_{s,t}| = \mathcal{O}(1)$ and $\sum_{t\in\mathcal{S}_p}|\psi_{s,t}| = \mathcal{O}(1),\:\forall s\in\mathcal{S}_p$. Moreover, there exists $\sigma_{\textrm{min}}>0$ such that for any route $y_p\in\mathcal{Y}_p$, $\sum_{s,t\in y_p} \sigma_{s,t} \ge \sigma_{\textrm{min}}$.

		\end{enumerate}  
	\end{assumption}
	
	The first part of the assumption introduces a route distribution $\mu_p$ for each size of the road network, from which historical routes and predicting routes are sampled. Note that $Y_{p,[N]}$ and $Y_p$ are capitalized because they are random in this setting. The second part of the assumption is justifiable in the ETA prediction context since spatial decay in the correlation of segment travel times is widely observed in empirical studies --- the correlation between two road segments decays as the distance between the two segments increases (see e.g., \citealt{Bernard2006CorrelationOL,rachtan2013spatiotemporal,guo2020understanding,woodard2017predicting}). It further implies that the sum of all the (co)variance components in the road network grows at most linearly to the total number of segments, $\sum_{s,t\in\mathcal{S}_p}\sigma_{s,t}\le\sum_{s,t\in\mathcal{S}_p}|\sigma_{s,t}|=\mathcal{O}(|\mathcal{S}_p|)$. %

	Similarly to \Cref{sec:finite}, we compare the accuracy of different estimators using integrated risk. To obtain our results in this asymptotic setting where routes are randomly sampled, we slightly alter the definition of the integrated risk used in \Cref{sec:finite}. Specifically, for a given road network of size $p$, we leverage Assumption \ref{assumption:3} to integrate the risk over the distribution of historical routes $Y_{p,[N]}$ and predicting route $Y_p$. This yields the following definition of integrated risk:
	\be
	\label{eq:risk_unconditional}
	R\left(\hat \Theta_{Y_p}\right) := \bb E\left[\left(\hat \Theta_{Y_p} - \Theta_{Y_p}\right)^2\right],%
	\ee
	where the expectation is now taken with respect to (1) the distribution over the historical routes $Y_{p,[N]}$, (2) the predicting route $Y_p$, in addition to (3) the adjusted travel times $\{T'_{n,s}\}_{n\in[N], s\in Y_{p,n}}$ and (4) the population distribution on $\{\theta_s\}_{s\in\mathcal{S}_p}$. Our results will compare the asymptotic integrated risk of travel time estimators $R(\hat \Theta_{Y_p})$ as $p \rightarrow \infty$ (and $N \rightarrow \infty$).

	\subsection{Grid Networks}
	\label{subsec:grid}
	
	We consider an example of grid road networks. Let $x = (i,j) \in \mathcal{V}_p$ for $\mathcal{V}_p = \{0,\ldots,p\}^2$ denote a vertex on the grid (a possible start or end point of a route), and $s\in \mathcal S_p$ denote a road segment, i.e., a directed edge between adjacent vertices. %
	We define the route distribution $\mu_p$ under grid size $p$ by assuming that the trip's origin $x_1 = (i_1,j_1)$ and destination $x_2 = (i_2,j_2)$ are drawn independently from the following probability distribution over vertices:
	\be %
	&\bb P[X = (i,j)] =\prod_{k\in\{i,j\}} {p \choose k} \frac{B(\alpha + k,\alpha+p-k)}{B(\alpha,\alpha)},
	\ee
	where $0 < \alpha \le 1$ and $B(\cdot,\cdot)$ denotes the beta function.\footnote{The case of $\alpha>1$ is less interesting as origins and destinations concentrate within the center of the grid, and so trips do not fully utilize the entire $p$ by $p$ grid. This case can somewhat be captured by a grid with a smaller size. Nevertheless, we look at the case of $\alpha>1$ in the numerical experiments in \Cref{sec:numerical}.}
	In other words, the east-west and north-south coordinates of the origin and destination are independently sampled from a \emph{symmetric beta-binomial distribution}. When $\alpha < 1$, this distribution has a ``horseshoe'' shape, with a high probability at the edges of the grid and a low probability in the center. For $\alpha = 1$, this is just the uniform distribution over $\mathcal{V}_p$. As $\alpha$ decreases, the distribution more heavily weighs the locations near the four corners of the grid. Given the origin and the destination, routes are sampled from some distributions we do not put restrictions on first.

	We consider a neighborhood $\delta^{\textrm{od}}(\cdot)$ that includes all historical routes whose origins and destinations are close to those of the predicting route respectively. We define $x_1(y_p)$ and $x_2(y_p)$ as the origin and destination of route $y_p$. Construct $\delta^{\textrm{od}}(y_p) = \{y\in \mathcal{Y}_p: \|x_1(y), x_1(y_p)\|_1\le c,  \|x_2(y), x_2(y_p)\|_1\le c\}$ for some fixed constant $c>0$ that does \emph{not} depend on $p$. In our first asymptotic result below, we compare a large family of simple segment-based estimators $\Thetaseg$ to the optimal route-based estimators $\Thetatrpopt$ with neighborhood $\delta^{\textrm{od}}(\cdot)$. This family of simple segment-based estimators only requires that $\phi_s(N_s)$ approaches $1$ quickly enough. %
	We provide this result without restricting to non-negative covariances, as required in \Cref{thm:opt_bayes_comparison}.
	
	\begin{theorem}
		\label{cor:grid}
		Under Assumptions \ref{assump:1}, \ref{assump:2} and \ref{assumption:3}, consider an optimal route-based estimator $\Thetatrpopt$ based on a route neighborhood $\delta^{\textrm{od}}(\cdot)$ with similar origin and destination as the those of the predicting route, if $1/4< \alpha\le 1$,
		\be
		\lim_{p\to\infty}\frac{R\left(\ThetasegP\right)}{R\left(\ThetatrpoptP\right)} = 0,
		\ee
		for any segment-based estimator $\Thetaseg$ with $\phi_s(N_s) = 1-\mathcal{O}(1/\sqrt{N_s}),\:\forall s\in\mathcal{S}_p$. In addition,
		\be
		R\big(\ThetasegP\big)=\mathcal{O}(p^2/N),\quad R\big(\ThetatrpoptP\big)=\begin{cases}\Omega\left(p^4/N \right), & 1/2<\alpha\le 1, \\
			\Omega\left(p^{8\alpha}/N \right), & 0<\alpha\le 1/2.
		\end{cases}
		\ee
	\end{theorem}
	
	\begin{remark}
		This result holds under \emph{any} route distribution given origin and destination. In fact, the integrated risk of the simple segment-based estimator $R\big(\ThetasegP\big)=\mathcal{O}(p^2/N)$ holds under \emph{any} distribution of origins and destinations.  		This result also does \emph{not} require any minimum data growth rate on $N$ as a function of $p$, which suggests that the result holds under data-sparse settings. 
	\end{remark} 
	
	\begin{remark}
		The family of segment-based estimators considered in \Cref{cor:grid} includes, for example, the optimal estimator under independent, Gaussian distributed segment travel times in \Cref{cor:independent_simple_form}, $\phi_s(N_s) = N_s \tau^2/(N_s\tau^2 + \sigma_s^2)$. It is interesting to note that the rate requirement $\phi_s(N_s) = 1-\mathcal{O}(1/\sqrt{N_s})$ gives some leeway in the sense that $\phi_s(N_s)$ can approach $1$ more slowly than what is required in the optimal estimator for the independent case. This family of segment-based estimators also includes other simple forms without any knowledge of the variance parameters, e.g., $\phi_s(N_s) = N_s/(N_s + \lambda)$, for any $\lambda>0$; or some threshold-based structure such as $\phi_s(N_s) = 1$ if $N_s\ge c$ with some constant $c>0$ and $\phi_s(N_s) = 0$ otherwise. 
		The former choice of $\phi_s(N_s)$ can be interpreted as an estimator $\hat{\Theta}_y$ minimizing penalized squared-error loss: $\min_{\hat{\theta}_s}\sum_{n: s\in y_n}(\hat{\theta}_s - T'_{n,s})^2 + \lambda(\hat{\theta}_s - \mu)^2$ where $\lambda>0$ is the regularizing parameter;
		and the latter choice of $\phi_s(N_s)$ can be thought of as a simple fallback logic --- predict the segment travel time using the sample average if sample size exceeds a threshold or using the population mean $\mu$ if there is not enough data. The simpler forms of $\phi_s(N_s)$ enhance the relevance of this result since, in practice, mapping services do not have access to the ``optimal'' form $\phi^\ast_s(N_s)$ and the choice of $\phi_s(N_s)$ is often tuned through cross-validation. %
	\end{remark}

        \begin{remark}
        \label{rmk:neighborhood}
        One can also develop a result similar to \Cref{cor:grid} where the size of the neighborhood $\delta^{\textrm{od}}(y_p)=\{y\in \mathcal{Y}_p: \|x_1(y), x_1(y_p)\|_1\le c,  \|x_2(y), x_2(y_p)\|_1\le c\}$, $c$, grows with the grid size $p$. One approach is to simply multiply the sample size by the number of distinct origin-destination pairs in the neighborhood and revise Theorem 3.1 accordingly. However, this is an overly optimistic lower bound of the integrated risk of the route-based method, because increasing the size of the neighborhood also introduces additional biases by including relatively irrelavant historical trips. In \Cref{sec:numerical}, we conduct numerical experiments to investigate this trade-off. These experiments show that having an increasing neighborhood size does not change the asymptotic comparisons in \Cref{cor:grid}. This suggests that the additional biases introduced by a larger neighborhood can offset the benefits of a larger sample size.
        \end{remark}

	\Cref{cor:grid} suggests that when the distributions of the route origins and destinations are not overly concentrated, a route-based estimator using routes with similar origins and destinations is \emph{asymptotically dominated} by a class of simple segment-based estimators. When origins and destinations of the routes are too concentrated, historical and predicting can be very similar --- in the extreme case where $\alpha \rightarrow 0^+$, all origins and destinations are concentrated at the four corners of the grid so that there are only 16 types of origin-destination pairs in the data where each route goes from one corner of the grid to another. This can give a route-based estimator some advantages. The requirement $1/4<\alpha\le 1$ ensures that there is enough dispersion among historical routes. This range is quite generous --- when $\alpha=1/4$, for a $10\times 10$ grid, the probability of sampling route origins or destinations at the corners is over $50$ times higher than the center of the grid, and this gap increases as the grid size increases. %

	\smallskip
	
	It turns out that we are able to say a lot more by directly comparing the segment-based estimators $\ThetasegP$ to the optimal estimator $\hat{\Theta}^\ast_{Y_p}$ characterized in \Cref{thm:optimal_estimator}.
	To do so, we first fully specify the route distribution $\mu_p$. Conditional on the origin and destination $x_1=(i_1,j_1)$ and $x_2=(i_2, j_2)$, we sample the route $Y_p$ uniformly from the set of all routes in $\mathcal Y_p$ that \emph{minimize both the number of traversals and turns} from $x_1$ to $x_2$, i.e., that have length equal to the grid distance $\|x_1-x_2\|_1 = |i_1 - i_2| + |j_1 - j_2|$ and the minimum number of turns. %
	\Cref{fig:grid2} below illustrates the route distribution given specific origin $x_1$ and destination $x_2$. On the left of \Cref{fig:grid2}, there is only one possible route between them, while on the right of \Cref{fig:grid2}, there are two possible routes, each with probability 0.5 being sampled. Although somewhat simplified, this route distribution $\mu_p$ biases towards route-based (and generalized segment-based) estimators as it significantly limits the set of possible routes $\mathcal{Y}_p$ and increases the sample size of each possible route $y\in\mathcal{Y}_p$. In other words, for segment-based estimators, having good relative performance under such a route distribution $\mu_p$ likely implies good relative performance under other route distributions.

	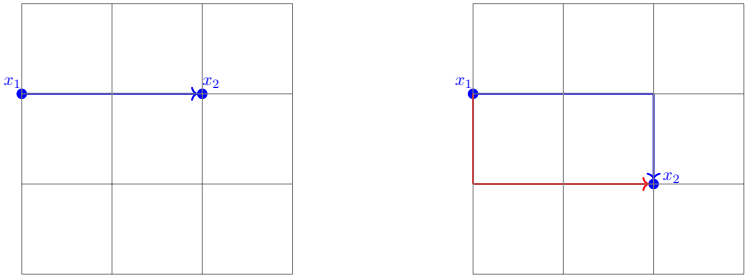
\begin{figure}[htbp]
		\centering
		\scalebox{0.6}{\begin{tikzpicture}

\draw node [right, blue, above] at (-0.2,4.0) {$x_1$};
\draw node [right, blue, above] at (4.2,4.0) {$x_2$};

\draw node [right, blue, above] at (9.8,4.0) {$x_1$};
\draw node [right, blue, above] at (14.4,1.9) {$x_2$};

\filldraw[blue] (10,4) circle (3pt);
\filldraw[blue] (14,2) circle (3pt);

\draw[very thick, ->, blue]   (0,4) -- (3.9,4);

\filldraw[blue] (0,4) circle (3pt);
\filldraw[blue] (4,4) circle (3pt);

\draw[very thick, -, blue]   (10,4) -- (14,4);
\draw[very thick, ->, blue]   (14,4) -- (14,2.1);
\draw[very thick, -, red]   (10,4) -- (10,2);
\draw[very thick, ->, red]   (10,2) -- (13.9,2);

\draw[step=2cm,gray,very thin] (0,0) grid (6,6);
\draw[step=2cm,gray,very thin] (10,0) grid (16,6);
\end{tikzpicture}}
		\caption{Examples of the route distribution $\mu_p$ conditional on origin $x_1$ and destination $x_2$.}
		\label{fig:grid2}
	\end{figure} 
	
	We now give the main result of the paper. We show that the same class of simple segment-based estimators considered in \Cref{cor:grid} is \emph{asymptotically optimal up to a logarithmic factor}. %
    Precisely, we compare its integrated risk with that of the optimal estimator $\hat{\Theta}^\ast_y$ under normal distributions and show that, in the asymptotic limit, the ratio of the risks can be upper bounded by a small logarithmic factor $\mathcal{O}(\log(p))$. Such comparison is non-trivial as we neither have a closed-form formula for the optimal estimator $\hat{\Theta}_{Y_p}^\ast$, nor for its risk (see \Cref{thm:optimal_estimator}). We thus compare $R(\ThetasegP)$ to a more tractable lower bound of the optimal risk $R(\hat{\Theta}_{Y_p}^{\ast})$. This lower bound is obtained by adapting the Bayesian Cram{\'e}r-Rao bound \citep{gill1995applications} through the van Trees inequality \citep{van2004detection}, a Bayesian analog of the information inequality (see \Cref{lm:info_lb} in the supplement).

	We now formally present the result regarding the asymptotic optimality of the segment-based estimators below.%
	
	\begin{theorem}[\textsc{Asymptotic Optimality of Segment-Based Estimators}]
		\label{thm:asymptotic_opt}
		In addition to Assumptions \ref{assump:1}, \ref{assump:2} and \ref{assumption:3}, assume that $(\mathcal{E}, \theta)$ are jointly Gaussian distributed. When $1/2\le\alpha\le 1$ and $N=\omega(p)$,
		\be
		\frac{R\left(\ThetasegP\right)}{R\left(\hat{\Theta}_{Y_p}^{\ast}\right)}=\mathcal{O}(\log(p)),
		\ee
		for any segment-based estimator with $\phi_s(N_s) = 1-\mathcal{O}(1/\sqrt{N_s}),\:\forall s\in\mathcal{S}_p$. When the data growth rate $N = \mathcal{O}(p)$, $\liminf_{p\rightarrow\infty}R\left(\hat{\Theta}_Y\right) > 0$
		for any estimator $\hat{\Theta}_Y$.
	\end{theorem}

    \Cref{thm:asymptotic_opt} has strong practical implications. It says that although improvement can be made in some cases over simple segment-based estimators by, for example, using a route-based method (\Cref{ex:negative_cov}) or combining segments into super-segments (\Cref{ex:seg_not_optimal}), their benefits are limited and in the asymptotic limit where grid size and sample size grow, these benefits can only make a difference up to a logarithmic factor. This gives reassurance that maintaining a segment-based travel time prediction architecture achieves most of the accuracy of the optimal estimator. Similar to \Cref{cor:grid}, \Cref{thm:asymptotic_opt} does require some conditions to make sure that historical routes are diverse enough. The condition $1/2\le\alpha\le 1$ is stricter than the one required in \Cref{cor:grid} but still quite generous --- when $\alpha=1/2$, for a $10\times 10$ grid, the probability of sampling route origins or destinations at the corners is more than $8$ times higher than in the center of the grid and this gap again increases as the grid size increases. In addition, \Cref{thm:asymptotic_opt} also requires the data growth rate to be at least $N=\omega(p)$. This turns out to be a very mild condition as for any slower data growth rate $N=\mathcal{O}(p)$, \emph{no} estimator can be consistent in the sense that the asymptotic risk tends to zero.

	\subsection{Numerical Examples}
	\label{sec:numerical}
	
	We numerically demonstrate the accuracy of different estimators based on representative correlation structures used for travel times on road networks. %
	We construct $p\times p$ grid networks where $p\in\{10,15,20,25,30\}$. For each grid size $p$, we consider different sample sizes of historical routes $N=p, p^2, p^3$, and $p^4$. Each historical route is generated from a route distribution $\mu_p$ as detailed at the beginning of \Cref{subsec:grid}. We consider a route distribution with $\alpha=1.0$ --- origins and destinations are generated uniformly over the grid.
    The covariance matrix of segment travel times is taken to be $u e^{-v\mathscr{L}} + I$ where $u, v>0$ are some parameters and $I$ is an identity matrix representing a white noise of travel time uncertainty. The matrix $\mathscr{L}=D^{-1/2}LD^{1/2}$ is the normalized Laplacian of the grid network where $D$ is the diagonal matrix of segment degrees, $A$ is the adjacency matrix of the grid network and $L=D-A$ is the graph Laplacian. This is also called the diffusion kernel. It models spatial decay of correlation among segment travel times. As $u$ increases, the matrix becomes more diffused in the sense that the correlation becomes relatively stronger. On the other hand, $v$ controls the weight between the diffusion kernel and the white noise. Note that this covariance structure also satisfies the second part of Assumption \ref{assumption:3}. We set the population variance of the means of segment travel times to be $\tau^2 = 0.5$, which is similar to the variances of the segment travel times $\sigma^2_s$.

	For each grid size $p\in\{10,15,20,25,30\}$ and taking the parameter of the route distribution to be $\alpha\in\{0.3, 1.0, 3.0\}$, we generate $100$ predicting routes and report the average integrated risk of these predicting routes for the following methods, under a covariance matrix $u e^{-v\mathscr{L}} + I$ specified with $(u,v) = (1,1)$.
	
	\begin{enumerate}
		\vspace{1mm}
		\item Simple segment-based method $\hat{\Theta}_{y_p}^{(\textrm{simple-seg})}$
		with $\phi_s(N_s) = N_s/(N_s + 1)$.
		
		\vspace{1mm}
		
		\item Optimal route-based method $\Thetatrpoptp$ %
		with $\delta(y_p) = \{y_n: x_1(y_n) = x_1(y_p), x_2(y_n) = x_2(y_p)\}$ that includes all historical routes sharing the same origin and destination with the predicting route $y_p$. 
		\vspace{1mm}
            \item Optimal route-based method $\Thetatrpoptp$ with a growing neighborhood $\delta(y_p) = \{y_n: \|x_1(y_n) - x_1(y_p)\|_1\le \left\lceil 0.1p \right\rceil, \|x_2(y_n) - x_2(y_p)\|_1\le \left\lceil 0.1p \right\rceil \}$ that includes all historical routes sharing similar origin and destination with the predicting route $y_p$ where the degree of similarity is growing with the grid size $p$.
            \vspace{1mm}
		\item Optimal estimator $\hat{\Theta}_{y_p}^\ast$ in \Cref{thm:optimal_estimator} under the assumption that $(\mathcal{E}, \theta)$ are jointly Gaussian distributed.
		\vspace{1mm}
		\item The information-theoretic lower bound developed in \Cref{lm:info_lb} in the supplement under the assumption that $(\mathcal{E}, \theta)$ are jointly Gaussian distributed. %
	\end{enumerate}
	
	\vspace{1mm}

    We reiterate that the average integrated risks of the simple segment-based method and the optimal route-based method do not depend on distributional assumptions.
    Figures \ref{fig:risk_0.3_1_1_1}, \ref{fig:risk_1.0_1_1_1} and \ref{fig:risk_3.0_1_1_1} report the average integrated risks (in logarithmic scale) of the aforementioned methods %
	over 100 predicting routes, under different sample sizes as the grid size $p$ increases. Each figure corresponds to a different route distribution. \Cref{fig:risk_0.3_1_1_1} depicts the situation under $\alpha=0.3$ where origins and destinations of the routes are more concentrated at the corners of the grid; \Cref{fig:risk_1.0_1_1_1} represents $\alpha = 1.0$ where route origins and destinations are uniformly distributed over the grid; and finally \Cref{fig:risk_3.0_1_1_1} reports the situation of $\alpha=3.0$ where origins and destinations are more concentrated at the central part of the grid, though this is outside the range of $\alpha\in(0,1]$ we assume in \Cref{subsec:grid}. %
    These numerical findings match our theoretical results. The increasing difference between the integrated risks of the two optimal route-based methods and those of the simple segment-based method reflects the asymptotic dominance result in \Cref{cor:grid} as well as \Cref{rmk:neighborhood} that the dominance results likely remain true even if we consider a route-based method with growing neighborhood size (marked by ``Optimal Route GN'').
	The average risks of the simple segment-based method tend to zero when the sample size grows faster than $N=p^{2.0}$.%
	The performance of the simple segment-based estimator is extremely competitive --- it almost matches the optimal risk. The gap between the risk of the simple segment-based method and the information-theoretic lower bound increases very mildly as grid size $p$ increases, which reflects the logarithmic scaling in \Cref{thm:asymptotic_opt}. It is worth noting that both the cases of $\alpha=3.0$ and $\alpha=0.3$ are outside the range of $\alpha$ assumed in \Cref{thm:asymptotic_opt} but the optimality seems to remain valid. %
	 In \Cref{sec:additional_num} of the supplement, we report additional numerical experiments under other covariance matrices $u e^{-v\mathscr{L}} + I$ specified with other values of $(u,v)$ as well as covariance structures that violate the second part of Assumption \ref{assumption:3}.

 	\begin{figure}[htbp]
		\captionsetup[subfigure]{justification=centering}
		\centering 
		
		\begin{subfigure}[h]{0.45\textwidth}
			\begin{tikzpicture}
				\begin{axis}[
					width=2.6in,
					xlabel = Grid Size ($p$),
					ylabel = $\log_{10}\textrm{(Average Risk)}$,
					legend style={legend pos=south east, font=\tiny},
					label style={font=\footnotesize},
					ymin=0,
					ymax=1.2,
					xtick={10,15,20,25,30}
					]
					\addplot+ [
					discard if not={alpha}{1.0},
					] table [
					x=grid_size,
					y=seg_simple,
					]{data/revision/results_0.3_cov_6_1_1_1.txt};\addlegendentry{Simple Segment};

					\addplot+ [
					discard if not={alpha}{1.0},
					] table [
					x=grid_size,
					y=route,
					]{data/revision/results_0.3_cov_6_1_1_1.txt};\addlegendentry{Optimal Route};

          			\addplot+ [
					discard if not={alpha}{1.0},
					] table [
					x=grid_size,
					y=route_grow,
					]{data/revision/results_0.3_cov_6_1_1_1.txt};\addlegendentry{Optimal Route GN};
					
					\addplot+ [
					discard if not={alpha}{1.0},
					] table [
					x=grid_size,
					y=bayes_optimal,
					]{data/revision/results_0.3_cov_6_1_1_1.txt};\addlegendentry{Optimal};
					
					\addplot+ [
					discard if not={alpha}{1.0},
					] table [
					x=grid_size,
					y=lb,
					]{data/revision/results_0.3_cov_6_1_1_1.txt};\addlegendentry{Lower Bound};
				\end{axis}
			\end{tikzpicture}
			\caption{$N= p^{1.0}$}
			\label{fig:risk_0.3_1_1_1_1.0}
		\end{subfigure}
		~~
		\begin{subfigure}[h]{0.45\textwidth}
			\begin{tikzpicture}
				\begin{axis}[
					width=2.6in,
					xlabel = Grid Size ($p$),
					ylabel = $\log_{10}\textrm{(Average Risk)}$,
					legend style={legend pos=south east, font=\tiny},
					label style={font=\footnotesize},
					ymin=-1.5,
					ymax=1.5,
					xtick={10,15,20,25,30}
					]
					\addplot+ [
					discard if not={alpha}{2.0},
					] table [
					x=grid_size,
					y=seg_simple,
					]{data/revision/results_0.3_cov_6_1_1_1.txt};\addlegendentry{Simple Segment};

					\addplot+ [
					discard if not={alpha}{2.0},
					] table [
					x=grid_size,
					y=route,
					]{data/revision/results_0.3_cov_6_1_1_1.txt};\addlegendentry{Optimal Route};

     				\addplot+ [
					discard if not={alpha}{2.0},
					] table [
					x=grid_size,
					y=route_grow,
					]{data/revision/results_0.3_cov_6_1_1_1.txt};\addlegendentry{Optimal Route GN};
					
					\addplot+ [
					discard if not={alpha}{2.0},
					] table [
					x=grid_size,
					y=bayes_optimal,
					]{data/revision/results_0.3_cov_6_1_1_1.txt};\addlegendentry{Optimal};
					
					\addplot+ [
					discard if not={alpha}{2.0},
					] table [
					x=grid_size,
					y=lb,
					]{data/revision/results_0.3_cov_6_1_1_1.txt};\addlegendentry{Lower Bound};
				\end{axis}
			\end{tikzpicture}
			\caption{$N= p^{2.0}$}
			\label{fig:risk_0.3_1_1_1_2.0}
		\end{subfigure}
		
		\begin{subfigure}[h]{0.45\textwidth}
			\begin{tikzpicture}
				\begin{axis}[
					width=2.6in,
					xlabel = Grid Size ($p$),
					ylabel = $\log_{10}\textrm{(Average Risk)}$,
					legend style={legend pos=south west, font=\tiny},
					label style={font=\footnotesize},
					ymin=-4.0,
					ymax=1.5,
					xtick={10,15,20,25,30},
					ytick={-4,-3,-2,-1,0,1}
					]
					\addplot+ [
					discard if not={alpha}{3.0},
					] table [
					x=grid_size,
					y=seg_simple,
					]{data/revision/results_0.3_cov_6_1_1_1.txt};\addlegendentry{Simple Segment};

					\addplot+ [
					discard if not={alpha}{3.0},
					] table [
					x=grid_size,
					y=route,
					]{data/revision/results_0.3_cov_6_1_1_1.txt};\addlegendentry{Optimal Route};

					\addplot+ [
					discard if not={alpha}{3.0},
					] table [
					x=grid_size,
					y=route_grow,
					]{data/revision/results_0.3_cov_6_1_1_1.txt};\addlegendentry{Optimal Route GN};
					
					\addplot+ [
					discard if not={alpha}{3.0},
					] table [
					x=grid_size,
					y=bayes_optimal,
					]{data/revision/results_0.3_cov_6_1_1_1.txt};\addlegendentry{Optimal};
					
					\addplot+ [
					discard if not={alpha}{3.0},
					] table [
					x=grid_size,
					y=lb,
					]{data/revision/results_0.3_cov_6_1_1_1.txt};\addlegendentry{Lower Bound};
				\end{axis}
			\end{tikzpicture}
			\caption{$N= p^{3.0}$}
			\label{fig:risk_0.3_1_1_1_3.0}
		\end{subfigure}
		~~
		\begin{subfigure}[h]{0.45\textwidth}
			\begin{tikzpicture}
				\begin{axis}[
					width=2.6in,
					xlabel = Grid Size ($p$),
					ylabel = $\log_{10}\textrm{(Average Risk)}$,
					legend style={legend pos=south west, font=\tiny},
					label style={font=\footnotesize},
					ymin=-6.0,
					ymax=1.2,
					xtick={10,15,20,25,30},
					ytick={-6,-5,-4,-3,-2,-1,0,1}
					]
					\addplot+ [
					discard if not={alpha}{4.0},
					] table [
					x=grid_size,
					y=seg_simple,
					]{data/revision/results_0.3_cov_6_1_1_1.txt};\addlegendentry{Simple Segment};

					\addplot+ [
					discard if not={alpha}{4.0},
					] table [
					x=grid_size,
					y=route,
					]{data/revision/results_0.3_cov_6_1_1_1.txt};\addlegendentry{Optimal Route};

					\addplot+ [
					discard if not={alpha}{4.0},
					] table [
					x=grid_size,
					y=route_grow,
					]{data/revision/results_0.3_cov_6_1_1_1.txt};\addlegendentry{Optimal Route GN};
					
					\addplot+ [
					discard if not={alpha}{4.0},
					] table [
					x=grid_size,
					y=bayes_optimal,
					]{data/revision/results_0.3_cov_6_1_1_1.txt};\addlegendentry{Optimal};
					
					\addplot+ [
					discard if not={alpha}{4.0},
					] table [
					x=grid_size,
					y=lb,
					]{data/revision/results_0.3_cov_6_1_1_1.txt};\addlegendentry{Lower Bound};
				\end{axis}
			\end{tikzpicture}
			\caption{$N= p^{4.0}$}
			\label{fig:risk_0.3_1_1_1_4.0}
		\end{subfigure}
		\caption{Average integrated risks of different methods ($\alpha = 0.3$).}
		\label{fig:risk_0.3_1_1_1}
	\end{figure}
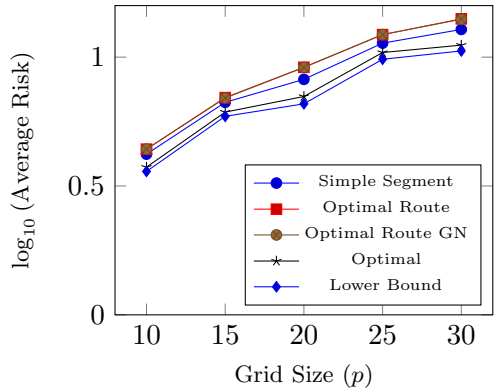
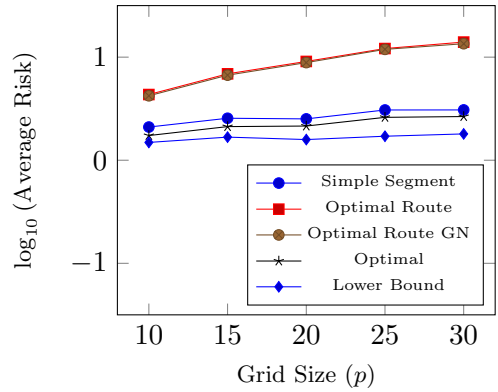
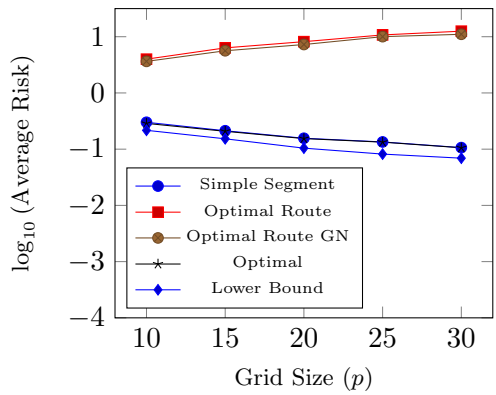
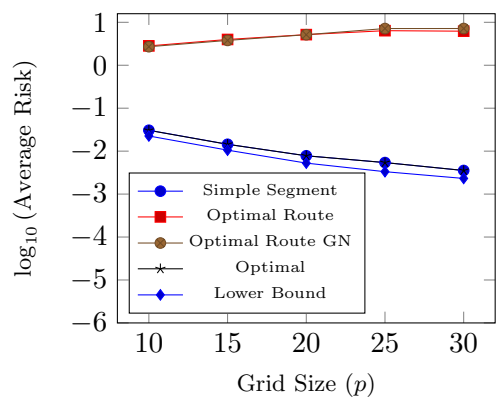

	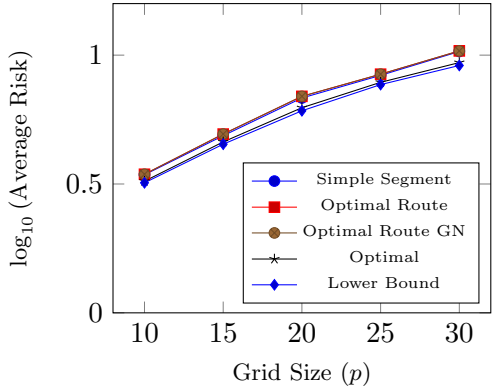
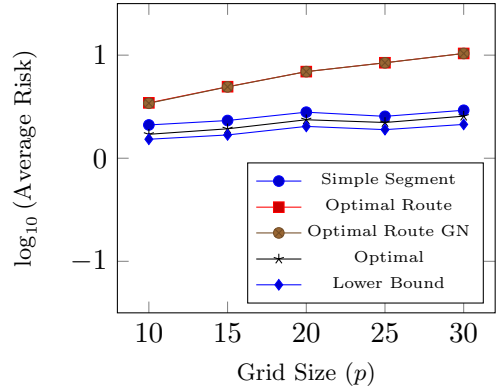
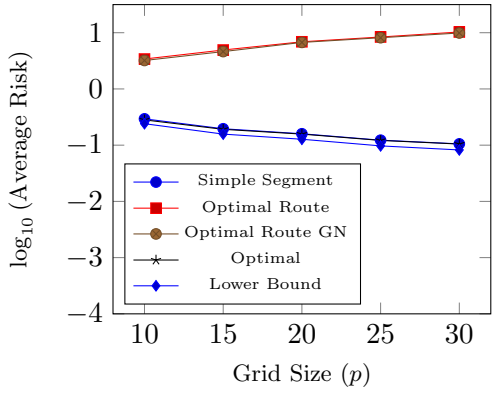
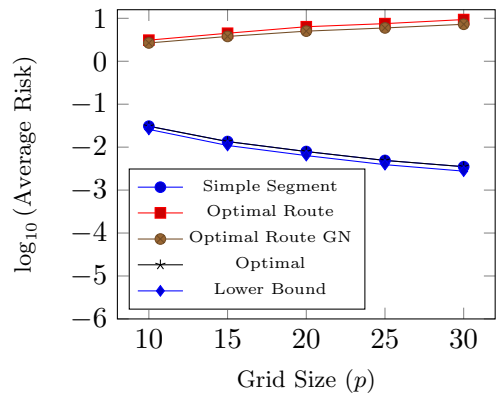
\begin{figure}[htbp]
		\captionsetup[subfigure]{justification=centering}
		\centering 
		
		\begin{subfigure}[h]{0.45\textwidth}
			\begin{tikzpicture}
				\begin{axis}[
					width=2.6in,
					xlabel = Grid Size ($p$),
					ylabel = $\log_{10}\textrm{(Average Risk)}$,
					legend style={legend pos=south east, font=\tiny},
					label style={font=\footnotesize},
					ymin=0,
					ymax=1.2,
					xtick={10,15,20,25,30}
					]
					\addplot+ [
					discard if not={alpha}{1.0},
					] table [
					x=grid_size,
					y=seg_simple,
					]{data/revision/results_1.0_cov_6_1_1_1.txt};\addlegendentry{Simple Segment};

					\addplot+ [
					discard if not={alpha}{1.0},
					] table [
					x=grid_size,
					y=route,
					]{data/revision/results_1.0_cov_6_1_1_1.txt};\addlegendentry{Optimal Route};

          			\addplot+ [
					discard if not={alpha}{1.0},
					] table [
					x=grid_size,
					y=route_grow,
					]{data/revision/results_1.0_cov_6_1_1_1.txt};\addlegendentry{Optimal Route GN};
					
					\addplot+ [
					discard if not={alpha}{1.0},
					] table [
					x=grid_size,
					y=bayes_optimal,
					]{data/revision/results_1.0_cov_6_1_1_1.txt};\addlegendentry{Optimal};
					
					\addplot+ [
					discard if not={alpha}{1.0},
					] table [
					x=grid_size,
					y=lb,
					]{data/revision/results_1.0_cov_6_1_1_1.txt};\addlegendentry{Lower Bound};
				\end{axis}
			\end{tikzpicture}
			\caption{$N= p^{1.0}$}
			\label{fig:risk_1.0_1_1_1_1.0}
		\end{subfigure}
		~~
		\begin{subfigure}[h]{0.45\textwidth}
			\begin{tikzpicture}
				\begin{axis}[
					width=2.6in,
					xlabel = Grid Size ($p$),
					ylabel = $\log_{10}\textrm{(Average Risk)}$,
					legend style={legend pos=south east, font=\tiny},
					label style={font=\footnotesize},
					ymin=-1.5,
					ymax=1.5,
					xtick={10,15,20,25,30}
					]
					\addplot+ [
					discard if not={alpha}{2.0},
					] table [
					x=grid_size,
					y=seg_simple,
					]{data/revision/results_1.0_cov_6_1_1_1.txt};\addlegendentry{Simple Segment};

					\addplot+ [
					discard if not={alpha}{2.0},
					] table [
					x=grid_size,
					y=route,
					]{data/revision/results_1.0_cov_6_1_1_1.txt};\addlegendentry{Optimal Route};

     				\addplot+ [
					discard if not={alpha}{2.0},
					] table [
					x=grid_size,
					y=route_grow,
					]{data/revision/results_1.0_cov_6_1_1_1.txt};\addlegendentry{Optimal Route GN};
					
					\addplot+ [
					discard if not={alpha}{2.0},
					] table [
					x=grid_size,
					y=bayes_optimal,
					]{data/revision/results_1.0_cov_6_1_1_1.txt};\addlegendentry{Optimal};
					
					\addplot+ [
					discard if not={alpha}{2.0},
					] table [
					x=grid_size,
					y=lb,
					]{data/revision/results_1.0_cov_6_1_1_1.txt};\addlegendentry{Lower Bound};
				\end{axis}
			\end{tikzpicture}
			\caption{$N= p^{2.0}$}
			\label{fig:risk_1.0_1_1_1_2.0}
		\end{subfigure}
		
		\begin{subfigure}[h]{0.45\textwidth}
			\begin{tikzpicture}
				\begin{axis}[
					width=2.6in,
					xlabel = Grid Size ($p$),
					ylabel = $\log_{10}\textrm{(Average Risk)}$,
					legend style={legend pos=south west, font=\tiny},
					label style={font=\footnotesize},
					ymin=-4.0,
					ymax=1.5,
					xtick={10,15,20,25,30},
					ytick={-4,-3,-2,-1,0,1}
					]
					\addplot+ [
					discard if not={alpha}{3.0},
					] table [
					x=grid_size,
					y=seg_simple,
					]{data/revision/results_1.0_cov_6_1_1_1.txt};\addlegendentry{Simple Segment};

					\addplot+ [
					discard if not={alpha}{3.0},
					] table [
					x=grid_size,
					y=route,
					]{data/revision/results_1.0_cov_6_1_1_1.txt};\addlegendentry{Optimal Route};

					\addplot+ [
					discard if not={alpha}{3.0},
					] table [
					x=grid_size,
					y=route_grow,
					]{data/revision/results_1.0_cov_6_1_1_1.txt};\addlegendentry{Optimal Route GN};
					
					\addplot+ [
					discard if not={alpha}{3.0},
					] table [
					x=grid_size,
					y=bayes_optimal,
					]{data/revision/results_1.0_cov_6_1_1_1.txt};\addlegendentry{Optimal};
					
					\addplot+ [
					discard if not={alpha}{3.0},
					] table [
					x=grid_size,
					y=lb,
					]{data/revision/results_1.0_cov_6_1_1_1.txt};\addlegendentry{Lower Bound};
				\end{axis}
			\end{tikzpicture}
			\caption{$N= p^{3.0}$}
			\label{fig:risk_1.0_1_1_1_3.0}
		\end{subfigure}
		~~
		\begin{subfigure}[h]{0.45\textwidth}
			\begin{tikzpicture}
				\begin{axis}[
					width=2.6in,
					xlabel = Grid Size ($p$),
					ylabel = $\log_{10}\textrm{(Average Risk)}$,
					legend style={legend pos=south west, font=\tiny},
					label style={font=\footnotesize},
					ymin=-6.0,
					ymax=1.2,
					xtick={10,15,20,25,30},
					ytick={-6,-5,-4,-3,-2,-1,0,1}
					]
					\addplot+ [
					discard if not={alpha}{4.0},
					] table [
					x=grid_size,
					y=seg_simple,
					]{data/revision/results_1.0_cov_6_1_1_1.txt};\addlegendentry{Simple Segment};

					\addplot+ [
					discard if not={alpha}{4.0},
					] table [
					x=grid_size,
					y=route,
					]{data/revision/results_1.0_cov_6_1_1_1.txt};\addlegendentry{Optimal Route};

					\addplot+ [
					discard if not={alpha}{4.0},
					] table [
					x=grid_size,
					y=route_grow,
					]{data/revision/results_1.0_cov_6_1_1_1.txt};\addlegendentry{Optimal Route GN};
					
					\addplot+ [
					discard if not={alpha}{4.0},
					] table [
					x=grid_size,
					y=bayes_optimal,
					]{data/revision/results_1.0_cov_6_1_1_1.txt};\addlegendentry{Optimal};
					
					\addplot+ [
					discard if not={alpha}{4.0},
					] table [
					x=grid_size,
					y=lb,
					]{data/revision/results_1.0_cov_6_1_1_1.txt};\addlegendentry{Lower Bound};
				\end{axis}
			\end{tikzpicture}
			\caption{$N= p^{4.0}$}
			\label{fig:risk_1.0_1_1_1_4.0}
		\end{subfigure}
		\caption{Average integrated risks of different methods ($\alpha = 1.0$).}
		\label{fig:risk_1.0_1_1_1}
	\end{figure}

 	\begin{figure}[htbp]
		\captionsetup[subfigure]{justification=centering}
		\centering 
		
		\begin{subfigure}[h]{0.45\textwidth}
			\begin{tikzpicture}
				\begin{axis}[
					width=2.6in,
					xlabel = Grid Size ($p$),
					ylabel = $\log_{10}\textrm{(Average Risk)}$,
					legend style={legend pos=south east, font=\tiny},
					label style={font=\footnotesize},
					ymin=-0.15,
					ymax=1.2,
					xtick={10,15,20,25,30}
					]
					\addplot+ [
					discard if not={alpha}{1.0},
					] table [
					x=grid_size,
					y=seg_simple,
					]{data/revision/results_3.0_cov_6_1_1_1.txt};\addlegendentry{Simple Segment};

					\addplot+ [
					discard if not={alpha}{1.0},
					] table [
					x=grid_size,
					y=route,
					]{data/revision/results_3.0_cov_6_1_1_1.txt};\addlegendentry{Optimal Route};

          			\addplot+ [
					discard if not={alpha}{1.0},
					] table [
					x=grid_size,
					y=route_grow,
					]{data/revision/results_3.0_cov_6_1_1_1.txt};\addlegendentry{Optimal Route GN};
					
					\addplot+ [
					discard if not={alpha}{1.0},
					] table [
					x=grid_size,
					y=bayes_optimal,
					]{data/revision/results_3.0_cov_6_1_1_1.txt};\addlegendentry{Optimal};
					
					\addplot+ [
					discard if not={alpha}{1.0},
					] table [
					x=grid_size,
					y=lb,
					]{data/revision/results_3.0_cov_6_1_1_1.txt};\addlegendentry{Lower Bound};
				\end{axis}
			\end{tikzpicture}
			\caption{$N= p^{1.0}$}
			\label{fig:risk_3.0_1_1_1_1.0}
		\end{subfigure}
		~~
		\begin{subfigure}[h]{0.45\textwidth}
			\begin{tikzpicture}
				\begin{axis}[
					width=2.6in,
					xlabel = Grid Size ($p$),
					ylabel = $\log_{10}\textrm{(Average Risk)}$,
					legend style={legend pos=south east, font=\tiny},
					label style={font=\footnotesize},
					ymin=-1.5,
					ymax=1.5,
					xtick={10,15,20,25,30}
					]
					\addplot+ [
					discard if not={alpha}{2.0},
					] table [
					x=grid_size,
					y=seg_simple,
					]{data/revision/results_3.0_cov_6_1_1_1.txt};\addlegendentry{Simple Segment};

					\addplot+ [
					discard if not={alpha}{2.0},
					] table [
					x=grid_size,
					y=route,
					]{data/revision/results_3.0_cov_6_1_1_1.txt};\addlegendentry{Optimal Route};

     				\addplot+ [
					discard if not={alpha}{2.0},
					] table [
					x=grid_size,
					y=route_grow,
					]{data/revision/results_3.0_cov_6_1_1_1.txt};\addlegendentry{Optimal Route GN};
					
					\addplot+ [
					discard if not={alpha}{2.0},
					] table [
					x=grid_size,
					y=bayes_optimal,
					]{data/revision/results_3.0_cov_6_1_1_1.txt};\addlegendentry{Optimal};
					
					\addplot+ [
					discard if not={alpha}{2.0},
					] table [
					x=grid_size,
					y=lb,
					]{data/revision/results_3.0_cov_6_1_1_1.txt};\addlegendentry{Lower Bound};
				\end{axis}
			\end{tikzpicture}
			\caption{$N= p^{2.0}$}
			\label{fig:risk_3.0_1_1_1_2.0}
		\end{subfigure}
		
		\begin{subfigure}[h]{0.45\textwidth}
			\begin{tikzpicture}
				\begin{axis}[
					width=2.6in,
					xlabel = Grid Size ($p$),
					ylabel = $\log_{10}\textrm{(Average Risk)}$,
					legend style={legend pos=south west, font=\tiny},
					label style={font=\footnotesize},
					ymin=-4.0,
					ymax=1.5,
					xtick={10,15,20,25,30},
					ytick={-4,-3,-2,-1,0,1}
					]
					\addplot+ [
					discard if not={alpha}{3.0},
					] table [
					x=grid_size,
					y=seg_simple,
					]{data/revision/results_3.0_cov_6_1_1_1.txt};\addlegendentry{Simple Segment};

					\addplot+ [
					discard if not={alpha}{3.0},
					] table [
					x=grid_size,
					y=route,
					]{data/revision/results_3.0_cov_6_1_1_1.txt};\addlegendentry{Optimal Route};

					\addplot+ [
					discard if not={alpha}{3.0},
					] table [
					x=grid_size,
					y=route_grow,
					]{data/revision/results_3.0_cov_6_1_1_1.txt};\addlegendentry{Optimal Route GN};
					
					\addplot+ [
					discard if not={alpha}{3.0},
					] table [
					x=grid_size,
					y=bayes_optimal,
					]{data/revision/results_3.0_cov_6_1_1_1.txt};\addlegendentry{Optimal};
					
					\addplot+ [
					discard if not={alpha}{3.0},
					] table [
					x=grid_size,
					y=lb,
					]{data/revision/results_3.0_cov_6_1_1_1.txt};\addlegendentry{Lower Bound};
				\end{axis}
			\end{tikzpicture}
			\caption{$N= p^{3.0}$}
			\label{fig:risk_3.0_1_1_1_3.0}
		\end{subfigure}
		~~
		\begin{subfigure}[h]{0.45\textwidth}
			\begin{tikzpicture}
				\begin{axis}[
					width=2.6in,
					xlabel = Grid Size ($p$),
					ylabel = $\log_{10}\textrm{(Average Risk)}$,
					legend style={legend pos=south west, font=\tiny},
					label style={font=\footnotesize},
					ymin=-6.0,
					ymax=1.2,
					xtick={10,15,20,25,30},
					ytick={-6,-5,-4,-3,-2,-1,0,1}
					]
					\addplot+ [
					discard if not={alpha}{4.0},
					] table [
					x=grid_size,
					y=seg_simple,
					]{data/revision/results_3.0_cov_6_1_1_1.txt};\addlegendentry{Simple Segment};

					\addplot+ [
					discard if not={alpha}{4.0},
					] table [
					x=grid_size,
					y=route,
					]{data/revision/results_3.0_cov_6_1_1_1.txt};\addlegendentry{Optimal Route};

					\addplot+ [
					discard if not={alpha}{4.0},
					] table [
					x=grid_size,
					y=route_grow,
					]{data/revision/results_3.0_cov_6_1_1_1.txt};\addlegendentry{Optimal Route GN};
					
					\addplot+ [
					discard if not={alpha}{4.0},
					] table [
					x=grid_size,
					y=bayes_optimal,
					]{data/revision/results_3.0_cov_6_1_1_1.txt};\addlegendentry{Optimal};
					
					\addplot+ [
					discard if not={alpha}{4.0},
					] table [
					x=grid_size,
					y=lb,
					]{data/revision/results_3.0_cov_6_1_1_1.txt};\addlegendentry{Lower Bound};
				\end{axis}
			\end{tikzpicture}
			\caption{$N= p^{4.0}$}
			\label{fig:risk_3.0_1_1_1_4.0}
		\end{subfigure}
		\caption{Average integrated risks of different methods ($\alpha = 3.0$).}
		\label{fig:risk_3.0_1_1_1}
	\end{figure}
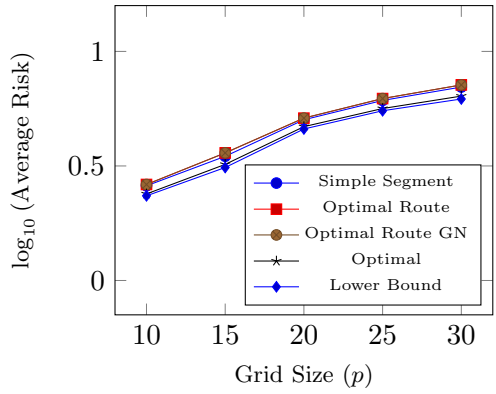
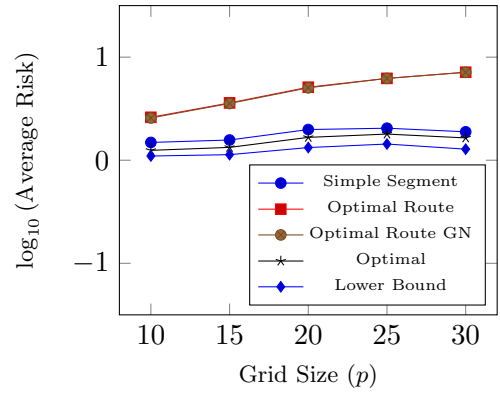
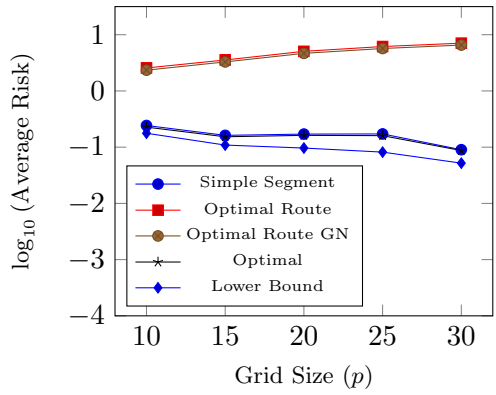
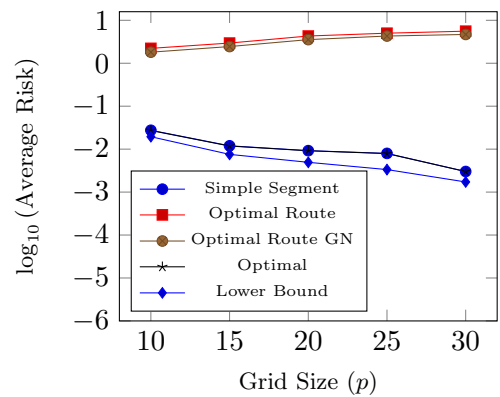

	\section{Concluding Remarks}
	\label{sec:dis}
	
	Our model and analysis reveal insights into the accuracy of various travel time predictors used in practice. 
	Our results favor segment-based estimators and show that a simple class of them is asymptotically optimal up to a logarithmic factor with a variety of trip-generating processes on a grid network. %
	At the core of our analysis is the following tradeoff. Segment-based estimators have the advantage of a larger sample size as there are more individual traversals on a segment level. However, the estimation can accumulate errors due to aggregating over road segments. On the other hand, route-based or generalized segment-based estimators can have the advantage of absorbing errors among segment travel times, but it is often at the cost of a smaller sample size. %
	Our results expose that, under mild conditions, the sample size difference is often of first-order importance, leading to favorable consideration towards a segment-based approach.

	It remains open whether similar insights hold under the setting of ETA prediction where one is only interested in predicting travel time from an origin to a destination without conditional on a route. Such settings occur in practice, for example, when one has little control over the route a driver will take \citep{hu2022deepreta}. Route-based methods which use data for all trip observations between the origin-destination pair can estimate travel time and the uncertain route distribution simultaneously, while segment-based methods require additional steps to estimate such route distribution. Extending our analyses in such settings can be meaningful follow-up work. %

\bibliography{references}

\newpage

\appendix

\noindent{}\textbf{\huge{Appendix}}

\bigskip

\noindent{}%

	\section{Additional Technical Results}
	\label{apx:additional_results}

        In this section, we give additional finite-sample and asymptotic results that are either used in or complement the main text. Their corresponding proofs can be found in \Cref{appx:proofs}.
	
	\subsection{Finite Sample Results}
	
	Our first result in this subsection bounds the integrated risk of any estimator using a Bayesian information-theoretical bound. This lower bound is obtained by adapting the Bayesian Cram{\'e}r-Rao bound \citep{gill1995applications} through the van Trees inequality \citep{van2004detection}, a Bayesian analog of the information inequality. The proof uses Gaussian priors and posteriors to obtain explicit closed-form bounds.

	\begin{lemma}[\textsc{Information-Theoretic Lower Bound}]
	\label{lm:info_lb}
	In addition to Assumptions \ref{assump:1} and \ref{assump:2}, assume that $(\mathcal{E}, \theta)$ are jointly Gaussian distributed. Given a set of historical routes $y_{[N]}$ and the predicting route $y$, 
	\be
	\label{eq:lb}
	R\left(\hat{\Theta}^\ast_{y} \:\bigg|\: y_{[N]}\right)\ge \frac{|y|^2}{\sum_{s,t\in y} N_{s\cup t}\psi_{s,t}+ |y|/\tau^2},%
	\ee
	where $\psi_{s,t}$ is the $(s,t)^{\textrm{th}}$ element in the precision matrix of the segment travel times $\Psi$. 
\end{lemma}

\smallskip

	\subsection{Asymptotic Results}
	
	We begin with a few remarks on the notation used in this section: for two functions $f(p)$ and $g(p)>0$, we write $f(p) = \mathcal{O}(g(p))$ (or $f(p) = \Omega(g(p))$) if there exists a constant $c_1$ and a constant $p_1$ such that $f(p) \le c_1g(p)$ (or $f(p) \ge c_1g(p)$) for all $p \ge p_1$; we write $f(p) = o(g(p))$ (or $f(p) = \omega(g(p))$) if $\lim_{p\to\infty} f(p)/g(p) = 0$ (or $\lim_{p\to\infty} f(p)/g(p) = +\infty$). In addition, we write $f(p)\gtrsim g(p)$ (or $f(p)\lesssim g(p)$) if there is a universal constant $c>0$ such that $f(p)\ge c g(p)$ (or $f(p)\le c g(p)$) for all $p\ge1$. %
	If $f(p)\lesssim g(p)$ and $f(p)\gtrsim g(p)$, we define $f(p)\simeq g(p)$. %
	
	Our first asymptotic result in this subsection generalizes \Cref{cor:grid} to non-grid road networks and general route distribution $\mu_p$. The key quantities determining the integrated risks of these estimators are the rates at which training data accumulates on particular road segments or on particular routes. Under a road network with size $p$, we let $q_s := \bb P[s\in Y_p]$ denote the probability that a specific road segment $s$ is traversed by a randomly generated route $Y_p\sim\mu_p$. %
	Similarly, given any route $y$, we define $q_{\delta(y)} := \bb P[Y_p \in \delta(y)]$ as the probability that a randomly sampled route belongs to the  neighborhood $\delta(y)$ of a given route $y$. With a slight abuse of notation, we further let $q_{\delta} := \bb P_{Y_p^\prime\sim\mu_p, Y_p\sim\mu_p}[Y_p^\prime \in \delta(Y_p)]=\sum_{y\in\mathcal{Y}_p}q_{\delta(y)}\mathbb P[Y_p = y]$. The quantity $q_{\delta}$ marginalizes over the distribution of the predicting route $y$ and 
	is the probability that a randomly sampled route $Y'_p$ belongs to the neighborhood of another independently sampled route $Y_p$. Intuitively, $q_{\delta}$ measures the rate at which the neighborhood of any randomly sampled predicting route accumulates training data.

	Our results will compare the asymptotic integrated risk of travel time estimators $R(\hat \Theta_{Y_p})$ as $p \rightarrow \infty$ (and $N \rightarrow \infty$). In our first asymptotic result below, we compare a large family of simple segment-based estimators $\Thetaseg$ to the optimal route-based estimators $\Thetatrpopt$. This family of simple segment-based estimators only requires that $\phi_s(N_s)$ approaches $1$ quickly enough. We give conditions under which this family of simple segment-based estimators is (much) more accurate when the size of the road network gets larger. This automatically implies that the optimal segment-based estimator also dominates the optimal route-based estimator, under the same set of conditions. We provide this result without restricting to non-negative covariances, as required in \Cref{thm:opt_bayes_comparison}. %

	\begin{theorem}
		\label{thm:asymptotic_general_seg_better_route}
		For any segment-based estimator $\Thetaseg$ with $\phi_s(N_s) = 1-\mathcal{O}(1/\sqrt{N_s}),\: s\in\mathcal{S}_p$, and any optimal route-based estimator $\Thetatrpopt$ with neighborhood $\delta(\cdot)$ and route distribution $\mu_p$ such that $q_{\delta} = o(1/|S_p|)$, 
		\be
		\lim_{p\to\infty}\frac{R\left(\ThetasegP\right)}{R\left(\ThetatrpoptP\right)} = 0.
		\ee
		In addition, $R\big(\ThetasegP\big)=\mathcal{O}(|S_p|/N)$ while $R\big(\ThetatrpoptP\big)=\Omega\left(1/\left(Nq_{\delta}\right)\right).$ %

	\end{theorem}
	
	\Cref{thm:asymptotic_general_seg_better_route} characterizes conditions under which a wide class of simple segment-based estimators dominates the optimal route-based estimator with a neighborhood such that the probability of any route within the neighborhood of a randomly sampled route being sampled scales as $o(1/|S_p|)$. We will give more explanation to this scaling under a grid network example in \Cref{subsec:grid}. It is interesting to note that \Cref{thm:asymptotic_general_seg_better_route} does not require any conditions on $q_s$. In the proof of \Cref{thm:asymptotic_general_seg_better_route}, we lower bound the integrated risk of the optimal route-based estimator $R\big(\ThetatrpoptP\big)$ by assuming that there is no additional bias introduced by including historical trips whose routes are not exactly the same as $y$ into the neighborhood $\delta(y)$.

	\subsubsection{Asymptotic Results for the Grid Networks}
	
    We provide a few additional asymptotic results for the grid networks based on the route distribution described in \Cref{subsec:grid}. These are useful results that help to prove the main results \Cref{cor:grid} and \Cref{thm:asymptotic_opt}. They concern various data accumulation rates on the grid networks. The first lemma below bounds the probability that a specific origin or destination is chosen on the grid. 
	
	\begin{lemma} \label{lm:NodeLowerBound}
		With $0 < \alpha \le 1$, %
		\be
		p^{-2}\lesssim \bb P[X=(i,j)] \lesssim  p^{-2\alpha},~\forall(i, j)\in\mathcal{V}_p.
		\ee
	\end{lemma}
	
	These bounds are tight. As we will show in the proof, at the four corners, $\bb P[X = (0,p)] = \bb P[X = (p,0)] = \bb P[X = (0,0)] = \bb P[X = (p,p)] \simeq p^{-2\alpha}$, while at the center of the grid we have $\bb P[X = (\lceil p/2\rceil, \lceil p/2\rceil)] \simeq p^{-2}$. With this lemma, we now give another result that characterizes the rate of $q_{\delta}$, the probability that a randomly sampled route belongs to the neighborhood of another independently sampled route. In particular, we consider a neighborhood $\delta^{\textrm{od}}(\cdot)$ that includes all historical routes whose origins and destinations are close to those of the predicting route respectively. We define $x_1(y_p)$ and $x_2(y_p)$ as the origin and destination of route $y_p$. Construct $\delta^{\textrm{od}}(y_p) = \{y\in \mathcal{Y}_p: \|x_1(y), x_1(y_p)\|_1\le c,  \|x_2(y), x_2(y_p)\|_1\le c\}$ for some fixed constant $c>0$ that does \emph{not} depend on $p$.
	
	\begin{proposition}
		\label{prop:q_delta_od}
		Consider a route neighborhood $\delta^{\textrm{od}}(\cdot)$ that includes routes with similar origin and destination as those of the predicting route,
		\be
		q_{\delta^{\textrm{od}}} = \bb P_{Y_p^\prime\sim\mu_p, Y_p\sim\mu_p}[Y_p^\prime \in \delta^{\textrm{od}}(Y_p)] =\sum_{y\in\mathcal{Y}_p} q_{\delta^{\textrm{od}}}(y)\cdot\mathbb P[Y_p = y] \simeq \begin{cases} p^{-4}, &  1/2 < \alpha\le1, \\ p^{-8\alpha}, & 0<\alpha\le1/2. \end{cases}
		\ee
	\end{proposition}

	Under route distribution $\mu_p$ specified in \Cref{fig:grid2} , we give bounds on the probability that a road segment $s$ is traversed. %
	
	\begin{proposition}\label{prop:seg_accum}
		Let $Y_p \sim \mu_p$. %
		For any road segment $s\in\mathcal{S}_p$ in the grid,
		\be
		p^{-1-\alpha} \lesssim q_s = \bb P[s\in Y_p] \lesssim p^{-\alpha}.
		\ee
	\end{proposition}

	\Cref{fig:grid3} illustrates the result whose proof is provided in the supplement. For segments with horizontal orientation, the segments that accumulate the most amount of data are at the center of the upper and lower boundaries. On the other hand, the segments that accumulate the least amount of data lie at the center of the left and right boundaries. The data accumulation rates on segments with vertical orientation can be obtained by symmetry. %
	
	\begin{figure}[h]
		\centering
		\scalebox{0.7}{\begin{tikzpicture}
\draw[step=1cm,gray,very thin] (0,0) grid (6,6);
\draw node [right, blue, above] at (3.1,2.9) {\Large $p^{-1}$};
\draw[thick, ->,dashed, blue]   (0.6,3.2) -- (2.6,3.2);
\draw[thick, ->,dashed, blue]   (5.4,3.2) -- (3.4,3.2);
\draw[thick, ->,dashed, blue]   (5.4,6.2) -- (3.4,6.2);
\draw[thick, ->,dashed, blue]   (5.4,-0.35) -- (3.4,-0.35);
\draw[thick, ->,dashed, blue]   (0.6,-0.35) -- (2.6,-0.35);
\draw[thick, ->,dashed, blue]   (0.6,6.2) -- (2.6,6.2);
\draw[thick, ->,dashed, blue]   (3.2,3.6) -- (3.2,5.8);
\draw[thick, ->,dashed, blue]   (3.2,2.8) -- (3.2,0.2);
\draw[thick, ->,dashed, blue]   (-0.2,3.6) -- (-0.2,5.8);
\draw[thick, ->,dashed, blue]   (-0.2,2.8) -- (-0.2,0.2);
\draw[thick, ->,dashed, blue]   (6.2,3.6) -- (6.2,5.8);
\draw[thick, ->,dashed, blue]   (6.2,2.8) -- (6.2,0.2);

\draw node [right, blue, above] at (-0.1,2.9) {\Large $p^{-1-\alpha}$};
\draw node [right, blue, above] at (6.1,2.9) {\Large $p^{-1-\alpha}$};
\draw node [right, blue, above] at (3.1,5.9) {\Large $p^{-\alpha}$};
\draw node [right, blue, above] at (3.1,-0.7) {\Large $p^{-\alpha}$};
\draw node [right, blue, above] at (-0.1,5.9) {\Large $p^{-2\alpha}$};
\draw node [right, blue, above] at (6.1,5.9) {\Large $p^{-2\alpha}$};
\draw node [right, blue, above] at (6.1,-0.7) {\Large $p^{-2\alpha}$};
\draw node [right, blue, above] at (-0.1,-0.7) {\Large $p^{-2\alpha}$};
\end{tikzpicture}}
		\caption{Data accumulation rates of traversals on segments with horizontal movements in a grid. The dashed arrows represent directions toward which the data accumulation rates over the segment increase.}
		\label{fig:grid3}
	\end{figure}
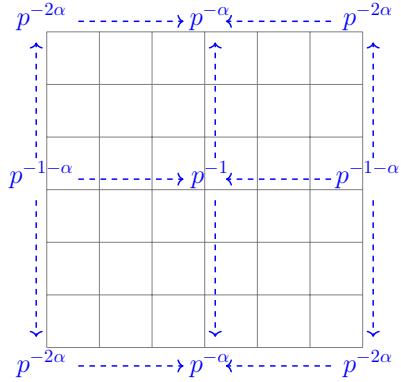

	\section{Proofs} \label{appx:proofs}
	
	\begin{proof}[Proof of  \Cref{thm:optimal_estimator}]
		In the proof, we use a notation that is common in multivariate statistics. When we want to denote that a multivariate normal random vector $Y$ has mean $\nu$ and covariance $\Sigma$, we will write
		\be
		Y \sim \mathcal{N}(\nu,\Sigma).
		\ee

  Recall that $\mathcal{E}\in\mathbb{R}^{M\times 1}$ is the vector of error terms of all segment travel times. Specifically, $\mathcal{E}_i$ is the error term of $Z_i$.
  With this notation, the entire data vector $Z$ satisfies
		\be
		Z \mid \theta \sim \mathcal{N}(U \theta,\Phi).
		\ee
		
		Let $b=\mu e$ and $B=\diag(\tau^2e)$. %
		\be
		\theta \sim \mathcal{N}(b,B).
		\ee
		
		First, notice that \emph{marginal of $\theta$}, we have
		\be
		Z \sim \mathcal{N}(Ub,UBU^\intercal+\Phi).
		\ee
		This is because $\mathbb E[Z] = \mathbb E[\mathbb E[Z \:|\:\theta]] = U \mathbb E[\theta] = U b$ and $\var[Z] = \mathbb \var[\mathbb E[Z \:|\: \theta]] + \mathbb E[\var[Z \:|\: \theta]]=UBU^\intercal+\Phi$. This means that jointly
		
		\be
		\begin{pmatrix} Z \\ \theta \end{pmatrix} \sim \mathcal{N}\left( \begin{pmatrix} Ub \\ b  \end{pmatrix}, \begin{pmatrix} UBU^\intercal + \Phi & A \\ A^\intercal & B \end{pmatrix} \right),
		\ee
		for some co-variance matrix $A$. This is because $(\mathcal{E},\theta)$ are jointly Gaussian, $(U\theta + \mathcal{E},\theta)^\top$, which is a linear transformation of $(\mathcal{E},\theta)^\top$, must also be jointly Gaussian.  Recall that, $Z=U\theta + \mathcal{E}$. Then, we have
  \begin{align}
  \notag
A=\textrm{cov}(Z,\theta)=\textrm{cov}(U\theta+\mathcal{E},\theta)=\textrm{cov}(U\theta,\theta)=UB,
  \end{align}
where the second to the last equality holds by the independence of $\theta$ and $\mathcal{E}$.

We now know the entire covariance structure for the joint distribution of $(Z,\theta)$. Then, the conditional distribution density $f(\theta\mid Z)$ should also be Gaussian. We have
\begin{align}
\notag
   &\log f(\theta\mid Z) \\[1mm] \notag
   =& \log \frac{f(Z\mid \theta) f(\theta)}{f(Z)} \\ \notag
   =& \log f(Z\mid \theta) + \log f(\theta) + c\\ \label{eq:expanding_1}
   =& -\frac{1}{2}(Z-U\theta)^\top \Phi^{-1}(Z-U\theta) -\frac{1}{2} (\theta-b)^\top B^{-1}(\theta-b) + c - \log\left(\sqrt{(2\pi)^M|\Phi|}\right) \\ \label{eq:expanding_2}
   =& -\frac{1}{2}(\theta - \mu_{\theta\mid Z})^\top \Sigma_{\theta|Z}^{-1} (\theta - \mu_{\theta\mid Z}) + c',
\end{align}
where $c$ and $c'$ are some constants, and in \cref{eq:expanding_2} we utilize the fact that the conditional distribution density must be Gaussian, and $\mu_{\theta\mid Z}$ is the conditional expectation and $\Sigma_{\theta|Z}$ is the conditional covariance matrix.

Expanding \cref{eq:expanding_1}, the terms related to $\theta$ read
\begin{align*}
-\frac{1}{2} \left(\theta^\top (U^\top \Phi^{-1}U + B^{-1})\theta - (Z^\top \Phi^{-1}U + b^\top B^{-1})\theta - \theta^\top (U^\top \Phi^{-1}Z + B^{-1}b)  \right).
\end{align*}
Similarly, expanding \cref{eq:expanding_2}, the terms related to $\theta$ read
\begin{align*}
-\frac{1}{2}\left(\theta^\top \Sigma_{\theta|Z}^{-1}\theta - (\mu_{\theta|Z}^\top \Sigma_{\theta|Z}^{-1})\theta - \theta^\top (\Sigma_{\theta|Z}^{-1}\mu_{\theta|Z})  \right).
\end{align*}

To make them equal, we must have
\begin{align*}
\Sigma_{\theta|Z}^{-1} = U^\top\Phi^{-1}U+B^{-1}, \quad \mu_{\theta \mid Z} =     \Sigma_{\theta|Z}(U^\top\Phi^{-1}Z + B^{-1}b).
\end{align*}
For brevity of the notation, we let $Q:=\Sigma_{\theta|Z}^{-1}$.

Because the estimator that minimizes the integrated risk based on squared error is the posterior mean \citep{berger2013statistical_sm},
		\be
		\hat{\Theta}^\ast_y = \mathbb E\left[\sum_{s\in y}\theta_s \:\Bigg|\: Z\right] = e_y^\intercal Q^{-1}(U^\intercal\Phi^{-1}Z+B^{-1}b).
		\ee

		\smallskip
		
		We now compute the integrated risk based on the squared error of the optimal estimator $\hat{\Theta}_y^\ast$. Recall that $E_y = e_y e_y^\intercal$. The integrated risk of $\hat{\Theta}_y^\ast$ conditional on $\theta$ is
		\be
		&\quad~\bb E\bigg[ \Big\| e_y^\intercal\theta - \hat{\Theta}_y^\ast \Big\|^2 \:\bigg|\: \theta \bigg] \\[1mm]
		&=\bb E\bigg[ \Big\| e_y^\intercal\theta - e_y^\intercal Q^{-1}\left(U^\intercal\Phi^{-1}Z+B^{-1}b\right) \Big\|^2 \:\bigg|\: \theta\bigg] \\[1mm]
		&= \bb E \bigg[\Big\| e_y^\intercal\left(\theta - Q^{-1} B^{-1} b\right) - e^\intercal_y Q^{-1} U^\intercal\Phi^{-1}Z\Big\|^2 \:\bigg|\: \theta\bigg] \\[1mm]
		&= \bb E\bigg[ \left(\theta - Q^{-1}B^{-1}b\right)^\intercal E_y\left(\theta - Q^{-1}B^{-1}b\right) -  2\left(\theta - Q^{-1}B^{-1}b\right)^\intercal E_yQ^{-1}U^\intercal\Phi^{-1}Z \\
		&\qquad~+ Z^\intercal\Phi^{-1}UQ^{-1}E_y Q^{-1}U^\intercal\Phi^{-1}Z \:\bigg|\: \theta \bigg]\\[1mm]
		&= \left(\theta - Q^{-1}B^{-1}b\right)^\intercal E_y\left(\theta - Q^{-1}B^{-1}b\right) - 2\left(\theta - Q^{-1}B^{-1}b\right)^\intercal E_yQ^{-1}U^\intercal\Phi^{-1} U\theta  \\[1mm]
		&\quad~ + \mathrm{tr}\left(\Phi^{-1} U Q^{-1}E_yQ^{-1} U^\intercal \Phi^{-1} \Phi\right) + \theta^\intercal U^\intercal \Phi^{-1} U Q^{-1}E_y Q^{-1} U^\intercal \Phi^{-1} U\theta \\[2mm]
		&= \theta^\intercal\left(E_y + U^\intercal\Phi^{-1}UQ^{-1}E_y Q^{-1}U^\intercal\Phi^{-1}U - 2 E_yQ^{-1}U^\intercal\Phi^{-1}U\right)\theta + 2 b^\intercal B^{-1}Q^{-1}E_yQ^{-1}U^\intercal\Phi^{-1}U\theta\\[1mm]
		&\quad~ +\mathrm{tr}\left(\Phi^{-1}UQ^{-1}E_yQ^{-1}U^\intercal\right) + b^\intercal B^{-1}Q^{-1}E_y Q^{-1}B^{-1}b - 2b^\intercal B^{-1}Q^{-1}E_y\theta.
		\ee

		Now we take an outer expectation over $\theta$ with respect to its prior to get the integrated risk,
		\be
		&\bb E\bigg[\bb E\bigg[ \Big\| e_y^\intercal\theta - \hat{\Theta}_y^\ast \Big\|^2 \:\bigg|\:\theta\bigg]\bigg] \\[1mm]
		=&\mathrm{tr}\left(B\left(E_y + U^\intercal\Phi^{-1}UQ^{-1}E_y Q^{-1} U^\intercal\Phi^{-1}U - 2 E_y Q^{-1}U^\intercal\Phi^{-1}U\right)\right) \\[1mm]
		&+ b^\intercal\left(E_y + U^\intercal\Phi^{-1}UQ^{-1}E_y Q^{-1} U^\intercal\Phi^{-1}U - 2 E_y Q^{-1}U^\intercal\Phi^{-1}U\right)b \\[1mm]
		&+ 2 b^\intercal B^{-1}Q^{-1}E_y Q^{-1}U^\intercal\Phi^{-1}Ub +\mathrm{tr}\left(\Phi^{-1}UQ^{-1}E_y Q^{-1}U^\intercal\right) \\[1mm]
		&+ b^\intercal B^{-1}Q^{-1}E_y Q^{-1}B^{-1}b - 2b^\intercal B^{-1}Q^{-1}E_yb\\
		=&\mathrm{tr}\left(B\left(E_y + U^\intercal\Phi^{-1}UQ^{-1}E_y Q^{-1} U^\intercal\Phi^{-1}U - 2 E_y Q^{-1}U^\intercal\Phi^{-1}U\right)\right) +\mathrm{tr}\left(\Phi^{-1}UQ^{-1}E_y Q^{-1}U^\intercal\right).
		\ee
		This completes the proof.
	\end{proof}
	
	\medskip
	
	\begin{proof}[Proof of \Cref{cor:independent_simple_form}]
		When $\sigma_{s,t}=0$,
		\be
		Q =&U^\intercal \Phi^{-1} U + \diag\left(\left(1/\tau^2\right)e\right) \\
		=&\diag\left( \left[N_s/\sigma_s^2\right]_{s\in\mathcal{S}}\right) + \diag\left(\left(1/\tau^2\right)e\right) \\
		=&\diag\left(\left[N_s/\sigma_s^2 + 1/\tau^2\right]_{s\in\mathcal{S}}\right).
		\ee
		
		Moreover,
		\be
		U^\intercal \Phi^{-1}z = \diag\Bigg( \left[\frac{\sum_{n:s\in y_n} T^\prime_{n,s}}{\sigma_s^2} \right]_{s\in\mathcal{S}}\Bigg).
		\ee
		
		This gives,
		\be
		\hat{\Theta}_y^\ast =& e^\intercal_y Q^{-1}\left(U^\intercal\Phi^{-1}z+\left(\mu/\tau^2\right)e\right)  \\[3mm]
		=&e_y^\intercal \diag\Bigg(\Big[\left(N_s/\sigma_s^2 + 1/\tau^2\right)^{-1}\Big]_{s\in\mathcal{S}}\Bigg)\diag\Bigg( \left[\sum_{n:s\in y_n} T^\prime_{n,s}/\sigma_s^2 + \mu/\tau^2 \right]_{s\in\mathcal{S}} \Bigg)\\
		=&e_y^\intercal \diag\Bigg( \Bigg[ \frac{\sum_{n:s\in y_n} T^\prime_{n,s}/\sigma_s^2 + \mu/\tau^2}{N_s/\sigma_s^2 + 1/\tau^2} \Bigg]_{s\in\mathcal{S}} \Bigg) \\
		=&\sum_{s\in y}\left( \frac{\sigma_s^2}{N_s\tau^2 + \sigma_s^2}\cdot\mu + \frac{N_s\tau^2}{N_s\tau^2 + \sigma_s^2}\cdot\frac{\sum_{n: s\in y_n} T_{n,s}^\prime}{N_s}\right).
		\ee
		This completes the proof.
	\end{proof}
	
	\medskip
	
	\begin{proof}[Proof of \Cref{prop:diff}]
		For the generalized segment-based estimator $\ThetasegG$,
		\be
		~&\mathbb{E}\left[\left(\ThetasegG - \sum_{s\in y}\theta_s\right)^2\,\middle\vert\,  y_{[N]}\right] \\
		=& \mathbb{E}\left[\; \mathbb{E}\left[\left(\ThetasegG-\sum_{s\in y}\theta_s\right)^2\,\middle\vert\, \{\theta_s\}_{s\in y},  y_{[N]}\right] \;\middle\vert\,  y_{[N]}\right]\\
		=& \mathbb{E}\left[  \var\left(\ThetasegG \,\middle\vert\,  \{\theta_s\}_{s\in y}, y_{[N]}  \right) + \Bias^2\left(\ThetasegG\Big| \{\theta_s\}_{s\in y},y_{[N]} \right)\,\middle\vert\,  y_{[N]}\right] \\
		=& \mathbb{E}\left[\var\left(\sum_{S\in \mathscr{S}_y}\phi_S(N_S)\frac{\sum_{n: S\subset y_n} \sum_{s\in S} T'_{n,s}}{N_S} \,\middle\vert\,  \{\theta_s\}_{s\in y}, y_{[N]} \right)  + \left(\sum_{S\in \mathscr{S}_y} (1-\phi_S(N_S))\left(|S|\mu - \sum_{s\in S}\theta_s\right) \right)^2 \,\middle\vert\, y_{[N]} \right] \\
		=&\mathbb{E}\left[\sum_{S,T\in\mathscr{S}_y}\frac{N_{S\cup T}}{N_S N_T}\phi_S(N_S)\phi_T(N_T)\left(\sum_{s\in S, t\in T} \sigma_{s,t}\right) + \left(\sum_{S\in \mathscr{S}_y} (1-\phi_S(N_S))\left(|S|\mu - \sum_{s\in S}\theta_s\right) \right)^2 \,\middle\vert\, y_{[N]} \right] \\
		=&\sum_{S,T\in\mathscr{S}_y}\frac{N_{S\cup T}}{N_S N_T}\phi_S(N_S)\phi_T(N_T)\left(\sum_{s\in S, t\in T} \sigma_{s,t}\right) + \mathbb{E}\left[\left(\sum_{S\in\mathscr{S}_y} (1-\phi_S(N_S))\left(|S|\mu - \sum_{s\in S}\theta_s\right) \right)^2 \,\middle\vert\, y_{[N]} \right] \\
		=&\sum_{S,T\in\mathscr{S}_y}\frac{N_{S\cup T}}{N_S N_T}\phi_S(N_S)\phi_T(N_T)\left(\sum_{s\in S, t\in T} \sigma_{s,t}\right) + \var\left(\sum_{S\in \mathscr{S}_y} (1-\phi_S(N_S))\left(|S|\mu - \sum_{s\in S}\theta_s\right) \,\middle\vert\, y_{[N]} \right) \\
		=&\sum_{S,T\in\mathscr{S}_y}\frac{N_{S\cup T}}{N_S N_T}\phi_S(N_S)\phi_{T}(N_T)\left(\sum_{s\in S, t\in T} \sigma_{s,t}\right)  + \sum_{S\in \mathscr{S}_y}(1-\phi_S(N_S))^2|S|\tau^2.
		\ee
		
		Similarly, for the route-based estimator $\Thetatrp$, to simplify the notation, we define $\Delta_{\delta(y)} = (\sum_{n:y_n\in\delta(y)}\sum_{s\in y_n}\theta_s - M_{\delta(y)}\sum_{s\in y}\theta_s ) /M_{\delta(y)}$. We have the following risk calculation. 
		\be
		~&\mathbb{E}\left[\left(\Thetatrp-\sum_{s\in y}\theta_s\right)^2\,\middle\vert\, y_{[N]}\right]\\
		=& \mathbb{E}\left[  \var\left(\Thetatrp \,\middle\vert\,  \{\theta_s\}_{s\in y}, y_{[N]}\right) + \Bias^2\left(\Thetatrp\,\middle\vert\,  \{ \{\theta_s\}_{s\in y}, y_{[N]}\}\right)\,\middle\vert\, y_{[N]}\right] \\
		=&\mathbb{E}\left[ \var\left(\phi_{\delta(y)}(M_{\delta(y)})\frac{\sum_{n : y_n\in\delta(y)} \sum_{s\in y} T'_{n,s}}{M_{\delta(y)}} \,\middle\vert\,  \{\theta_s\}_{s\in y}, y_{[N]} \right)\right. \\
		&+ \left.\left(\phi_{\delta(y)}(M_{\delta(y)}) \Delta_{\delta(y)} + \sum_{s\in y} (1-\phi_{\delta(y)}(M_{\delta(y)}))(\mu - \theta_s) \right)^2 \,\middle\vert\, y_{[N]} \right]  \\
		=&\mathbb{E}\left[\left(\frac{\phi_{\delta(y)}(M_{\delta(y)})}{M_{\delta(y)}}\right)^2\bigg(\sum_{n: y_n\in\delta(y)}\sum_{s,t\in y_n}\sigma_{s,t}\bigg)\right.\\
		& + \left. \left(\phi_{\delta(y)}(M_{\delta(y)})\Delta_{\delta(y)} + \sum_{s\in y} (1-\phi_{\delta(y)}(M_{\delta(y)}))(\mu - \theta_s) \right)^2 \,\middle\vert\, y_{[N]} \right]  \\
		=&\left(\frac{\phi_{\delta(y)}(M_{\delta(y)})}{M_{\delta(y)}}\right)^2\bigg(\sum_{n: y_n\in\delta(y)}\sum_{s,t\in y_n}\sigma_{s,t}\bigg) + \left(\phi_y\left(M_{\delta(y)}\right)\bb E\left[\Delta_{\delta(y)}\right]\right)^2 \\
		&+ \var\left(\phi_{\delta(y)}(M_{\delta(y)})\Delta_{\delta(y)} + \sum_{s\in y}(1-\phi_{\delta(y)}(M_{\delta(y)}))(\mu - \theta_s) \,\middle\vert\, y_{[N]}\right).
		\ee
		
		We further have,
		\be
		\bb E\left[ \Delta_{\delta(y)} \right] = \frac{\left(\sum_{n: y_n\in\delta(y)}|y_n|\right) - M_{\delta(y)}|y|}{M_{\delta(y)}}\mu = (\bar{y}_{\delta(y)} - |y|)\mu,
		\ee
		
		and,
		\be
		\var&\left( \sum_{s\in y}(1-\phi_{\delta(y)}(M_{\delta(y)}))(\mu - \theta_s) + \phi_{\delta(y)}(M_{\delta(y)})\Delta_{\delta(y)} \,\middle\vert\, y_{[N]}\right) \\
		=\var&\Bigg(|y|\left(1-\phi_{y}(M_{\delta(y)})\right)\mu - \left(1-\phi_{\delta(y)}(M_{\delta(y)})\right)\left(\sum_{s\in y}\theta_s\right) \\
		&- \phi_{\delta(y)}(M_{\delta(y)})\left(\sum_{s\in y}\theta_s\right) + \phi_{\delta(y)}(M_{\delta(y)})\frac{\sum_{n:y_n\in\delta(y)}\sum_{s\in y_n}\theta_s}{M_{\delta(y)}} \Bigg) \\
		=\var&\left(-\left(\sum_{s\in y}\theta_s\right) + \phi_{\delta(y)}(M_{\delta(y)})\frac{\sum_{n:y_n\in\delta(y)}\sum_{s\in y_n}\theta_s}{M_{\delta(y)}} \right)\\
		=\var&\left(-\left(\sum_{s\in y}\theta_s\right) + \phi_{\delta(y)}(M_{\delta(y)})\left(\sum_{s\in y}\frac{N_s^{\delta(y)}}{M_{\delta(y)}}\theta_s + \sum_{s\in\mathcal{S}_{\delta(y)}\setminus y}\frac{N_s^{\delta(y)}}{M_{\delta(y)}}\theta_s \right) \right) \\
		=\sum_{s\in y}&\left(1 - \phi_{\delta(y)}(M_{\delta(y)}) \frac{N_s^{\delta(y)}}{M_{\delta(y)}} \right)^2\tau^2 + \sum_{s\in \mathcal{S}_{\delta}\setminus y}\left(\phi_{\delta(y)}(M_{\delta(y)}) \frac{N_s^{\delta(y)}}{M_{\delta(y)}}\right)^2\tau^2.
		\ee
		
		Putting this all together leads to,
		
		\be
		R\left(\Thetatrp\,\middle\vert\,  y_{[N]}\right)=&%
		\left(\frac{\phi_{\delta(y)}(M_{\delta(y)})}{M_{\delta(y)}}\right)^2 \bigg(\sum_{n:y_n\in\delta(y)}\sum_{s, t\in y_n}  \sigma_{s,t}\bigg) + \left(\phi_{\delta(y)}(M_{\delta(y)})(\bar{y}_{\delta(y)} - |y|)\mu\right)^2 \\
		&+\sum_{s\in \mathcal{S}_{\delta(y)}\setminus y}\left( \phi_{\delta(y)}(M_{\delta(y)})\frac{N_s^{\delta(y)}}{M_{\delta(y)}}\right)^2\tau^2 + \sum_{s\in y}\left(1 - \phi_{\delta(y)}(M_{\delta(y)}) \frac{N_s^{\delta(y)}}{M_{\delta(y)}} \right)^2\tau^2\\
		=&%
		\left(\frac{\phi_{\delta(y)}(M_{\delta(y)})}{M_{\delta(y)}}\right)^2 \bigg(\sum_{s, t\in \mathcal{S}_{\delta(y)}} N_{s\cup t}^{\delta(y)} \sigma_{s,t}\bigg) + \left(\phi_{\delta(y)}(M_{\delta(y)})(\bar{y}_{\delta(y)} - |y|)\mu\right)^2 \\
		&+\sum_{s\in \mathcal{S}_{\delta(y)}\setminus y}\left( \phi_{\delta(y)}(M_{\delta(y)})\frac{N_s^{\delta(y)}}{M_{\delta(y)}}\right)^2\tau^2 + \sum_{s\in y}\left(1 - \phi_{\delta(y)}(M_{\delta(y)}) \frac{N_s^{\delta(y)}}{M_{\delta(y)}} \right)^2\tau^2.
		\ee
		
		This completes the proof.
	\end{proof}
	
	\medskip
	
	\begin{proof}[Proof of \Cref{prop:optimal_formula}]
		For any route $y$, 
		\be
		\phi^\ast_{\delta(y)}(M_{\delta(y)}) = \arg\min_{\phi_y(\cdot)} \mathbb{E}\left[\left(\Thetatrp-\sum_{s\in y}\theta_s\right)^2 \,\middle\vert\,  y_{[N]} \right]
		\ee
		is well defined by checking the first-order condition of \eqref{eq:trp_risk} as $\mathbb{E}\big[\big(\Thetatrp-\sum_{s\in y}\theta_s\big)^2 \:\big|\:  y_{[N]} \big]$ is strictly convex in $\phi_{\delta(y)}(M_{\delta(y)})$,
		\be
		&\frac{2\phi_{\delta(y)}^\ast(M_{\delta(y)})}{(M_{\delta(y)})^2}\left(\sum_{y_n\in\delta(y)}\sum_{s,t\in y_n}\sigma_{s,t}\right) + 2\left(\mu\left(\bar{y}_{\delta(y)} - |y| \right)\right)^2\phi_{\delta(y)}^\ast(M_{\delta(y)}) + \sum_{s\in\mathcal{S}_{\delta(y)}\setminus  y}2\tau^2\left(\frac{N_s^{\delta(y)}}{M_{\delta(y)}}\right)^2\phi_{\delta(y)}^\ast(M_{\delta(y)}) \\
		=& \sum_{s\in y} 2\tau^2\left( 1 - \phi_{\delta(y)}^\ast(M_{\delta(y)}) \frac{N_s^{\delta(y)}}{M_{\delta(y)}}\right)\frac{N_s^{\delta(y)}}{M_{\delta(y)}}.
		\ee
		This gives the optimal route-based estimator $\Thetatrpopt$,
		\be
		\Thetatrpopt:=(1-\phi_{\delta(y)}^\ast(M_{\delta(y)}&))|y|\theta + \phi^\ast_{\delta(y)}(M_{\delta(y)})\frac{\sum_{n : y_n = y} \sum_{s\in y} T'_{n,s}}{M_{\delta(y)}}, \\ 
		\phi_{\delta(y)}^\ast(M_{\delta(y)}) =\left(\sum_{s\in y} N_s^{\delta(y)} \right)\tau^2\Bigg/ \Bigg(&\sum_{s\in \mathcal{S}_{\delta(y)}} \frac{\left(N_s^{\delta(y)}\right)^2}{M_{\delta(y)}}\tau^2 + \frac{\sum_{n: y_n\in\delta(y)}\sum_{s,t\in y_n}\sigma_{s,t}}{M_{\delta(y)}} + M_{\delta(y)}\mu^2(\bar{y}_{\delta(y)} - |y|)^2\Bigg).
		\ee
		
		It can be checked that the Hessian of the integrated risk $\mathbb{E}\big[\big(\ThetasegG - \sum_{s\in y}\theta_s\big)^2 \:\big|\:  y_{[N]} \big]$ is symmetric and positive definite (PD). Suppose $|\mathscr{S}_y| = m$. Let $S_i,\: i\in\{1,\cdots, m\}$ be the $i^{\textrm{th}}$ super-segment in $\mathscr{S}_y$.
		\be
		\Hessian&\left(\mathbb{E}\left[\left(\ThetasegG-\sum_{s\in y}\theta_s\right)^2 \,\middle\vert\,  y_{[N]} \right]\right) \\
		=&   2\cdot\begin{bmatrix} 
			\frac{1}{N_{S_1}} \left(\sum_{s\in S_1, t\in S_1}\sigma_{s,t}\right) + |S_1|\tau^2 & \dots & \frac{N_{S_i\cup S_j}}{N_{S_i}N_{S_j}}\left(\sum_{s\in S_i, t\in S_j} \sigma_{s,t} \right) & \dots\\
			\vdots & \ddots & & \dots\\
			\frac{N_{S_i\cup S_j}}{N_{S_i}N_{S_j}}\left(\sum_{s\in S_i, t\in S_j} \sigma_{s,t} \right) &        & \ddots & \dots \\
			\vdots & \dots & \dots & \frac{1}{N_{S_m}}\left( \sum_{s\in S_m, t\in S_m}\sigma_{s,t}\right)+|S_m|\tau^2
		\end{bmatrix} \\
		=&2\cdot\begin{bmatrix} 
			\frac{1}{N_{S_1}} \left(\sum_{s\in S_1, t\in S_1}\sigma_{s,t}\right)  & \dots & \frac{N_{S_i\cup S_j}}{N_{S_i}N_{S_j}}\left(\sum_{s\in S_i, t\in S_j} \sigma_{s,t} \right) & \dots\\
			\vdots & \ddots & & \dots\\
			\frac{N_{S_i\cup S_j}}{N_{S_i}N_{S_j}}\left(\sum_{s\in S_i, t\in S_j} \sigma_{s,t} \right) &        & \ddots & \dots \\
			\vdots & \dots & \dots & \frac{1}{N_{S_m}}\left( \sum_{s\in S_m, t\in S_m}\sigma_{s,t}\right)
		\end{bmatrix} \\
		&+ 2\tau^2\cdot\begin{bmatrix} 
			|S_1| &  & \\
			& \ddots & \\
			& & |S_m|
		\end{bmatrix}.
		\ee
		
		The first matrix in the last equality is positive semidefinite (PSD). To see this, note that $1/N_{S_i} \ge N_{S_i\cup S_j}/\left(N_{S_i} N_{S_j}\right), \forall i,j\in\{1,\cdots,m\}$. Such scaling of the entries in a PSD matrix (the covariance matrix of super-segment travel times is PSD) results in a PSD matrix. Moreover, the second matrix in the last equality is positive definite (PD) because it is a diagonal matrix with strictly positive entries. As the sum of PSD and PD matrices is PD, this verifies the strict convexity of the integrated risk. The first-order conditions, presented in the statement, must admit a unique solution. This completes the proof. 
	\end{proof}
	
	\medskip
	
	\begin{proof}[Proof of \Cref{thm:opt_bayes_comparison}]
		
		The proof is by construction. Consider a segment-based estimator $\Thetasegprime$ with $ \phi_s(N_s)= \phi_{\delta(y)}^\ast(M_{\delta(y)})\cdot(N_s^{\delta(y)}/M_{\delta(y)})$ where $\phi_{\delta(y)}^\ast(\cdot)$ has a closed-form 
		as indicated in \Cref{prop:optimal_formula}. For any set of historical routes $y_{[N]}$, the integrated risk of this estimator is,
		\be
		&\mathbb{E}\left[\left(\Thetasegprime-\sum_{s\in y}\theta_s\right)^2\,\middle\vert\,  y_{[N]}\right] \\
		=&\sum_{s,t\in y}\frac{N_{s\cup t}}{N_s N_t}\frac{N_s^{\delta(y)} N_t^{\delta(y)}}{M_{\delta(y)}^2}
		\phi_{\delta(y)}^\ast(M_{\delta(y)})^2\sigma_{s,t}  + \sum_{s\in y}\left(1-\phi_{\delta(y)}^\ast(M_{\delta(y)})\cdot\frac{N_s^{\delta(y)}}{M_{\delta(y)}}\right)^2\tau^2 \\
		\le&\frac{\phi_{\delta(y)}^\ast(M_{\delta(y)})^2}{M_{\delta(y)}^2}
		\bigg(\sum_{s,t\in y}N_{s\cup t}^{\delta(y)}\sigma_{s,t}\bigg)  + \sum_{s\in y}\left(1-\phi_{\delta(y)}^\ast(M_{\delta(y)})\cdot\frac{N_s^{\delta(y)}}{M_{\delta(y)}}\right)^2\tau^2 \\
		\le& \frac{\phi_{\delta(y)}^\ast(M_{\delta(y)})^2}{M_{\delta(y)}^2}\bigg(\sum_{s, t\in \mathcal{S}_{\delta(y)}} N_{s\cup t}^{\delta(y)} \sigma_{s,t}\bigg)  + \sum_{s\in y}\left(1-\phi_{\delta(y)}^\ast(M_{\delta(y)})\cdot\frac{N_s^{\delta(y)}}{M_{\delta(y)}}\right)^2\tau^2 \\
		\le&	\left(\frac{\phi^\ast_{\delta(y)}(M_{\delta(y)})}{M_{\delta(y)}}\right)^2 \bigg(\sum_{s, t\in \mathcal{S}_{\delta(y)}} N_{s\cup t}^{\delta(y)} \sigma_{s,t}\bigg) + \left(\phi^\ast_{\delta(y)}(M_{\delta(y)})(\bar{y}_{\delta(y)} - |y|)\mu\right)^2 \\
		&+\sum_{s\in \mathcal{S}_{\delta(y)}\setminus y}\left(\phi^\ast_{\delta(y)}(M_{\delta(y)}) \frac{N_s^{\delta(y)}}{M_{\delta(y)}}\right)^2\tau^2 + \sum_{s\in y}\left(1 - \phi^\ast_{\delta(y)}(M_{\delta(y)}) \frac{N_s^{\delta(y)}}{M_{\delta(y)}} \right)^2\tau^2 \\
		=& \mathbb{E}\left[\left(\Thetatrpopt-\sum_{s\in y}\theta_s\right)^2 \,\middle\vert\,   y_{[N]}\right].
		\ee
		The first inequality uses the assumption that $N_{s\cup t} N_s^{\delta(y)}N_t^{\delta(y)}\le N_{s\cup t}^{\delta(y)}N_sN_t$. The second inequality uses $\sigma_{s,t}\ge0,\:\forall s,t\in\mathcal{S}$ and $\mathcal{S}_{\delta(y)}\supset y$. The third inequality holds because the additional two terms (the second and third terms) are both non-negative. The proof is then completed by 
		\be
		\mathbb{E}\left[\left(\Thetasegopt-\sum_{s\in y}\theta_s\right)^2 \,\middle\vert\,   y_{[N]}\right]\le&\mathbb{E}\left[\left(\Thetasegprime-\sum_{s\in y}\theta_s\right)^2\,\middle\vert\,   y_{[N]}\right] \le \mathbb{E}\left[\left(\Thetatrpopt-\sum_{s\in y}\theta_s\right)^2 \,\middle\vert\,   y_{[N]}\right].
		\ee
		
	\end{proof}
	
	\medskip
	
	\begin{proof}[Proof of \Cref{cor:same_route}]
		
		We first show that $R\big(\Thetasegopt \:\big|\: y_{[N]}\big) \le R\big(\ThetasegoptG \:\big|\: y_{[N]}\big)$.
		Consider a segment-based estimator $\Thetasegprime$ with $\phi_s(N_s)=\phi_{y}^\ast(N_y),\:\forall s\in y$. Note that here $\phi_y^\ast(\cdot)$ has a closed form $\phi_{y}^\ast(N_y) =N_y|y|\tau^2\big/\big(N_y|y|\tau^2 + \sum_{s,t\in y}\sigma_{s,t}\big)$ as $\mathscr{S}_y = \{\{y\}\}$.
		For any set of historical routes $y_{[N]}$, the integrated risk of this estimator is,
		\be
		&\mathbb{E}\left[\left(\Thetasegprime-\sum_{s\in y}\theta_s\right)^2\,\middle\vert\,  y_{[N]}\right] \\
		=&\sum_{s,t\in y}\frac{N_{s\cup t}}{N_s N_t}\phi_{y}^\ast(N_y)^2\sigma_{s,t}  + \sum_{s\in y}(1-\phi_{y}^\ast(N_y))^2\tau^2 \\[2mm] \label{eq:thm1_1}
		\le& \sum_{s,t\in y}\frac{1}{N_y}\phi_{y}^\ast(N_y)^2\sigma_{s,t}  + |y|(1-\phi_{y}^\ast(N_y))^2\tau^2 \\ 
		=& \mathbb{E}\left[\left(\ThetasegoptG-\sum_{s\in y}\theta_s\right)^2 \,\middle\vert\,   y_{[N]}\right].
		\ee
		
		Inequality \eqref{eq:thm1_1} holds because $N_{y}\le N_s$ and $N_{s\cup t}\le N_t$ and the assumption that $\sigma_{s,t}\ge0$. These imply $\sigma_{s,t}N_{s\cup t}/(N_s N_t)\le \sigma_{s,t}/N_s\le \sigma_{s,t}/N_y$. The proof of the inequality is then completed by
		\be
		\mathbb{E}\left[\left(\Thetasegopt-\sum_{s\in y}\theta_s\right)^2 \,\middle\vert\,y_{[N]}\right]\le&\mathbb{E}\left[\left(\Thetasegprime-\sum_{s\in y}\theta_s\right)^2\,\middle\vert\,   y_{[N]}\right] 
		\le  \mathbb{E}\left[\left(\ThetasegoptG-\sum_{s\in y}\theta_s\right)^2 \,\middle\vert\,y_{[N]}\right].
		\ee
		
		\smallskip
		
		We then prove that $R\big(\ThetasegoptG \:|\:  y_{[N]}\big) \le R\big(\Thetatrpopt \:|\:  y_{[N]}\big)$. When $\mathscr{S}_y=\{\{y\}\}$ and $\delta(y) = \{y_n: y_n = y\}$, $\phi_y^\ast(\cdot)$ and $\phi_{\delta(y)}^\ast(\cdot)$ have the same form,
		\be
		\phi_{y}^\ast(N_y) =\frac{N_y|y|\tau^2}{N_y|y|\tau^2 + \sum_{s,t\in y}\sigma_{s,t}},\quad\phi_{\delta(y)}^\ast(M_{\delta(y)}) =\frac{M_{\delta(y)}|y|\tau^2}{M_{\delta(y)}|y|\tau^2 + \sum_{s,t\in y}\sigma_{s,t}}.
		\ee
		
		The optimal integrated risks also share the same form,
		\be
		R\left(\ThetasegoptG \,\middle\vert\,   y_{[N]}\right) =& \frac{1}{N_y} (\phi_y^\ast(N_y))^2\left(\sum_{s,t\in y} \sigma_{s,t}\right) + (1-\phi^\ast_y(N_y))^2|y|\tau^2,\\
		R\left(\Thetatrpopt \,\middle\vert\,   y_{[N]}\right) =& \frac{1}{M_{\delta(y)}} (\phi_{\delta(y)}^\ast(M_{\delta(y)}))^2\left(\sum_{s,t\in y} \sigma_{s,t}\right) + (1-\phi^\ast_{\delta(y)}(M_{\delta(y)}))^2|y|\tau^2.
		\ee
		
		Because $N_y\ge M_{\delta(y)}$ and the optimal integrated risk decreases with sample size, we thus have 
		\be
		R\left(\ThetasegoptG \,\middle\vert\,   y_{[N]}\right)\le R\left(\Thetatrpopt \,\middle\vert\,   y_{[N]}\right).
		\ee
		This completes the proof.
	\end{proof}
	
	\medskip
	
	\begin{proof}[Proof of \Cref{thm:asymptotic_general_seg_better_route}]
		Under a given road network size $p$, let indicator variable $I_s$ denote whether a road segment $s$ is traversed from a randomly sampled route $Y_p\sim \mu_p$. For the segment-based estimator, given a predicting route that covers road segments indicated by $I = \{I_s\}_{s\in\mathcal{S}_p}$, 
		\be
		R\left(\ThetasegP\:\Big|\:I \right) =& \sum_{s,t\in \mathcal{S}_p} \bb E\left[\frac{N_{s\cup t}}{N_s N_{t}}\phi_s(N_s)\phi_{t}(N_{t}) \right] I_s I_t \sigma_{s,t}  + \sum_{s\in \mathcal{S}_p}\bb E\left[(1-\phi_s(N_s))^2 \right]I_s\tau^2 \\ 
		\le&\sum_{s,t\in \mathcal{S}_p}\bb E\left[\frac{N_{s\cup t}}{N_s N_t}\1\{N_s, N_t>0\}\right]I_s I_t\max\{\sigma_{s,t},0\} + \sum_{s\in \mathcal{S}_p}\bb E\left[(1-\phi_s(N_s))^2\right]I_s\tau^2 \\ \label{ineq:temp}
		\le&\sum_{s,t\in \mathcal{S}_p}\bb E\left[\frac{N_{s\cup t}}{N_s N_t}\1\{N_s, N_t>0\}\right]I_s I_t|\sigma_{s,t}| + \sum_{s\in \mathcal{S}_p}\bb E\left[(1-\phi_s(N_s))^2\right]I_s\tau^2.
		\ee
		
		Note that $\bb E[N_s] = Nq_s$. By Chernoff bound, for any $\beta>0$, $\bb P(N_s\le(1-\beta) \bb E[N_s]) = \bb P(N_s\le(1-\beta) Nq_s)\le e^{-\beta^2Nq_s/2}$. This yields,
		\be
		&\bb E\left[(1-\phi_s(N_s))^2\right] \\[1mm]
		=&\bb E\left[ (1-\phi_s(N_s))^2 \mid N_s\le (1-\beta) Nq_s \right] \bb P\left(N_s\le (1-\beta) Nq_s\right)\\
		&+ \bb E\left[ (1-\phi_s(N_s))^2 \mid N_s> (1-\beta) Nq_s \right]\bb P\left(N_s> (1-\beta) Nq_s\right) \\[1mm]
		\le&\bb P\left(N_s\le (1-\beta) Nq_s\right)+\bb E\left[ (1-\phi_s(N_s))^2 \mid N_s> (1-\beta) Nq_s \right] \\[1mm] 
		=& \bb P\left(N_s\le (1-\beta) Nq_s\right)+\bb E\left[ \mathcal{O}(1/N_s)\mid N_s> (1-\beta) Nq_s \right] \\[1mm] 
		\le& \bb P\left(N_s\le (1-\beta) Nq_s\right) + \mathcal{O}(1/((1-\beta)Nq_s)) \\[1mm] \label{ineq:temp2}
		=& \mathcal{O}(1/(Nq_s)).
		\ee
		
		The second equality holds because $(1-\phi_s(N_s)) = \mathcal{O}(1/\sqrt{N_s})$ and the last equality holds because $\bb P(N_s\le(1-\beta) Nq_s)\le e^{-\beta^2Nq_s/2}=\mathcal{O}(1/(Nq_s))$. Using this observation, 
		\be
		&~\eqref{ineq:temp}=\sum_{s,t\in \mathcal{S}_p}\bb E\left[\frac{N_{s\cup t}}{N_s N_t}\1\{N_s, N_t>0\}\right]I_sI_t|\sigma_{s,t}| +\sum_{s\in \mathcal{S}_p}\mathcal{O}(1/(Nq_s))I_s\tau^2 \\
		&~~\quad\qquad\le\sum_{s,t\in \mathcal{S}_p}\bb E\left[\frac{1}{N_{s\cup t}}\1\{N_{s\cup t}>0\}\right]I_sI_t|\sigma_{s,t}| + \sum_{s\in \mathcal{S}_p}\mathcal{O}(1/(Nq_s))I_s\tau^2 \\
		&~\quad\stackrel{(\Cref{lm:SumBound})}{\le}\sum_{s,t\in \mathcal{S}_p}\frac{2}{\bb E[N_{s\cup t}]}I_sI_t|\sigma_{s,t}| + \sum_{s\in \mathcal{S}_p}\mathcal{O}(1/(Nq_s))I_s\tau^2.
		\ee
		The second inequality holds as $N_{s\cup t}\le N_s$ and $N_{s\cup t}\le N_t$ for any $s,t\in y_p$. The third equality holds by \Cref{lm:SumBound} (\Cref{apx:lemmas}) which implies that $\bb E\left[1/N_s\1\{N_s>0\}\right] < 2/(q_sN)$. 
		
		\smallskip
		
		Taking the expectation over $I=\{I_s\}_{s\in\mathcal{S}_p}$,
		\be
		R\left(\ThetasegP \right) \le \sum_{s,t\in \mathcal{S}_p}\frac{2}{\bb E[N_{s\cup t}]}\bb E\left[I_sI_t\right]|\sigma_{s,t}| + \sum_{s\in \mathcal{S}_p}\mathcal{O}(1/(Nq_s))\bb E\left[I_s\right]\tau^2.
		\ee
		
		Note that $\bb E[I_sI_t] = \bb P(I_s=1, I_t=1) = \bb E[N_{s\cup t}]/N$ and $\bb E[I_s]= \bb E[N_s]/N$. This gives,
		\be
		R\left(\ThetasegP \right) \le& \sum_{s,t\in \mathcal{S}_p}\frac{2}{\bb E[N_{s\cup t}]}\bb E\left[I_sI_t\right]|\sigma_{s,t}| + \sum_{s\in \mathcal{S}_p}\mathcal{O}(1/(Nq_s))\bb E\left[I_s\right]\tau^2\\
		=&\sum_{s,t\in \mathcal{S}_p}\frac{2}{N}\bb |\sigma_{s,t}| + \sum_{s\in \mathcal{S}_p}\mathcal{O}(1/N)\tau^2\\[1mm] \label{eq:seg_ub}
		=&\mathcal{O}(|\mathcal{S}_p|/N).
		\ee
		
		\medskip
		
		The last equality uses the second part of Assumption \ref{assumption:3}. We now focus on the integrated risk of the optimal route-based estimator $\ThetatrpoptP$. conditional on the predicting route $Y_p$,
		\be
		R\left(\ThetatrpoptP\:\bigg|\: Y_p\right)=&%
		\bb E\left[\left(\frac{\phi^\ast_{Y_p}(M_{\delta(Y_p)})}{M_{\delta(Y_p)}}\right)^2 \sum_{n:y_n\in\delta(Y_p)}\sum_{s, t\in y_n}  \sigma_{s,t}\right] + \bb E\left[\left(\phi^\ast_{Y_p}(M_{\delta(Y_p)})(\bar{y}_{\delta(Y_p)} - |Y_p|)\mu\right)^2\right] \\ 
		&+\sum_{s\in \delta(Y_p)\setminus Y_p}\bb E\left[\left( \frac{N_s^{\delta(Y_p)}}{M_{\delta(Y_p)}}\right)^2\tau^2\right] + \sum_{s\in Y_p}\bb E\left[\left(1 - \phi^\ast_{Y_p}(M_{\delta(Y_p)}) \frac{N_s^{\delta(Y_p)}}{M_{\delta(Y_p)}} \right)^2\tau^2\right] \\
		\ge&\bb E\left[\left(\frac{\phi^\ast_{Y_p}(M_{\delta(Y_p)})}{M_{\delta(Y_p)}}\right)^2 M_{\delta(Y_p)}\sigma_{\textrm{min}}^2\right] + \sum_{s\in Y_p}\bb E\left[\left(1 - \phi^\ast_{Y_p}(M_{\delta(Y_p)}) \right)^2\tau^2\right] \\ \label{eq:intermediate}
		=&\bb E\left[\frac{\phi^\ast_{Y_p}(M_{\delta(Y_p)})^2}{M_{\delta(Y_p)}}\right]\sigma_{\textrm{min}}^2 + \sum_{s\in Y_p}\bb E\left[\left(1 - \phi^\ast_{Y_p}(M_{\delta(Y_p)}) \right)^2\tau^2\right]. 
		\ee

		The first inequality holds by the second part of Assumption \ref{assumption:3} and $N_s^{\delta(Y_p)}\le M_{\delta(Y_p)},\forall s\in Y_p$. We let $\phi^{\ast\ast}_{Y_p}(M_{\delta(Y_p)}) := M_{\delta(Y_p)}|Y_p|\tau^2 / (M_{\delta(Y_p)}|Y_p|\tau^2 + \sigma^2_{\textrm{min}})$ which minimizes
		\be
		\frac{\phi_{Y_p}(M_{\delta(Y_p)})^2}{M_{\delta(Y_p)}} \sigma^2_{\textrm{min}} + \sum_{s\in Y_p}\left(1-\phi_{Y_p}(M_{\delta(Y_p)})\right)^2\tau^2,
		\ee
		for any realization of $\delta(Y_p)$. This yields,
		
		\be
		\eqref{eq:intermediate}\ge&\bb E\left[\frac{\phi^{\ast\ast}_{Y_p}(M_{\delta(Y_p)})^2}{M_{\delta(Y_p)}}\right]\sigma_{\textrm{min}}^2 + \sum_{s\in Y_p}\bb E\left[\left(1 - \phi^{\ast\ast}_{Y_p}(M_{\delta(Y_p)})  \right)^2\tau^2\right].
		\ee
		
		We now consider two scenarios. First, when $M_{\delta(Y_p)}\ge1$,
		\be
		&\bb E\left[\frac{\phi^{\ast\ast}_{Y_p}(M_{\delta(Y_p)})^2}{M_{\delta(Y_p)}}\right]\sigma_{\textrm{min}}^2 + \sum_{s\in Y_p}\bb E\left[\left(1 - \phi^{\ast\ast}_{Y_p}(M_{\delta(Y_p)})  \right)^2\tau^2\right]\\
		\ge&\bb E\left[\frac{\phi^{\ast\ast}_{Y_p}(M_{\delta(Y_p)})^2}{M_{\delta(Y_p)}}\right]\sigma_{\textrm{min}}^2 \\
		\ge&\phi_{Y_p}^{\ast\ast}(1)\frac{\phi^{\ast\ast}_{Y_p}\left(\bb E\left[M_{\delta(Y_p)} \right]\right)}{\bb E\left[M_{\delta(Y_p)} \right]}\sigma^2_{\textrm{min}} \\
		=&\frac{|Y_p|\tau^2}{|Y_p|\tau^2 + \sigma^2_{\textrm{min}}} \cdot\frac{|Y_p|\tau^2}{\bb E\left[M_{\delta(Y_p)} \right]|Y_p|\tau^2 + \sigma^2_{\textrm{min}}}\cdot \sigma^2_{\textrm{min}} \\
		\ge&\frac{\tau^2}{\tau^2 + \sigma^2_{\textrm{min}}} \cdot\frac{\tau^2}{\bb E\left[M_{\delta(Y_p)} \right]\tau^2 + \sigma^2_{\textrm{min}}} \cdot\sigma^2_{\textrm{min}}.
		\ee
		
		The second inequality holds because $\phi^{\ast\ast}_{y_p}(\cdot)$ is non-decreasing so that $\phi^{\ast\ast}_{y_p}(M_{\delta(Y_p)})\ge \phi^{\ast\ast}_{y_p}(1)$ and $\mathbb E\big[\phi_{Y_p}^{\ast\ast}(M_{\delta(Y_p)}) \big/M_{\delta(Y_p)} \big] \ge \phi_{Y_p}^{\ast\ast}(\mathbb E[M_{\delta(Y_p)}]) \big/\mathbb E[M_{\delta(Y_p)}]$ by Jensen's inequality because the term $\phi^{\ast\ast}_{Y_p}(M_{\delta(Y_p)})/M_{\delta(Y_p)}$ is convex in $M_{\delta(Y_p)}$.
		The last inequality holds as $|Y_p|\ge1$. Note that the final term $(\tau^2/(\tau^2 + \sigma^2_{\textrm{min}}))\cdot(\tau^2/(\bb E\left[M_{\delta(Y_p)} \right]\tau^2 + \sigma^2_{\textrm{min}}))\cdot \sigma^2_{\textrm{min}}\le\tau^2$. 
		
		\medskip
		
		On the other hand, when $M_{\delta(Y_p)}=0$,
		\be
		&\bb E\left[\frac{\phi^{\ast\ast}_{Y_p}(M_{\delta(Y_p)})^2}{M_{\delta(Y_p)}}\right]\sigma_{\textrm{min}}^2 + \sum_{s\in Y_p}\bb E\left[\left(1 - \phi^{\ast\ast}_{Y_p}(M_{\delta(Y_p)})  \right)^2\tau^2\right]=|Y_p|\tau^2\ge\tau^2.  
		\ee
		
		This suggests that for all $Y_p$, 
		\be
		\eqref{eq:intermediate}\ge\frac{\tau^2}{\tau^2 + \sigma^2_{\textrm{min}}} \cdot\frac{\tau^2}{\bb E\left[M_{\delta(Y_p)} \right]\tau^2 + \sigma^2_{\textrm{min}}} \cdot\sigma^2_{\textrm{min}}.
		\ee
		
		By taking the expectation over $Y_p$,
		\be
		\label{eq:route_lb}
		R\left(\ThetatrpoptP\right) \ge \frac{\tau^2}{\tau^2 + \sigma^2_{\textrm{min}}} \cdot\frac{\tau^2}{Nq_{\delta}\tau^2 + \sigma^2_{\textrm{min}}} \cdot\sigma^2_{\textrm{min}}=\Omega\left(1/\left(Nq_{\delta}\right)\right).
		\ee
		
		Based on \eqref{eq:seg_ub} and \eqref{eq:route_lb}, we have that when $q_{\delta}=o(1/|\mathcal{S}_p|)$, 
		\be
		\lim_{p\to\infty}\frac{R\left(\ThetasegP\right)}{R\left(\ThetatrpoptP\right)} = 0.
		\ee
		
		This completes the proof.
	\end{proof}
	
	\medskip
	
	\begin{proof}[Proof of \Cref{lm:NodeLowerBound}]
		We first derive the lower bounds. For even $p$, the probability mass function (PMF) of the symmetric beta-binomial distribution is symmetric about $p/2$, and has a minimum value on its support at $p/2$. For simplicity and without loss of generality, we will assume $p$ is even in this proof. The odd case can be proven with some minor modifications. We have,
		\be
		\bb P[x = (p/2,\cdot)] &= {p \choose p/2} \frac{B(\alpha + p/2,\alpha+p/2)}{B(\alpha,\alpha)} \\
		&= \frac{1}{B(\alpha,\alpha)} \frac{\Gamma(p+1)}{\Gamma(p/2+1)\Gamma(p/2+1)} \cdot\frac{\Gamma(p/2+\alpha)\Gamma(p/2+\alpha)}{\Gamma(p+2 \alpha)},
		\ee
		where $B(\cdot, \cdot)$ is Beta function and $\Gamma(\cdot)$ is Gamma function. 
		
		\smallskip
		
		Define by
		\be
		f(p) := \frac{\Gamma(p+1)}{\Gamma(p/2+1)\Gamma(p/2+1)}  \frac{\Gamma(p/2+\alpha)\Gamma(p/2+\alpha)}{\Gamma(p+2 \alpha)}.
		\ee
		By Gautschi's inequality \citep[Eq. 5.6.4]{NIST:DLMF}, we have,
		\be
		x^{1-\beta}\le &\frac{\Gamma\left(x+1\right)}{\Gamma\left(x+\beta\right)}\le(x+1)^{1-\beta},~~0 < \beta \le 1; \\
		(x+1)^{1-\beta}\le&\frac{\Gamma\left(x+1\right)}{\Gamma\left(x+\beta\right)}\le x^{1-\beta},\qquad\quad 1 < \beta \le 2.
		\ee
		
		\begin{itemize}
			\item Under the case that $0<\alpha\le 1/2$, it follows that
			\be
			\frac{\Gamma(p+1)}{\Gamma(p+2\alpha)} &\ge p^{1-2\alpha}, \\
			\frac{\Gamma(p/2+1)}{\Gamma(p/2+\alpha)} &\le (p/2+1)^{1-\alpha},
			\ee
			and so
			\be
			\frac{\Gamma(p/2+\alpha)}{\Gamma(p/2+1)} \ge (p/2+1)^{\alpha-1}. 
			\ee
			It follows that
			\be
			f(p) &\ge (p/2+1)^{2\alpha - 2} p^{1-2\alpha} \\
			&= 2^{2-2\alpha} \frac{p}{(p+2)^2} \left( \frac{p+2}{p} \right)^{2 \alpha}
			\\
			&\ge 2^{2-2\alpha} \frac{p}{(p+2)^2} \\
			&\ge \frac19 2^{2-2\alpha} p^{-1},
			\ee
			and therefore
			\be
			\bb P[x = (i, \cdot)] \ge \bb P[x = (p/2, \cdot)] \ge \frac{4^{1-\alpha }}{9B(\alpha,\alpha)} p^{-1}.
			\ee
			
			This gives,
			\be
			\bb P[x = (i, j)] \ge \bb P[x = (p/2, p/2)] \ge \frac{4^{2-2\alpha}}{81B(\alpha,\alpha)^2} p^{-2}.
			\ee

			\item Under the case that $1/2<\alpha \le 1$, it follows that,
			\be
			\frac{\Gamma(p+1)}{\Gamma(p+2\alpha)} &\ge (p+1)^{1-2\alpha}, \\
			\frac{\Gamma(p/2+1)}{\Gamma(p/2+\alpha)} &\le (p/2+1)^{1-\alpha},
			\ee
			and so
			\be
			\frac{\Gamma(p/2+\alpha)}{\Gamma(p/2+1)} \ge (p/2+1)^{\alpha-1}. 
			\ee
			
			It follows that
			\be
			f(p) &\ge (p/2+1)^{2\alpha - 2} (p+1)^{1-2\alpha} \\[3mm]
			&= 2^{2-2\alpha} (p+2)^{2\alpha -2} (p+1)^{1-2\alpha}
			\\[1mm]
			&= 2^{2-2\alpha} \frac{(p+2)^{2\alpha -2}}{(p+1)^{2\alpha -1}} \\
			&= 2^{2-2\alpha} \frac{(p+2)^{2\alpha -2}}{(p+1)^{2\alpha -2}} \frac{1}{p+1} \\
			&> 2^{2-2\alpha} \frac1{p+1} \\
			&\ge 2^{2-2\alpha} \frac1{2p} = 2^{1-2\alpha}p^{-1},
			\ee
			
			and therefore
			\be
			\bb P[x = (i, \cdot)] \ge \bb P[x = (p/2, \cdot)] \ge \frac{2^{1-2\alpha }}{B(\alpha,\alpha)} p^{-1}.
			\ee
			
			This gives,
			\be
			\bb P[x = (i, j)] \ge \bb P[x = (p/2, p/2)] \ge \frac{4^{1-2\alpha}}{B(\alpha,\alpha)^2} p^{-2}.
			\ee
		\end{itemize}
		
		We note that using the other side of Gautschi's inequality, we can also show that $\bb P[x=(p/2,p/2)]\lesssim p^{-2}$. This gives $\bb P[x=(p/2,p/2)]\simeq p^{-2}$.
		
		\medskip
		
		We then derive the upper bounds. The PMF of the symmetric beta-binomial distribution has a maximum value  on its support at either $0$ or $p$. Without loss of generality, we select the maximum at $0$.
		\be
		\bb P[x = (0,\cdot)] &= {p \choose 0} \frac{B(p+\alpha,\alpha)}{B(\alpha,\alpha)} \\
		&= \frac{\Gamma(\alpha)}{B(\alpha,\alpha)}\cdot\frac{\Gamma(p+\alpha)}{\Gamma(p+2 \alpha)}.
		\ee
		
		Similarly, we now look at two cases. 
		
		\smallskip
		
		\begin{itemize}
			\item Under the case that $0<\alpha\le 1/2$, we have
			\be
			\frac{\Gamma(p+\alpha)}{\Gamma(p+2\alpha)} =& \frac{\Gamma(p+\alpha)}{\Gamma(p+1)}\cdot\frac{\Gamma(p+1)}{\Gamma(p+2\alpha)} \\
			\le&p^{\alpha-1}(p+1)^{1-2\alpha} \\[1mm]
			\le&p^{-\alpha},
			\ee
			and therefore
			\be
			\bb P[x=(i,\cdot)]\le \bb P[x=(0,\cdot)]\le\frac{\Gamma(\alpha)}{B(\alpha, \alpha)}p^{-\alpha}.
			\ee
			
			This gives,
			\be
			\bb P[x=(i,j)]\le \bb P[x=(0,0)] \le \frac{\Gamma(\alpha)^2}{B(\alpha, \alpha)^2}p^{-2\alpha}.
			\ee
			
			\item Under the case that $1/2<\alpha\le 1$, we have
			\be
			\frac{\Gamma(p+\alpha)}{\Gamma(p+2\alpha)} =& \frac{\Gamma(p+\alpha)}{\Gamma(p+1)}\cdot\frac{\Gamma(p+1)}{\Gamma(p+2\alpha)} \\
			\le&p^{\alpha-1}p^{1-2\alpha} \\
			=&p^{-\alpha}.
			\ee
			
			Similarly, this gives,
			\be
			\bb P[x=(i,j)]\le \bb P[x=(0,0)] \le \frac{\Gamma(\alpha)^2}{B(\alpha, \alpha)^2}p^{-2\alpha}.
			\ee
		\end{itemize}
		Using the other side of Gautschi's inequality, we can show that $\bb P[x = (0,0)]\gtrsim p^{-2\alpha}$. By symmetry, we have $\bb P[x = (0,0)]=\bb P[x=(p,0)]=\bb P[x=(0,p)]=\bb P[x=(p,p)] \simeq p^{-2\alpha}$. This completes the proof.
	\end{proof}
	
	\medskip
	
	\begin{proof}[Proof of \Cref{prop:q_delta_od}]
		Consider a particular neighborhood near a predicting route $y_p$, $\delta^{\textrm{od}\ast}(y_p) = \{y\in \mathcal{Y}_p: \|x_1(y), x_1(y_p)\|_1=0,  \|x_2(y), x_2(y_p)\|_1=0\}$. In words, neighborhood $\delta^{\textrm{od}\ast}(y_p)$ includes historical routes that have the same origin and destination as those of route $y_p$. It is clear that for any other route neighborhood $\delta^{\textrm{od}}(\cdot)$ with $\delta^{\textrm{od}}(y_p) = \{y\in \mathcal{Y}_p: \|x_1(y), x_1(y_p)\|_1\le c,  \|x_2(y), x_2(y_p)\|_1\le c\}$ for some constant $c>0$ that does not depend on $p$,
		\be
		q_{\delta^{\textrm{od}}(y_p)} \simeq q_{\delta^{\textrm{od}\ast}(y_p)}.
		\ee
		We thus focus on analyzing $q_{\delta^{\textrm{od}\ast}} = \bb P_{Y_p\sim\mu_p, Y_p^\prime\sim\mu_p}[Y_p \in \delta^{\textrm{od}\ast}(Y_p^\prime)]=\sum_{y\in\mathcal{Y}_p}q_{\delta^{\textrm{od}\ast}(y)}\mathbb P[Y^\prime_p = y]$ instead. %
		\be
		q_{\delta^{\textrm{od}\ast}} =& \sum_{y\in\mathcal{Y}_p}q_{\delta^{\textrm{od}\ast}(y)}\mathbb P[Y^\prime_p = y] \\
		=&\sum_{x_1\in\mathcal{V}_p, x_2\in\mathcal{V}_p}\mathbb P\left[x_1(Y_p) = x_1, x_2(Y_p)=x_2\right]\cdot \mathbb P\left[x_1(Y^\prime_p) = x_1, x_2(Y^\prime_p)=x_2\right]\\
		=&\sum_{x_1\in\mathcal{V}_p,x_2\in\mathcal{V}_p} \bb P^2\left[x_1(Y_p)=x_1, x_2(Y_p)=x_2\right] \\
		=&\sum_{x_1\in\mathcal{V}_p,x_2\in\mathcal{V}_p}\bb P^2\left[x_1(Y_p) = x_1 \right] \bb P^2\left[x_2(Y_p) = x_2 \right]\\
		=&\sum_{i,j,l,m\in\{0,\cdots,p\}}\bb P^2\left[x_1(Y_p) = (i,\cdot) \right]\bb P^2\left[x_1(Y_p) = (\cdot,j)\right]\bb P^2\left[x_2(Y_p) = (l,\cdot)\right]\bb P^2\left[x_2(Y_p) = (\cdot,m)\right] \\ \label{eq:quad}
		=&\left(\sum_{i\in\{0,\cdots,p\}}\bb P^2\left[x_1(Y_p) = (i,\cdot) \right] \right)^4.
		\ee
		
		The third equality holds because the sampling processes of origins and destinations are independent. Similarly, the fourth equality holds because the sampling processes of horizontal and vertical coordinates are independent. 
		
		\smallskip
		
		For any $i\in\{0,\cdots,p-1\}$,
		\be
		\frac{\bb P\left[x_1(Y_p)=(i+1,\cdot)\right]}{\bb P\left[x_1(Y_p)=(i,\cdot)\right]} =& \frac{{p \choose i+1}B(i+\alpha,\: p-i+\alpha)/B(\alpha,\alpha)}{{p \choose i}B(i+\alpha, p-i+\alpha)/B(\alpha,\alpha)} \\
		=&\frac{{p \choose i+1}B(i+\alpha,\: p-i+\alpha)}{{p \choose i}B(i+\alpha, n-i+\alpha)} \\
		=&\frac{p-i}{i+1}\cdot\frac{\Gamma(i+1+\alpha)\Gamma(p-i-1+\alpha)/\Gamma(p+2\alpha)}{\Gamma(i+\alpha)\Gamma(p-i+\alpha)/\Gamma(p+2\alpha)} \\
		=&\frac{p-i}{i+1}\cdot\frac{\Gamma(i+1+\alpha)\Gamma(p-i-1+\alpha)}{\Gamma(i+\alpha)\Gamma(p-i+\alpha)} \\
		=&\frac{p-i}{i+1}\cdot(i+\alpha)\cdot\frac{1}{p-i-1+\alpha}\\ \label{eq:recursion}
		=&\frac{i+\alpha}{i+1}\cdot\frac{p-i}{p-i-1+\alpha}. 
		\ee 
		
		The second-to-last equation holds by using the fact that $\Gamma(z+1)=z\Gamma(z)$. Let $g(j)=\prod_{i=1}^j \left(\frac{i+\alpha}{i+1}\right)\cdot\left(\frac{p-i}{p-i-1+\alpha}\right)$. We consider the case where the grid size $p$ is even. The case of $p$ being odd can be proven with minor modifications. Using recursion \eqref{eq:recursion}, 
		\be
		&\sum_{i\in\{0,\cdots,p\}}\bb P^2[x_1(Y_p)=(i,\cdot)] \\ \label{simeq:2}
		\simeq&\sum_{i\in\{0,\cdots,p/2-1\}}\bb P^2[x_1(Y_p)=(i,\cdot)] \\
		=&\mathbb P^2\left[x_1(Y_p)=(0,\cdot)\right]\left(1+\sum_{j=1}^{p/2-1}g^2(j)\right)\\ \label{simeq:1}
		\simeq&p^{-2\alpha}\left(1+\sum_{j=1}^{p/2-1}g^2(j)\right).
		\ee
		
		Equation \eqref{simeq:2} holds by the symmetry of the distributions of origins and destinations. Equation \eqref{simeq:1} uses the fact that $P\left[x_1(Y_p)=(0,\cdot)\right]\simeq p^{-\alpha}$ from the proof of \Cref{lm:NodeLowerBound}. By \Cref{lm:useful_rates} in \Cref{apx:lemmas} which shows that $g(j) \simeq \left(1/(j+1)\right)^{1-\alpha}$ for all $j\le p/2-1$,
		\be
		\eqref{simeq:2} \simeq&p^{-2\alpha}\left(1 + \sum_{j=1}^{p/2-1} \left(\frac{1}{j+1}\right)^{2-2\alpha} \right) \\
		\simeq&p^{-2\alpha}\left(\sum_{j=1}^{p/2}\left(\frac{1}{j}\right)^{2-2\alpha} \right)
		\ee
		
		\begin{itemize}
			\item When $0<\alpha<1/2$, we know that $\sum_{j=1}^{\infty}(1/j)^{2-2\alpha}<+\infty$ converges. As a result, $\sum_{i\in\{0,\cdots,p\}}\bb P^2[x_1(Y_p)=(i,\cdot)]$ as a function of $p$ satisfies,
			\be
			\sum_{i\in\{0,\cdots,p\}}\bb P^2[x_1(Y_p)=(i,\cdot)] \simeq \sum_{i\in\{0,\cdots,p/2-1\}}\bb P^2[x_1(Y_p)=(i,\cdot)]  \simeq p^{-2\alpha}.
			\ee
			
			\item When $1/2\le\alpha\le 1$, we know that  $\sum_{j=1}^{\infty}(1/j)^{2-2\alpha}$ diverges. For $\alpha>1/2$, 
			\be
			\int_{1}^{p/2}\left(\frac{1}{j+1}\right)^{2-2\alpha}dj &\le \sum_{j=1}^{p/2}\left(\frac{1}{j}\right)^{2-2\alpha} \le \int_{1}^{p/2}\left(\frac{1}{j}\right)^{2-2\alpha}dj \\
			\Longleftrightarrow\quad \frac{1}{2\alpha-1}\left( \left(\frac{p}{2} + 1\right)^{2\alpha-1} - 2^{2\alpha - 1} \right) &\le \sum_{j=1}^{p/2}\left(\frac{1}{j}\right)^{2-2\alpha} \le \frac{1}{2\alpha-1}\left(\left(\frac{p}{2}\right)^{2\alpha-1} - 1 \right).
			\ee
			
			This yields,
			\be
			&\sum_{i\in\{0,\cdots,p\}}\bb P^2[x_1(Y_p)=(i,\cdot)] \\
			\simeq&\sum_{i\in\{0,\cdots,p/2-1\}}\bb P^2[x_1(Y_p)=(i,\cdot)]  \\
			\simeq&p^{-2\alpha}\left(\sum_{j=1}^{p/2}\left(\frac{1}{j}\right)^{2-2\alpha} \right) \\
			\simeq&p^{-2\alpha}\cdot p^{2\alpha - 1} \\[2mm]
			\simeq&p^{-1}.
			\ee
		\end{itemize}
		
		Plugging these rates into \eqref{eq:quad} completes the proof.
	\end{proof}
	
	\medskip
	
	\begin{proof}[Proof of \Cref{prop:seg_accum}]
		Consider a segment $s=(i,j)\rightarrow(i+1, j)\in \mathcal{S}_p$ from the grid of even size $p$ and assume that $i<p/2$. By symmetry, segments with vertical movements or at other positions can be proven in the same way. The case with odd $p$ can be proven with minor modifications. Consider a route $Y_p\sim\mu_p$ that covers segment $s$. Let $X_1$ and $X_2$ be the corresponding origin and destination of route $Y_p$. There are two scenarios in which $s\in Y_p$.
		\begin{itemize}
			\item $X_1=(i_1,j)$ for some $i_1\le i$ and $X_2=(i_2,\cdot)$ for some $i_2> i$. This is with probability
			\be
			\simeq&\sum_{i_1=0}^i\sum_{i_2=i+1}^p\mathbb P(X_1 = (i_1, j))\mathbb P(X_2=(i_2,\cdot)) \\ 
			=&\mathbb P(X_1=(\cdot, j))\cdot\left(\sum_{i_1=0}^i\sum_{i_2=i+1}^p\mathbb P(X_1 = (i_1, \cdot))\mathbb P(X_2=(i_2,\cdot))\right)\\ \label{eq:probability}
			=&\mathbb P(X_1=(\cdot, j))\cdot\left(\sum_{i_1=0}^i\mathbb P(X_1 = (i_1, \cdot))\right)\left(\sum_{i_2=i+1}^p\mathbb P(X_2=(i_2,\cdot))\right).
			\ee
			
			Clearly, given $i$, \eqref{eq:probability} achieves its minimum value when $j=p/2$ and achieves its maximum value when $j=0$ or $p$. Similarly, given $j$, we can show that \eqref{eq:probability} achieves its minimum value at $i=0$ and its maximum value at $i=p/2-1$. To see that, for any $i<p/2$,
			\be
			\sum_{i_1=0}^i\mathbb P(X_1 = (i_1, \cdot)) < \sum_{i_2=i+1}^p\mathbb P(X_2=(i_2,\cdot)).
			\ee
			This yields,
			\be
			&\left(\sum_{i_1=0}^{i-1}\mathbb P(X_1 = (i_1, \cdot))\right)\left(\sum_{i_2=i}^p\mathbb P(X_2=(i_2,\cdot))\right) \\
			=&\left(\sum_{i_1=0}^{i}\mathbb P(X_1 = (i_1, \cdot)) - \bb P(X_1=(i,\cdot))\right)\left(\sum_{i_2=i+1}^p\mathbb P(X_2=(i_2,\cdot)) + \bb P(X_2=(i,\cdot))\right)\\
			=&\left(\sum_{i_1=0}^{i}\mathbb P(X_1 = (i_1, \cdot)) - \bb P(X_1=(i,\cdot))\right)\left(\sum_{i_2=i+1}^p\mathbb P(X_2=(i_2,\cdot)) + \bb P(X_1=(i,\cdot))\right) \\
			=&\left(\sum_{i_1=0}^i\mathbb P(X_1 = (i_1, \cdot))\right)\left(\sum_{i_2=i+1}^p\mathbb P(X_2=(i_2,\cdot))\right) \\
			&-\mathbb P(X_1=(i,\cdot))\left(\sum_{i_2=i+1}^p\mathbb P(X_2=(i_2,\cdot)) - \sum_{i_1=0}^i\mathbb P(X_1 = (i_1, \cdot))\right)\\
			&-\bb P^2(X_1=(i,\cdot))\\
			\le&\left(\sum_{i_1=0}^i\mathbb P(X_1 = (i_1, \cdot))\right)\left(\sum_{i_2=i+1}^p\mathbb P(X_2=(i_2,\cdot))\right),
			\ee
			for all $i<p/2$. This suggests that \eqref{eq:probability} achieves its overall maximum at $i=p/2 - 1, j = 0$ or $p$ with 
			\be
			\sum_{i_1=0}^{p/2-1}\sum_{i_2=p/2}^p\mathbb P(X_1 = (i_1, 0))\mathbb P(X_2=(i_2,\cdot)) \simeq \bb P(X_1 = (\cdot,0))\simeq p^{-\alpha}.
			\ee
			On the other hand, it achieves its overall minimum at $i=0,j = p/2$ with
			\be
			\sum_{i_1=0}^{0}\sum_{i_2=1}^p\mathbb P(X_1 = (i_1, p/2))\mathbb P(X_2=(i_2,\cdot)) \simeq \mathbb P(X_1 = (0, p/2)) \simeq p^{-1-\alpha}.
			\ee
			
			\smallskip
			
			\item $X_1=(i_1,\cdot)$ for some $i_1\le i$ and $x_2=(i_2,j)$ for some $i_2> i$. This is with probability 
			\be
			\simeq&\sum_{i_1=0}^i\sum_{i_2=i+1}^p\mathbb P(X_1 = (i_1, \cdot))\mathbb P(X_2=(i_2,j)) \\ 
			=&\mathbb P(X_2=(\cdot, j))\cdot\left(\sum_{i_1=0}^i\sum_{i_2=i+1}^p\mathbb P(X_1 = (i_1, \cdot))\mathbb P(X_2=(i_2,\cdot))\right)\\ 
			=&\mathbb P(X_2=(\cdot, j))\cdot\left(\sum_{i_1=0}^i\mathbb P(X_1 = (i_1, \cdot))\right)\left(\sum_{i_2=i+1}^p\mathbb P(X_2=(i_2,\cdot))\right),
			\ee
			which is symmetric to the previous case and thus has the same conclusion.  
		\end{itemize}
		
		This concludes the proof.
		
	\end{proof}
	
	\medskip
	
	\begin{proof}[Proof of \Cref{thm:asymptotic_opt}]
		We first give a lower bound for $R(\hat{\Theta}_y^\ast)$. By \Cref{lm:info_lb},
		\be
		\label{eq:info_lb}
		R\left(\hat{\Theta}^\ast_{y} \:\bigg|\: y_{[N]}\right)\ge \frac{|y|^2}{\sum_{s,t\in y} N_{s\cup t}\psi_{s,t}+ |y|/\tau^2}.
		\ee
		
		Under a given road network size $p$, let indicator random variable $I_s$ denote whether a road segment $s$ is traversed by a randomly sampled route $Y_p\sim \mu_p$. Given $\{I_s\}_{s\in\mathcal{S}_p}$, we rewrote \eqref{eq:info_lb} as 
		\be
		R\left(\hat{\Theta}^\ast_{y} \:\bigg|\: y_{[N]}\right)\ge& \frac{\left(\sum_{s\in\mathcal{S}_p}I_s\right)^2}{\sum_{s,t\in \mathcal{S}_p} N_{s\cup t}\psi_{s,t}I_sI_t + \left(\sum_{s\in\mathcal{S}_p}I_s\right)/\tau^2}\\
		=&\frac{1}{\sum_{s,t\in\mathcal{S}_p}N_{s\cup t}\psi_{s,t}\frac{I_s I_t}{(\sum_{s'\in\mathcal{S}_p}I_{s'})^2}+\frac{1}{\left(\sum_{s'\in\mathcal{S}_p}I_{s'}\right)\tau^2}} \\
		\ge&\frac{1}{\sum_{s,t\in\mathcal{S}_p}N_{s\cup t}|\psi_{s,t}|\frac{I_s I_t}{(\sum_{s'\in\mathcal{S}_p}I_{s'})^2}+\frac{1}{\left(\sum_{s'\in\mathcal{S}_p}I_{s'}\right)\tau^2}}.
		\ee
		
		Taking expectations over the predicting route and the historical routes,
		\be
		R(\hat{\Theta}_{Y_p}^\ast)\ge&\mathbb E\left[\frac{1}{\sum_{s,t\in\mathcal{S}_p}N_{s\cup t}|\psi_{s,t}|\frac{I_s I_t}{(\sum_{s'\in\mathcal{S}_p}I_{s'})^2}+\frac{1}{\left(\sum_{s'\in\mathcal{S}_p}I_{s'}\right)\tau^2}}\right]\\
		\ge&\frac{1}{\sum_{s,t\in\mathcal{S}_p}\bb E[N_{s\cup t}]|\psi_{s,t}|\bb E\Big[\frac{I_s I_t}{(\sum_{s'\in\mathcal{S}_p}I_{s'})^2}\Big]+\frac{1}{\tau^2}\bb E \Big[\frac{1}{\sum_{s'\in\mathcal{S}_p}I_{s'}}\Big]}.
		\ee
		The second inequality holds because of $\mathbb E[1/X] \ge 1/\mathbb E[X]$ for any non-negative random variable $X$. By \Cref{cor:grid}, we have $R(\Thetasegpbigy) = \mathcal{O}(p^2/N)$. This gives,
		\be
		\frac{R(\Thetasegpbigy)}{R(\hat{\Theta}_{Y_p}^\ast)}\le&\mathcal{O}(p^2/N)\Bigg(\sum_{s,t\in\mathcal{S}_p}\bb E[N_{s\cup t}]|\psi_{s,t}|\bb E\left[\frac{I_s I_t}{(\sum_{s'\in\mathcal{S}_p}I_{s'})^2}\right] \\ \label{eq:ub_risk}
		&+\frac{1}{\tau^2}\bb E \Big[\frac{1}{\sum_{s'\in\mathcal{S}_p}I_{s'}}\Big]\Bigg).
		\ee
		We analyze the two terms in the parenthesis separately. For the first term,
		\be
		&\mathcal{O}(p^2/N)\sum_{s,t\in\mathcal{S}_p}\bb E[N_{s\cup t}]|\psi_{s,t}|\bb E\left[\frac{I_s I_t}{(\sum_{s'\in\mathcal{S}_p} I_{s'})^2}\right]\\
		=&\mathcal{O}(p^2)\sum_{s,t\in\mathcal{S}_p}\frac{\bb E[N_{s\cup t}]}{N}|\psi_{s,t}|\bb E\left[\frac{I_s I_t}{(\sum_{s'\in\mathcal{S}_p} I_{s'})^2}\right]\\
		\le&\mathcal{O}(p^2)\sum_{s,t\in\mathcal{S}_p}\bb P(I_s=1, I_t=1)|\psi_{s,t}|\bb E\left[\frac{I_s}{(\sum_{s'\in\mathcal{S}_p} I_{s'})^2}\right]\\
		\le&\mathcal{O}(p^2)\sum_{s,t\in\mathcal{S}_p}\bb P(I_s=1)|\psi_{s,t}|\bb E\left[\frac{I_s}{(\sum_{s'\in\mathcal{S}_p} I_{s'})^2}\right]\\
		=&\mathcal{O}(p^2)\sum_{s\in\mathcal{S}_p}\sum_{t\in\mathcal{S}_p}\bb P^2(I_s=1)|\psi_{s,t}|\bb E\left[\frac{1}{(\sum_{s'\in\mathcal{S}_p} I_{s'})^2}\:\Bigg|\: I_s=1\right] \\
		=&\mathcal{O}(p^2)\sum_{s\in\mathcal{S}_p}\bb P^2(I_s=1)\bb E\left[\frac{1}{(\sum_{s'\in\mathcal{S}_p} I_{s'})^2}\:\Bigg|\: I_s=1\right]\left(\sum_{t\in\mathcal{S}_p}|\psi_{s,t}|\right) \\
		=&\mathcal{O}(p^2)\sum_{s\in\mathcal{S}_p}\bb P^2(I_s=1)\bb E\left[\frac{1}{(\sum_{s'\in\mathcal{S}_p} I_{s'})^2}\:\Bigg|\: I_s=1\right]\mathcal{O}(1)\\ \label{eq:temp3}
		=&\mathcal{O}(p^2)\sum_{s\in\mathcal{S}_p}\bb P^2(I_s=1)\bb E\left[\frac{1}{(\sum_{s'\in\mathcal{S}_p} I_{s'})^2}\:\Bigg|\: I_s=1\right].
		\ee

		The second-to-last equality uses the second part of Assumption \ref{assumption:3}. We have the following claims whose proof can be found in the proof of \Cref{lm:conditional_length} in \Cref{apx:lemmas}. For any segment $s\in\mathcal{S}_p$,
		\be
		\bb E\left[\frac{1}{(\sum_{s'\in\mathcal{S}_p} I_{s'})^2}\:\Bigg|\: I_s=1\right] = \mathcal{O}(\log(p) p^{-2}),\quad\bb E\left[\frac{1}{(\sum_{s'\in\mathcal{S}_p} I_{s'})^2}\:\Bigg|\: I_s=0\right] = \mathcal{O}(\log(p) p^{-2}).
		\ee
		
		\medskip
		
		Now go back to equation \eqref{eq:temp3}, for any $1/2\le\alpha\le 1$,
		\be
		&\mathcal{O}(p^2)\sum_{s\in\mathcal{S}_p}\bb P^2(I_s=1)\bb E\left[\frac{1}{(\sum_{s'\in\mathcal{S}_p} I_{s'})^2}\:\Bigg|\: I_s=1\right] \\
		=&\mathcal{O}(p^2)\left(\sum_{s\in\mathcal{S}_p}\bb P^2(I_s=1)\right)\mathcal{O}(\mathrm{log}(p)p^{-2}) \\
		=&\mathcal{O}(p^2)\mathcal{O}(1)\mathcal{O}(\mathrm{log}(p)p^{-2})\quad\qquad(\textrm{by \Cref{lm:sum_square_ub} in \Cref{apx:lemmas}})\\[4mm] 
		=&\mathcal{O}(\mathrm{log}(p)).
		\ee
		
		\medskip

		\medskip
		
		For the second term in \eqref{eq:ub_risk}, 
		\be
		~&\mathcal{O}(p^2/N)\frac{1}{\tau^2}\mathbb E\left[\frac{1}{\sum_{s^\prime\in\mathcal{S}_p}I_{s^\prime}}\right] \\
		\le& \mathcal{O}(p^2/N)\frac{1}{\tau^2}\sqrt{\bb E\left[\frac{1}{\left(\sum_{s'\in\mathcal{S}_p}I_{s'}\right)^2}\right]} \\
		=&\mathcal{O}(p^2/N)\frac{1}{\tau^2}\sqrt{\bb E\left[\frac{1}{\left(\sum_{s'\in\mathcal{S}_p}I_{s'}\right)^2} \,\middle\vert\, I_s=1 \right]\mathbb P(I_s=1) + \bb E\left[\frac{1}{\left(\sum_{s'\in\mathcal{S}_p}I_{s'}\right)^2} \,\middle\vert\, I_s=0 \right]\mathbb P(I_s=0)} \\
		=&\mathcal{O}(p^2/N)\frac{1}{\tau^2}\sqrt{\mathcal{O}(p^{-2}\log(p))} \\[1mm]
		=&\mathcal{O}\left(p\sqrt{\log(p)}/N\right)\frac{1}{\tau^2}\\[2mm]
		=&\mathcal{O}(1).\quad\qquad\qquad(\textrm{as } N=\omega(p))
		\ee
		The first inequality uses the fact that $\mathbb E[X]\le \sqrt{\mathbb E[X^2]}$ for any random variable $X$.
		
		\medskip
		
		This completes the proof that $R(\Thetasegpbigy)/R(\hat{\Theta}_{Y_p}^\ast)=\mathcal{O}(\log(p))$ when $N=\omega(p)$ and $1/2\le\alpha\le1$. For the second part of the theorem, using the information-theoretic lower bound (\Cref{lm:info_lb}), given a set of historical routes $Y_{p, [N]}$ and the predicting route $Y_p$ under a grid size $p$, for any estimator,
		\be
		R\left(\hat{\Theta}_{Y_p} \:\bigg|\: Y_{p, [N]}\right)\ge &\frac{|Y_p|^2}{\sum_{s,t\in Y_p} N_{s\cup t}\psi_{s,t}+ |Y_p|/\tau^2} \\
		\ge&\frac{|Y_p|^2}{\sum_{s,t\in Y_p} N\psi_{s,t}+ |Y_p|/\tau^2}\qquad\qquad\textrm{($N\ge N_{s\cup t},\:\forall s,t\in Y_p$)} \\
		=&\frac{|Y_p|^2}{\sum_{s,t\in Y_p} \mathcal{O}(p)\psi_{s,t}+ |Y_p|/\tau^2} \\
		=&\frac{|Y_p|^2}{|Y_p| \mathcal{O}(p)+ |Y_p|/\tau^2} \qquad\qquad\qquad~\textrm{(${\textstyle\sum}_{t\in Y_p} \psi_{s,t}=\mathcal{O}(1)$ by Assumption 2.2)} \\
		=&\frac{|Y_p|}{\mathcal{O}(p)+ 1/\tau^2}.
		\ee
		
		Taking the expectation over the predicting route $Y_p\in\mu_p$ yields
		\be
		R\left(\hat{\Theta}_{Y_p}\right) \ge \frac{\mathbb E[|Y_p|]}{\mathcal{O}(p)+ 1/\tau^2}.
		\ee
		
		We know that as $\alpha$ increases, route origins and destinations are distributed toward the center of the grid. We thus can focus on the case of $\alpha=1$ to get a lower bound for $\mathbb E[|Y_p|]$. When $\alpha=1$,
		\be
		\bb E[|Y_p|] &= \frac{1}{(p+1)^2}\left( \sum_{i_1 = 0}^p \sum_{i_2 = 0}^p |i_1 - i_2| + \sum_{j_1 = 0}^p \sum_{j_2 = 0}^p |j_1 - j_2|\right) \\
		&= \frac{2}{(p+1)^2}\left( \sum_{i_1 = 0}^p \sum_{i_2 = 0}^p |i_1 - i_2|\right) \\
		&= \frac{2}{(p+1)^2} \left( \sum_{i_1 = 0}^p \sum_{i_2 = i_1+1}^p (i_2-i_1) + \sum_{i_1 = 0}^p \sum_{i_2 = 0}^{i_1-1} (i_1 - i_2) \right) \\
		&=\frac{2}{(p+1)^2}\cdot 2\cdot \sum_{i=1}^p \frac{i(i+1)}{2}\\
		&=\frac{2}{(p+1)^2}\left(\sum_{i=1}^p i^2 + \sum_{i=1}^p i \right) \\
		&=\frac{2}{(p+1)^2}\left(  \frac{p(p+1)(2p+1)}{6}+\frac{p(p+1)}{2}\right)\\[1mm]
		&=\frac{p(2p+1)}{3(p+1)} + \frac{p}{p+1}\\[3mm]
		&=\Omega(p).
		\ee
		
		Thus for any $\alpha\in(0,1]$, there exist $\epsilon>0$ such that for any estimator,
		\be
		R\left(\hat{\Theta}_{Y_p}\right) \ge \frac{\Omega(p)}{\mathcal{O}(p)+ 1/\tau^2}>\epsilon,\quad\forall p.
		\ee
		This completes the proof that $\liminf_{p\rightarrow\infty}R\left(\hat{\Theta}_{Y_p}\right) > 0$ for any estimator.
	\end{proof}
	
	\medskip
	
	\begin{proof}[Proof of \Cref{lm:info_lb}]
		The result follows by adapting Theorem 1 of \cite{gill1995applications} which gives a multivariate version of the van Trees inequality. Consider estimating the total travel time on route $y$, $\Theta_y = \sum_{s\in y}\theta_s$, with one single observation $z\in\mathbb R_{\ge0}^{M\times1}$ where $z_i$ is a single observed travel time for a single segment on a single trip and $M=\sum_{n=1}^N|y_n|$. Let $u_i=s$ if the $i^{\textrm{th}}$ observation in $Z$ is a travel time on segment $s$.
		Let $w_i = n$ if the $i^{\textrm{th}}$ observation in $z$ is a travel time from trip $n$. Thus, $z_i = T^\prime_{w_i, u_i},\:\forall i\in\{1,\cdots, M\}$. With a bit abuse of notation, let $\theta_{y} = [\theta_s]_{s\in y}$. Define the Fisher information matrix
		\be
		\mathcal{I}(\theta_y) = \mathbb E\left[\left(\frac{\partial \log f(Z\:|\:\theta_y)}{\partial \theta_y}\right)^\intercal\left(\frac{\partial \log f(Z\:|\:\theta_y)}{\partial \theta_y}\right) \right]\in\mathbb R^{|y|\times|y|},
		\ee
		where the expectation is taken over $Z$ with $u_i$ and $w_i$ fixed and $f(Z \:|\: \theta_y)$ is the density of $Z$ given $\theta_y$. Suppose $f(\cdot\:|\:\theta_y)$ is on an arbitrary measure space $\mathbf{Z}$ for all $\theta_y$. Note that $\partial \log f(Z\:|\:\theta_y)/\partial \theta_y := [\partial \log f(Z\:|\:\theta_y)/\partial \theta_s]_{s\in y}\in\mathbb R^{1\times |y|}$. Let $\lambda(\theta_y)$ be the prior density. Suppose $\theta_y\in\mathbf{\Theta}_y\in\mathbb R^{1\times|y|}$ and $\theta_s\in\mathbf{\Theta}_s\in\mathbb R$. Define the information on the prior distribution $\lambda(\cdot)$,
		\be
		\mathcal{I}(\lambda) = \mathbb E\left[\left(\frac{\partial \log \lambda(\theta_y)}{\partial\theta_y}\right)^\intercal\left(\frac{\partial \log \lambda(\theta_y)}{\partial\theta_y}\right)\right]\in\mathbb R^{|y|\times|y|},
		\ee
		where the expectation is taken over $\theta_y$. The following result is taken from Theorem 1 of \cite{gill1995applications} and adapted to our setup by choosing $B(\theta_y) = 1$ and $C(\theta_y) = 1^{1\times|y|}$ in their theorem. The original assumptions stated in \cite{gill1995applications} are provided below. We call a function $g(\theta_y)$ nice if for each $s\in y$, it is absolutely continuous in $\theta_s$ for almost all values of the other components of $\theta_y$ and its partial derivatives $\partial g/\partial \theta_s$ are measurable in $\theta_y$.
		
		\begin{assumption}
			\label{assump:prior}
			\cite{gill1995applications} impose the following assumptions.
			\begin{enumerate}
				\item $f(z \mid \theta_y)$ is nice in $\theta_y$ for almost all $z$ and its partial derivatives with respect to $\theta_y$ are measurable in $z, \theta_y$.
				
				\item The Fisher information matrix $\mathcal{I}(\theta_y)$ exists and $\diag(\mathcal{I}(\theta_y))^{1/2}$ is locally integrable in $\theta_y$.
				
				\item $\lambda(\theta_y)$ is nice in $\theta_y$; $\mathbf{\Theta}_y$ is compact with boundary which is piecewise $C^1$-smooth; $\lambda(\theta_y)$ is positive on the interior of $\mathbf{\Theta}_y$ and zero on its boundary.  
			\end{enumerate}
		\end{assumption}
		
		\begin{theorem}[Multivariate van Trees inequality]\label{thm:van_Trees} For any estimator $\hat{\Theta}_y$,
			\be
			&\int_{\mathbf{\Theta}_y} \bb E\left[\left(\hat{\Theta}_y - \Theta_y\right)^2\:\bigg|\: y_{[N]}, \theta_y \right]\lambda(\theta_y)d\theta_y\\ \label{eq:van_trees}
			\ge&\frac{|y|^2}{\int_{\mathbf{\Theta}_y}\mathrm{trace}(\mathcal{I}(\theta_y))\lambda(\theta_y)d\theta_y + \mathrm{trace}(\mathcal{I}(\lambda))}.
			\ee
		\end{theorem}
		
		\noindent\textbf{Revised the third part of Assumption \ref{assump:prior}.} We note that the compactness of $\mathbf{\Theta}_y$ can be replaced with $\mathbf{\Theta}_y=\mathbb R^{1\times|y|}$
		and $\lim_{\theta_s\to+\infty} \Theta_y\lambda(\theta_y)=0$ and $\lim_{\theta_s\to-\infty} \Theta_y\lambda(\theta_y)=0$ for all $s\in y$ and for almost all values of the other components of $\theta_y$. We still require $\lambda(\theta_y)$ to be nice in $\theta_y\in\mathbf{\Theta}_y$. We will use this changed assumption later in the proof as we will impose a Gaussian prior for $\lambda(\theta_y)$. Here we provide an updated proof based on this revised assumption. 
		\begin{proof}[Proof of \Cref{thm:van_Trees} under revised Assumption \ref{assump:prior}]
			Most of the proof follows exactly as the one provided on page 65 of \cite{gill1995applications}. We do not repeat the arguments here. The only part we have to re-verify is the derivation of $\mathbb E[XY]$ where $X = \hat{\Theta}_y - \Theta_y$ and $Y = \sum_{s\in y} (\partial \{f(Z\:|\: \theta_y)\lambda(\theta_y) \}/\partial \theta_s)\cdot(1/(f(Z\:|\: \theta_y)\lambda(\theta_y)))$. We let $\mathbf{\Theta}_{-s}$ and $\theta_{-s}$ define the measure space and vector excluding the $s^{\textrm{th}}$ component.
			\be
			\mathbb E[XY] =& \int_{\mathbf{Z}}\int_{\mathbf{\Theta}_y}  (\hat{\Theta}_y - \Theta_y) \sum_{s\in y} \frac{\partial \{f(z\:|\: \theta_y)\lambda(\theta_y)\} }{\partial \theta_s}d\theta dz \\
			=& \int_{\mathbf{Z}}\left(\sum_{s\in y}\int_{\mathbf{\Theta}_{- s}}\int_{\mathbf{\Theta}_s} (\hat{\Theta}_y - \Theta_y) \frac{\partial \{f(z\:|\: \theta_y)\lambda(\theta_y)\} }{\partial \theta_s}d\theta_s d\theta_{-s}\right) dz\\
			=& \int_{\mathbf{Z}}\left(\sum_{s\in y}\int_{\mathbf{\Theta}_{-s}}\int_{-\infty}^{+\infty} (\hat{\Theta}_y - \Theta_y) \frac{\partial \{f(z\:|\: \theta_y)\lambda(\theta_y)\} }{\partial \theta_s}d\theta_s d\theta_{-s}\right) dz  \\
			=& \int_{\mathbf{Z}}\left(\sum_{s\in y}\int_{\mathbf{\Theta}_{-s}}\left(\left[(\hat{\Theta}_y - \Theta_y) f(z\:|\: \theta_y)\lambda(\theta_y)\right]_{-\infty}^{+\infty} + \int_{-\infty}^{+\infty} f(z\:|\:\theta_y)\lambda(\theta_y)d\theta_s \right)d\theta_{-s}\right) dz \\
			=& \int_{\mathbf{Z}}\left(\sum_{s\in y}\int_{\mathbf{\Theta}_{-s}}\left(\int_{-\infty}^{+\infty} f(z\:|\:\theta_y)\lambda(\theta_y)d\theta_s \right)d\theta_{-s}\right) dz \\
			=&|y| \int_{\mathbf{Z}}\int_{\mathbf{\Theta}_y}f(z\:|\:\theta_y)\lambda(\theta_y)d\theta_{y} dz  \\[2mm]
			=&|y|.
			\ee 
			The fourth equality is integration by parts. The fifth equality holds as $\lim_{\theta_s\to+\infty} \Theta_y\lambda(\theta_y)=0$ and $\lim_{\theta_s\to-\infty} \Theta_y\lambda(\theta_y)=0$, which implies that $\lim_{\theta_s\to+\infty} \lambda(\theta_y)=0$ and $\lim_{\theta_s\to-\infty} \lambda(\theta_y)=0$.
			The rest of the proof follows exactly as the one in \cite{gill1995applications} by showing that 
			\be
			\mathbb E[Y^2] = \int_{\mathbf{\Theta}_y}\mathrm{trace}(\mathcal{I}(\theta_y))\lambda(\theta_y)d\theta_y + \mathrm{trace}(\mathcal{I}(\lambda)),
			\ee
			and finally by Cauchy-Schwarz inequality,
			\be
			\int_{\mathbf{\Theta}_y} \bb E\left[\left(\hat{\Theta}_y - \Theta_y\right)^2\:\bigg|\: y_{[N]}, \theta_y \right]\lambda(\theta_y)d\theta_y=\mathbb E[X^2] \ge& \frac{\mathbb E[XY]^2}{\mathbb E[Y^2]} \\
			=&\frac{|y|^2}{\int_{\mathbf{\Theta}_y}\mathrm{trace}(\mathcal{I}(\theta_y))\lambda(\theta_y)d\theta_y + \mathrm{trace}(\mathcal{I}(\lambda))}.
			\ee
			This completes the proof.
		\end{proof}

		\medskip
		
		By the data generative process, we know that
		\be
		f(z \:|\: \theta_y) = \prod_{n=1}^N f_n( \{z_i: w_i = n\} \:|\: \theta_y)~\Rightarrow~\log f(z\:|\: \theta_y) = \sum_{n=1}^N \log f_n( \{z_i: w_i = n\} \:|\: \theta_y),
		\ee
		where $f_n( \{z_i: w_i = n\} \:|\: \theta_y)$ is the density of observing the segment travel times on trip $n$ given $\theta_y$. Define the Fisher information matrix for each trip $n$,
		\be
		\mathcal{I}_n(\theta_y) = \mathbb E\left[\left(\frac{\partial \log f_n( \{Z_i: w_i = n\} \:|\: \theta_y)}{\partial \theta_y}\right)^\intercal\left(\frac{\partial \log f_n(\{Z_i: w_i = n\} \:|\: \theta_y)}{\partial \theta_y}\right) \right]\in\mathbb R^{|y|\times|y|}.
		\ee
		
		By an equivalent definition of the Fisher information matrix under mild regularity conditions (see Lemma 5.3 of \cite{lehmann2006theory}), we have
		\be
		\left[\mathcal{I}(\theta_y)\right]_{s,t} =& -\mathbb E\left[\frac{\partial^2}{\partial \theta_s\partial\theta_t}\log f(Z\:|\:\theta_y) \right] \\
		=& -\sum_{n=1}^N\mathbb E\left[\frac{\partial^2}{\partial \theta_s\partial\theta_t}\log f_n(\{Z_i: w_i = n\}\:|\:\theta_y) \right] \\
		=&\sum_{n=1}^N  \left[\mathcal{I}_n(\theta_y)\right]_{s,t},
		\ee
		for all $s,t\in y$. Plugging this into \eqref{eq:van_trees} yields,
		\be
		&\int_{\mathbf{\Theta}_y} \bb E\left[\left(\hat{\Theta}_y - \Theta_y\right)^2\:\bigg|\: y_{[N]}, \theta_y \right]\lambda(\theta_y)d\theta_y \\
            \ge& R\left(\hat{\Theta}^\ast_{y} \:\bigg|\: y_{[N]}\right) \\ \label{eq:van_tree_2}
		\ge&\frac{|y|^2}{\sum_{n=1}^N \int_{\mathbf{\Theta}_{y}}\textrm{tr}\:\mathcal{I}_n(\theta_{y})\lambda(\theta_{y})d\theta_{y} + \textrm{tr}\:\mathcal{I}(\lambda)}.
		\ee

       Under the assumption that $(\mathcal{E},\theta)$ are jointly Gaussian distributed, both the means of segment travel times, and the segment travel times conditional on their means are normally distributed. Note that Gaussian priors and posteriors satisfy the revised Assumption \ref{assump:prior}. This greatly simplifies the analysis --- we know that under multivariate normal the Fisher information matrix is simply the precision matrix. This gives $\mathcal{I}(\lambda)=\diag([\underbrace{1/\tau^2,\cdots,1/\tau^2}_{|y|}])$. Moreover,
		\be
		\left[\mathcal{I}_n(\theta_{y})\right]_{s,t} = \begin{cases}
			\psi_{s,t}, & \text{if $s,t\in y_n$},\\
			0, & \text{o/w},
		\end{cases}
		\ee
		for all $s,t\in y$. This further yields,
		\be
		R\left(\hat{\Theta}^\ast_{y}\:\Big|\:y_{[N]}\right)\ge&\frac{|y|^2}{\sum_{n=1}^N \sum_{s,t\in y_n}\psi_{s,t} + |y|/\tau^2} \\
		=&\frac{|y|^2}{\sum_{s,t\in y} N_{s\cup t}\psi_{s,t}+ |y|/\tau^2}.
		\ee
		
		As a sanity check, when the sample size is zero ($N_{s\cup t}=0,\:\forall s,t\in y$), we have $R(\hat{\Theta}^\ast_{y}\:|\: y_{[N]})\ge|y|\tau^2$. This matches the intuition because if there is no historical data at all, the best estimator should just be the prior mean $|y|\mu$. It leads to the integrated risk $|y|\tau^2$ which contains no variance but only bias. %
	\end{proof}
	
	\bigskip

	\section{Lemmas} 
	\label{apx:lemmas}
	
	\begin{lemma} \label{lm:SumBound}
		Let 
		\be
		S(q_s,N) = \sum_{N_s = 1}^{N} \frac1{N_s} {N \choose N_s} q_s^{N_s} (1-q_s)^{N-N_s}.
		\ee
		Then
		\be
		\frac{1-(1-q_s)^{N+1}}{(N+1)q_s} - (1-q_s)^N &< S(q_s,N) \\
		2 \left( \frac{1-(1-q_s)^{N+1}}{(N+1)q_s} - (1-q_s)^N \right) &> S(q_s,N)
		\ee
	\end{lemma}
	\begin{proof}
		For $N_s \sim \Binom(N,q_s)$, we use \cite[Eqn. 3.4]{chao1972negative} to obtain
		\be
		\bb E[(N_s+1)^{-1}] &= (1-q_s)^N + \sum_{N_s = 1}^N \frac1{N_s + 1} {N \choose N_s} q_s^{N_s} (1-q_s)^{N-N_s} \\
		&= \frac{1-(1-q_s)^{N+1}}{(N+1)q_s},
		\ee
		and since
		\be
		\sum_{N_s = 1}^N \frac1{N_s + 1} {N \choose N_s} q_s^{N_s} (1-q_s)^{N-N_s} &< S(q_s,N), \\
		2 \sum_{N_s = 1}^N \frac1{N_s + 1} {N \choose N_s} q_s^{N_s} (1-q_s)^{N-N_s} &> S(q_s,N),
		\ee
		the result follows.
	\end{proof}
	
	\medskip
	
	\begin{lemma}
		\label{lm:useful_rates}
		For an even number $p$, consider the following function of $j\in\bb Z_{>0}$,
		\be
		g(j)=\prod_{i=1}^j \frac{i+\alpha}{i+1}\cdot\frac{p-i}{p-i-1+\alpha},
		\ee
		we have,
		\be
		g(j) \simeq \left(\frac{1}{j+1}\right)^{1-\alpha},
		\ee
		for all $0\le\alpha\le 1$ and $j\le p/2-1$.
	\end{lemma}
	\begin{proof}
		We first show that $(i+\alpha)/(i+1)\ge (i/(i+1))^{1-\alpha},\:\forall i\in\mathbb Z_{>0}$. Notice that at the two end points $\alpha=0$ and $\alpha=1$, we have $(i+\alpha)/(i+1)=(i/(i+1))^{1-\alpha},\:\forall i\in\mathbb Z_{>0}$. Given that
		\be
		\frac{d^2}{d\alpha^2}\left(\frac{i+\alpha}{i+1} - \left(\frac{i}{i+1}\right)^{1-\alpha}\right) = -\left(\frac{i}{i+1}\right)^{1-\alpha}\log^2\left(\frac{i}{i+1}\right) < 0,
		\ee
		we know that $(i+\alpha)/(i+1)-(i/(i+1))^{1-\alpha}$ is concave in $\alpha$ for any $i\in\bb Z_{>0}$. This gives $(i+\alpha)/(i+1)-(i/(i+1))^{1-\alpha}\ge0,\:\forall i\in\mathbb Z_{>0},\forall\alpha\in[0,1]$. We then have
		\be
		g(j)\ge\prod_{i=1}^j \frac{i+\alpha}{i+1}\ge\prod_{i=1}^j\left(\frac{i}{i+1}\right)^{1-\alpha}=\left(\frac{1}{j+1}\right)^{1-\alpha}.
		\ee
		\medskip
		
		We now prove the other direction. For some $i\in\mathbb Z_{>0}$, we first look at the function 
		\be
		f(\alpha) = \frac{(i+\alpha)/(i+1)}{\left(i/(i+1)\right)^{1-\alpha}},\:\:\forall \alpha\in[0,1].
		\ee
		
		The first-order condition of $f(\alpha)$ is
		\be
		\frac{d}{d\alpha} f(\alpha) = \frac{\left(\frac{i}{i+1}\right)^\alpha\left((i+\alpha)\log(\frac{i}{i+1}) + 1\right)}{i} = 0,
		\ee
		which yields
		\be
		\label{eq:opt_alpha}
		\alpha^\ast(i) = \frac{1}{\log((i+1)/i)} - i \in[0,1],\:\:\forall i\in\mathbb Z_{>0}.
		\ee
		
		The second derivative of $f(\alpha)$ is
		\be
		\frac{d^2}{d\alpha^2} f(\alpha) =& \frac{(i+\alpha)\left(\frac{i}{i+1}\right)^{\alpha-1}\log^2\left(\frac{i}{i+1}\right)}{i+1} + \frac{2\left(\frac{i}{i+1}\right)^{\alpha-1}\log\left(\frac{i}{i+1}\right)}{i+1} \\
		=&\frac{1}{i+1}\left(\frac{i}{i+1}\right)^{\alpha-1}\log\left(\frac{i}{i+1}\right)\left((i+\alpha)\log\left(\frac{i}{i+1}\right)+2\right)<0,\:\:\forall \alpha\in[0,1], \forall i\in\mathbb Z_{>0}.
		\ee
		
		This suggests that $\alpha^\ast(i)$ in \eqref{eq:opt_alpha} is the solution that maximizes $f(\alpha)$ for a given $i\in\mathbb Z_{>0}$. Moreover, $\alpha^\ast(i)$ is increasing in $i$ and $\lim_{i\rightarrow+\infty} \alpha^\ast(i) = 1/2$, which gives $\alpha^\ast(i)\in[0,1/2],\:\forall i \in\mathbb Z_{>0}$.
		
		\smallskip
		
		We now show that 
		\be
		\prod_{i=1}^j \frac{i+\alpha}{i+1}\lesssim \prod_{i=1}^j \left(\frac{i}{i+1}\right)^{1-\alpha} = \left(\frac{1}{j+1}\right)^{1-\alpha},\:\:\forall \alpha\in[0,1],\forall j\in\mathbb Z_{>0}.
		\ee
		
		To see that,
		\be
		\label{eq:temp4}
		\frac{\prod_{i=1}^j \frac{i+\alpha}{i+1}}{\prod_{i=1}^j \left(\frac{i}{i+1}\right)^{1-\alpha}} \le \frac{\prod_{i=1}^j \frac{i+\alpha^\ast(i)}{i+1}}{\prod_{i=1}^j \left(\frac{i}{i+1}\right)^{1-\alpha^\ast(i)}} \le \frac{\prod_{i=1}^j \frac{i+1/2}{i+1}}{\prod_{i=1}^j \left(\frac{i}{i+1}\right)^{1-\alpha^\ast(i)}}.
		\ee
		
		For $\prod_{i=1}^j \left(\frac{i}{i+1}\right)^{1-\alpha^\ast(i)}$,
		\be
		&\prod_{i=1}^j \left(\frac{i}{i+1}\right)^{1-\alpha^\ast(i)} \\
		=& \prod_{i=1}^j \left(\frac{i}{i+1}\right)^{1+i-\frac{1}{\log((i+1)/i)}} \\
		=&\prod_{i=1}^j \left(\frac{i}{i+1}\right)^{1/2}\left(\frac{i}{i+1}\right)^{1/2 + i}\left(\frac{i}{i+1}\right)^{-\frac{1}{\log((i+1)/i)}} \\
		=&\prod_{i=1}^j \left(\frac{i}{i+1}\right)^{1/2}\left(\frac{i}{i+1}\right)^{1/2 + i}e \\
		\ge&\prod_{i=1}^j \left(\frac{i}{i+1}\right)^{1/2}\left(e^{-1} - \frac{1}{12e\cdot i^2}\right)e\qquad\quad(\textrm{using Taylor series of $\left(i/(i+1)\right)^{1/2 + i}$ at $+\infty$}) \\
		=&\left(\frac{1}{j+1}\right)^{1/2}\prod_{i=1}^j\left(1 - \frac{1}{12\cdot i^2}\right)\\
		\gtrsim&\left(\frac{1}{j+1}\right)^{1/2}.
		\ee
		The last inequality uses a result in analysis that for a series of $0<p_i<1,\: i\in\mathbb Z_{>0}$, a sufficient and necessary condition for $\prod_{i=1}^{+\infty} (1 - p_i)>0$ is $\sum_{i=1}^{+\infty} p_i < +\infty$. This leads to $\prod_{i=1}^j(1 - 1/(12\cdot i^2))\ge\prod_{i=1}^{+\infty}(1 - 1/(12\cdot i^2))>0$.
		
		\smallskip
		
		On the other hand,
		\be
		\prod_{i=1}^j \frac{i+1/2}{i+1} = \prod_{i=1}^j \frac{2i+1}{2i+2}.
		\ee
		
		We know that
		\be
		\prod_{i=1}^j \frac{2i+1}{2i+2} \cdot \prod_{i=1}^j \frac{2i}{2i+1} = \frac{1}{j+1},\quad\frac{\prod_{i=1}^j \frac{2i+1}{2i+2}}{\prod_{i=1}^j \frac{2i}{2i+1}} \le (3/4)/(2/3) = 9/8,
		\ee
		which suggests that
		\be
		\prod_{i=1}^j \frac{i+1/2}{i+1} \lesssim \left(\frac{1}{j+1}\right)^{1/2}.
		\ee
		
		Using \eqref{eq:temp4}, 
		\be
		\frac{\prod_{i=1}^j \frac{i+\alpha}{i+1}}{\prod_{i=1}^j \left(\frac{i}{i+1}\right)^{1-\alpha}} \lesssim \frac{\left(\frac{1}{j+1}\right)^{1/2}}{\left(\frac{1}{j+1}\right)^{1/2}} = 1\quad\Rightarrow\quad \prod_{i=1}^j \frac{i+\alpha}{i+1}\lesssim \prod_{i=1}^j \left(\frac{i}{i+1}\right)^{1-\alpha} = \left(\frac{1}{j+1}\right)^{1-\alpha}.
		\ee
		
		This further gives,
		\be
		g(j)=&\prod_{i=1}^j \frac{i+\alpha}{i+1}\cdot\frac{p-i}{p-i-1+\alpha} \\
		\le&\prod_{i=1}^j \frac{i+\alpha}{i+1}\cdot\frac{p-i}{p-i-1} \\
		=&\left(\prod_{i=1}^j \frac{i+\alpha}{i+1}\right)\cdot\frac{p-1}{p-j-1} \\
		\le&\left(\prod_{i=1}^j \frac{i+\alpha}{i+1}\right)\cdot\frac{p-1}{p-p/2+1-1} \\
		=&\left(\prod_{i=1}^j \frac{i+\alpha}{i+1}\right)\cdot\frac{p-1}{p/2} \\
		\le&2\left(\prod_{i=1}^j \frac{i+\alpha}{i+1}\right) \\
		\lesssim&\left(\frac{1}{j+1} \right)^{1-\alpha}.
		\ee
		This completes the proof.
	\end{proof}
	
	\medskip

	\begin{lemma}
		\label{lm:sum_square_ub}
		Under the route distribution $\mu_p$ in \Cref{subsec:grid},
		\be
		\bb \sum_{s\in\mathcal{S}_p}q_s^2=\sum_{s\in\mathcal{S}_p} \bb P^2\left[s\in Y_p\right]\simeq \begin{cases} 1, &  \frac12< \alpha\le1, \\ p^{1-2\alpha}, & 0<\alpha\le\frac12. \end{cases}.
		\ee
	\end{lemma}
	\begin{proof}
		Without loss of generality, we focus on the case that $p$ is even. From the proof of \Cref{prop:q_delta_od}, we know that for a segment $s=(i,j)\rightarrow(i+1,j)$, 
		\be
		&\bb P\left[s\in Y_p \right] \\
		\simeq&\sum_{i_1=0}^i\sum_{i_2=i+1}^p\mathbb P(x_1 = (i_1, j))\mathbb P(x_2=(i_2,\cdot)) \\ 
		=&\mathbb P(x_1=(\cdot, j))\cdot\left(\sum_{i_1=0}^i\sum_{i_2=i+1}^p\mathbb P(x_1 = (i_1, \cdot))\mathbb P(x_2=(i_2,\cdot))\right)\\
		=&\mathbb P(x_1=(\cdot, j))\cdot\left(\sum_{i_1=0}^i\mathbb P(x_1 = (i_1, \cdot))\right)\left(\sum_{i_2=i+1}^p\mathbb P(x_2=(i_2,\cdot))\right).
		\ee
		
		Moreover, because $\mu_p$ is symmetric,
		\be
		&\sum_{s\in\mathcal{S}_p} \bb P^2\left[s\in Y_p\right]\\ \simeq&\sum_{i\in\{0,\cdots,p/2-1\}}\sum_{j\in\{0,\cdots,p/2\}} \bb P^2\left[s=(i,j)\rightarrow(i+1,j)\in Y_p\right] \\
		\simeq&\sum_{i\in\{0,\cdots,p/2-1\}}\sum_{j\in\{0,\cdots,p/2\}}\mathbb P^2(x_1=(\cdot, j))\cdot\left(\sum_{i_1=0}^i\mathbb P(x_1 = (i_1, \cdot))\right)^2\left(\sum_{i_2=i+1}^p\mathbb P(x_2=(i_2,\cdot))\right)^2 \\
		=&\left(\sum_{j\in\{0,\cdots,p/2\}}\mathbb P^2(x_1=(\cdot, j))\right)\cdot\left(\sum_{i\in\{0,\cdots,p/2-1\}}\left(\sum_{i_1=0}^i\mathbb P(x_1 = (i_1, \cdot))\right)^2\left(\sum_{i_2=i+1}^p\mathbb P(x_2=(i_2,\cdot))\right)^2 \right)\\ \label{eq:temp1}
		\simeq&\left(\sum_{j\in\{0,\cdots,p/2\}}\mathbb P^2(x_1=(\cdot, j))\right)\cdot\left(\sum_{i\in\{0,\cdots,p/2-1\}}\left(\sum_{i_1=0}^i\mathbb P(x_1 = (i_1, \cdot))\right)^2\right).
		\ee
		\begin{itemize}
			\item When $1/2<\alpha\le1$, from the proof of \Cref{prop:q_delta_od},
			\be
			\sum_{j\in\{0,\cdots,p/2\}}\mathbb P^2(x_1=(\cdot, j))\simeq p^{-1},
			\ee
			and
			\be
			\sum_{i_1=0}^i\mathbb P(x_1 = (i_1, \cdot))\simeq p^{-\alpha}\left(\sum_{i_1=0}^i\left(\frac{1}{i_1+1}\right)^{1-\alpha} \right)\simeq p^{-\alpha}(i+1)^\alpha.
			\ee
			This yields,
			\be
			\eqref{eq:temp1}\simeq p^{-1} p^{-2\alpha}\sum_{i\in\{1,\cdots,p/2\}}i^{2\alpha}\simeq p^{-1}p^{-2\alpha}p^{2\alpha+1} \simeq 1.
			\ee
			\item When $0<\alpha\le 1/2$, from the proof of \Cref{prop:q_delta_od},
			\be
			\sum_{j\in\{0,\cdots,p/2\}}\mathbb P^2(x_1=(\cdot, j))\simeq p^{-2\alpha},
			\ee
			and
			\be
			\sum_{i_1=0}^i\mathbb P(x_1 = (i_1, \cdot))\simeq p^{-\alpha}\left(\sum_{i_1=0}^i\left(\frac{1}{i_1+1}\right)^{1-\alpha} \right)\simeq p^{-\alpha}(i+1)^\alpha.
			\ee
			This yields,
			\be
			\eqref{eq:temp1}\simeq p^{-2\alpha} p^{-2\alpha}\sum_{i\in\{1,\cdots,p/2\}}i^{2\alpha}\simeq p^{-2\alpha}p^{-2\alpha}p^{2\alpha+1} \simeq p^{1-2\alpha}.
			\ee
		\end{itemize}
		This completes the proof.
	\end{proof}
	
	\medskip
	
	\begin{lemma} For any $i,j\in\{0,\cdots,p\}$,
		\label{lm:algebra}
		\be
		&\sum_{i_1=0}^i\sum_{i_2=i+1}^p\sum_{j_2=0}^p\frac{1}{(|i_1-i_2| + |j-j_2|)^2} \le2\left(\sum_{n=1}^{i+1}\frac{1+\cdots+n}{n^2} + \sum_{n=i+2}^{2p}\frac{1+\cdots + (i + 1) + (i + 1)(n-i-1)}{n^2}\right).
		\ee
	\end{lemma}
	\begin{proof}
		We have,
		\be
		&\sum_{i_1=0}^i\sum_{i_2=i+1}^p\sum_{j_2=0}^p\frac{1}{(|i_1-i_2| + |j-j_2|)^2} \\
		=&\sum_{j_2=0}^p\sum_{i_2=i+1}^p\sum_{i_1=i}^0\frac{1}{(|i_1-i_2| + |j-j_2|)^2} \\
		\le&2\left(\sum_{j_2=0}^p\sum_{i_2=i+1}^p\sum_{i_1=i}^0\frac{1}{(|i_1-i_2| + j_2)^2}\right)\quad\qquad\textrm{(by symmetry and $j\in\{0,\cdots,p\}$)} \\
		=&2\left(\sum_{j_2=0}^p\sum_{i_2=1}^{p-i}\sum_{i_1=0}^{i}\frac{1}{(i_1 + i_2 + j_2)^2}\right) \\
		\le&2\left(\sum_{n=1}^{i+1}\frac{1+\cdots+n}{n^2} + \sum_{n=i+2}^{2p}\frac{1+\cdots+ (i+1) + (i+1)(n - i - 1)}{n^2}\right).
		\ee
		This completes the proof. 
	\end{proof}
	
	\medskip
	
	\begin{lemma}
		\label{lm:conditional_length}
		Under the route distribution $\mu_p$ in \Cref{subsec:grid}, for any segment $s\in\mathcal{S}_p$,
		\be
		\bb E\left[\frac{1}{(\sum_{s'\in\mathcal{S}_p} I_{s'})^2}\:\Bigg|\: I_s=1\right] = \mathcal{O}(\log(p) p^{-2}),\quad\bb E\left[\frac{1}{(\sum_{s'\in\mathcal{S}_p} I_{s'})^2}\:\Bigg|\: I_s=0\right] = \mathcal{O}(\log(p) p^{-2}).
		\ee
	\end{lemma}
	\begin{proof}
		For any segment $s$, $\bb E[1 /(\sum_{s'\in\mathcal{S}_p} I_{s'})^2 \:|\: I_s=1]$ and $\bb E[1 /(\sum_{s'\in\mathcal{S}_p} I_{s'})^2 \:|\: I_s=0]$ are increasing in $\alpha\in(0,1]$, this is simply because as $\alpha$ increases, route origins and destinations are more concentrated in the center of the grid. We can thus focus on the case where $\alpha=1$ to get upper bounds. Let $\mathcal{X}_s$ be the set of origins and destinations such that $\bb P[x_1=(i_1,j_1), x_2=(i_2,j_2) \:|\: I_s = 1]>0$ for any origin-destination pair $((i_1, j_1), (i_2,j_2))\in \mathcal{X}_s$. Similarly, let $\mathcal{X}^\prime_s$ be the set of origins and destinations such that $\bb P[x_1=(i_1,j_1), x_2=(i_2,j_2) \:|\: I_s = 0]>0$ for any origin-destination pair $((i_1, j_1), (i_2,j_2))\in \mathcal{X}^\prime_s$.  
		
		\medskip
		
		For any segment $s\in\mathcal{S}_p$, when $\alpha=1$, i.e., route origins and destinations are uniformly distributed over the grid,
		\be
		&\bb E\left[\frac{1}{(\sum_{s'\in\mathcal{S}_p} I_{s'})^2}\:\Bigg|\: I_s=1\right]\\ =&\sum_{((i_1,j_1),(i_2,j_2))\in\mathcal{X}_s}\frac{1}{(|i_1-i_2| + |j_1-j_2|)^2}\cdot\bb P[x_1 = (i_1,j_1), x_2=(i_2,j_2) \mid I_s = 1] \\
		=&\sum_{((i_1,j_1),(i_2,j_2))\in\mathcal{X}_s}\frac{1}{(|i_1-i_2| + |j_1-j_2|)^2}\cdot\frac{\bb P[I_s = 1 \mid x_1 = (i_1,j_1), x_2=(i_2,j_2)]\bb P[x_1 = (i_1,j_1), x_2=(i_2,j_2)]}{\bb P[I_s = 1]}  \\ \label{eq:sim_eq}
		\simeq&\sum_{((i_1,j_1),(i_2,j_2))\in\mathcal{X}_s}\frac{1}{(|i_1-i_2| + |j_1-j_2|)^2}\cdot\frac{\bb P[x_1 = (i_1,j_1), x_2=(i_2,j_2)]}{\sum_{((i_1,j_1),(i_2,j_2))\in\mathcal{X}_s}\bb P[x_1 = (i_1,j_1), x_2=(i_2,j_2)]} \\ \label{eq:temp2}
		=&\sum_{((i_1,j_1),(i_2,j_2))\in\mathcal{X}_s}\frac{1}{(|i_1-i_2| + |j_1-j_2|)^2}\cdot\frac{1}{|\mathcal{X}_s|}.
		\ee
		
		Equation \eqref{eq:sim_eq} holds because $\bb P[I_s = 1 \mid x_1 = (i_1,j_1), x_2=(i_2,j_2)]\in\{0.5, 1\}$, $\forall ((i_1,j_1),(i_2,j_2))\in\mathcal{X}_s$. Moreover, equation \eqref{eq:temp2} holds because of the uniformity of the distribution of origins and destinations. 
		
		\smallskip
		
		Without loss of generality, consider any segment with horizontal movement $s=(i,j)\rightarrow(i+1,j)\in\mathcal{S}_p$ with $i<p/2$,
		\be
		&\bb E\left[\frac{1}{(\sum_{s'\in\mathcal{S}_p} I_{s'})^2}\:\Bigg|\: I_s=1\right]\\
		\simeq&\sum_{((i_1,j_1),(i_2,j_2))\in\mathcal{X}_s}\frac{1}{(|i_1-i_2| + |j_1-j_2|)^2}\cdot\frac{1}{|\mathcal{X}_s|} \\
		=&\sum_{i_1\in\{0,\cdots,i\}}\sum_{i_2\in\{i+1,\cdots,p\}}\sum_{j_2\in\{0,\cdots,p\}}\frac{1}{(|i_1-i_2| + |j-j_2|)^2}\cdot\frac{1}{|\mathcal{X}_s|}\\
		&+\sum_{i_1\in\{0,\cdots,i\}}\sum_{j_1\in\{0,\cdots,p\}}\sum_{i_2\in\{i+1,\cdots,p\}}\frac{1}{(|i_1-i_2| + |j_1-j|)^2}\cdot\frac{1}{|\mathcal{X}_s|} \\
		=&\frac{2}{|\mathcal{X}_s|}\left(\sum_{i_1\in\{0,\cdots,i\}}\sum_{i_2\in\{i+1,\cdots,p\}}\sum_{j_2\in\{0,\cdots,p\}}\frac{1}{(|i_1-i_2| + |j-j_2|)^2}\right) \\
		\le&\frac{4}{|\mathcal{X}_s|}\left(\sum_{n=1}^{i+1}\frac{1 + \cdots + n}{n^2} + \sum_{n=i+2}^{2p}\frac{1+\cdots+ (i+1) + (i+1)(n-i-1)}{n^2}\right)~\textrm{(by \Cref{lm:algebra} in \Cref{apx:lemmas})} \\
		=&\frac{4}{|\mathcal{X}_s|}\left(\sum_{n=1}^{i+1}\frac{n(n+1)/2}{n^2} + \sum_{n=i+2}^{2p}\frac{(i+2)(i+1)/2 + (i+1)(n-i-1)}{n^2}\right)\\
		\le&\frac{4}{|\mathcal{X}_s|}\left((i+1) + (i+1)\sum_{n=i+2}^{2p}\frac{n-i/2}{n^2}\right) \\
		=&\frac{4}{|\mathcal{X}_s|}(i+1)\mathcal{O}(\mathrm{log}(p))\\%(\log(p))\\
		\simeq&\frac{4}{(i+1)p^2} (i+1) \mathcal{O}(\mathrm{log}(p)) \\[2mm]
		=&\mathcal{O}(\mathrm{log}(p)p^{-2}).
		\ee

		\medskip

		Similarly, 
		\be
		&\bb E\left[\frac{1}{(\sum_{s'\in\mathcal{S}_p} I_{s'})^2}\:\Bigg|\: I_s=0\right]\\ =&\sum_{((i_1,j_1),(i_2,j_2))\in\mathcal{X}^\prime_s}\frac{1}{(|i_1-i_2| + |j_1-j_2|)^2}\cdot\bb P[x_1 = (i_1,j_1), x_2=(i_2,j_2) \mid I_s = 0] \\
		=&\sum_{((i_1,j_1),(i_2,j_2))\in\mathcal{X}^\prime_s}\frac{1}{(|i_1-i_2| + |j_1-j_2|)^2}\cdot\frac{\bb P[I_s = 0 \mid x_1 = (i_1,j_1), x_2=(i_2,j_2)]\bb P[x_1 = (i_1,j_1), x_2=(i_2,j_2)]}{\bb P[I_s = 0]}  \\ \label{eq:sim_eq_2}
		\simeq&\sum_{((i_1,j_1),(i_2,j_2))\in\mathcal{X}^\prime_s}\frac{1}{(|i_1-i_2| + |j_1-j_2|)^2}\cdot\frac{\bb P[x_1 = (i_1,j_1), x_2=(i_2,j_2)]}{\sum_{((i_1,j_1),(i_2,j_2))\in\mathcal{X}^\prime_s}\bb P[x_1 = (i_1,j_1), x_2=(i_2,j_2)]} \\ \label{eq:temp2_2}
		=&\sum_{((i_1,j_1),(i_2,j_2))\in\mathcal{X}^\prime_s}\frac{1}{(|i_1-i_2| + |j_1-j_2|)^2}\cdot\frac{1}{|\mathcal{X}^\prime_s|}.
		\ee
		
		Equation \eqref{eq:sim_eq_2} holds because $\bb P[I_s = 0 \mid x_1 = (i_1,j_1), x_2=(i_2,j_2)]\in\{0.5, 1\}$, $\forall ((i_1,j_1),(i_2,j_2))\in\mathcal{X}^\prime_s$, and equation \eqref{eq:temp2_2} holds because of the uniformity of the distribution of origins and destinations. 
		
		\medskip
		
		Without loss of generality, consider any segment with horizontal movement $s=(i,j)\rightarrow(i+1,j)\in\mathcal{S}_p$ with $i<p/2$,
		\be
		&\bb E\left[\frac{1}{(\sum_{s'\in\mathcal{S}_p} I_{s'})^2}\:\Bigg|\: I_s=0\right]\\
		\simeq&\sum_{((i_1,j_1),(i_2,j_2))\in\mathcal{X}^\prime_s}\frac{1}{(|i_1-i_2| + |j_1-j_2|)^2}\cdot\frac{1}{|\mathcal{X}^\prime_s|} \\
		\simeq& \left(\sum_{i_1=0}^p\sum_{i_2=0}^p\sum_{j_1=0}^p\sum_{j_2=0}^p\frac{1}{(|i_1-i_2| + |j_1-j_2|)^2}\right)\cdot\frac{1}{p^4} \\
		=&\mathcal{O}\left((p^{2}\mathrm{log}(p)) \cdot \frac{1}{p^4}\right) \\[3mm]
		=&\mathcal{O}(p^{-2}\mathrm{log}(p)).
		\ee
		
		The second equality holds because the only origin-destination pairs that are \emph{excluded} in $\mathcal{X}^\prime_s$ are those with $j_1=j_2=j$ and $i_1\in\{0,\cdots, i\}, i_2\in\{i+1,\cdots, p\}$. The cardinality of these origin-destination pairs is of a much smaller order ($p^2$) compared to the cardinality of all possible origin-destination pairs (of the order of $p^4$). This completes the proof. 
	\end{proof}
	
	\bigskip
	
	\section{Additional Numerical Experiments}
	\label{sec:additional_num}

	\Cref{fig:risk_1.0_1_3_1} and \Cref{fig:risk_1.0_3_1_1} report additional numerical experiments based on the setup in \Cref{sec:numerical} under $\alpha=1.0$ with two different specifications of the covariance matrices $e^{-3\mathscr{L}} + I$ and $3e^{-\mathscr{L}} + I$. The results are qualitatively similar. Simple segment-based method tends to perform a bit worse when correlation is strong and sample size is small, but it quickly regains competitiveness as the sample size increases.
    We further test two additional covariance structures which do not satisfy the second part of Assumption \ref{assumption:3}, under $\alpha=1.0$. In both cases, the covariance matrices of the segment travel times $\Sigma_p$ for the grid network with size $p$ is constructed as $\Sigma_p = (1/|\mathcal{S}_p|^2)K_p^\intercal K_p$ where $K_p$ is an $|\mathcal{S}_p|\times|\mathcal{S}_p|$ random matrix. The two cases differ in the way $K_p$ is generated. 
	
	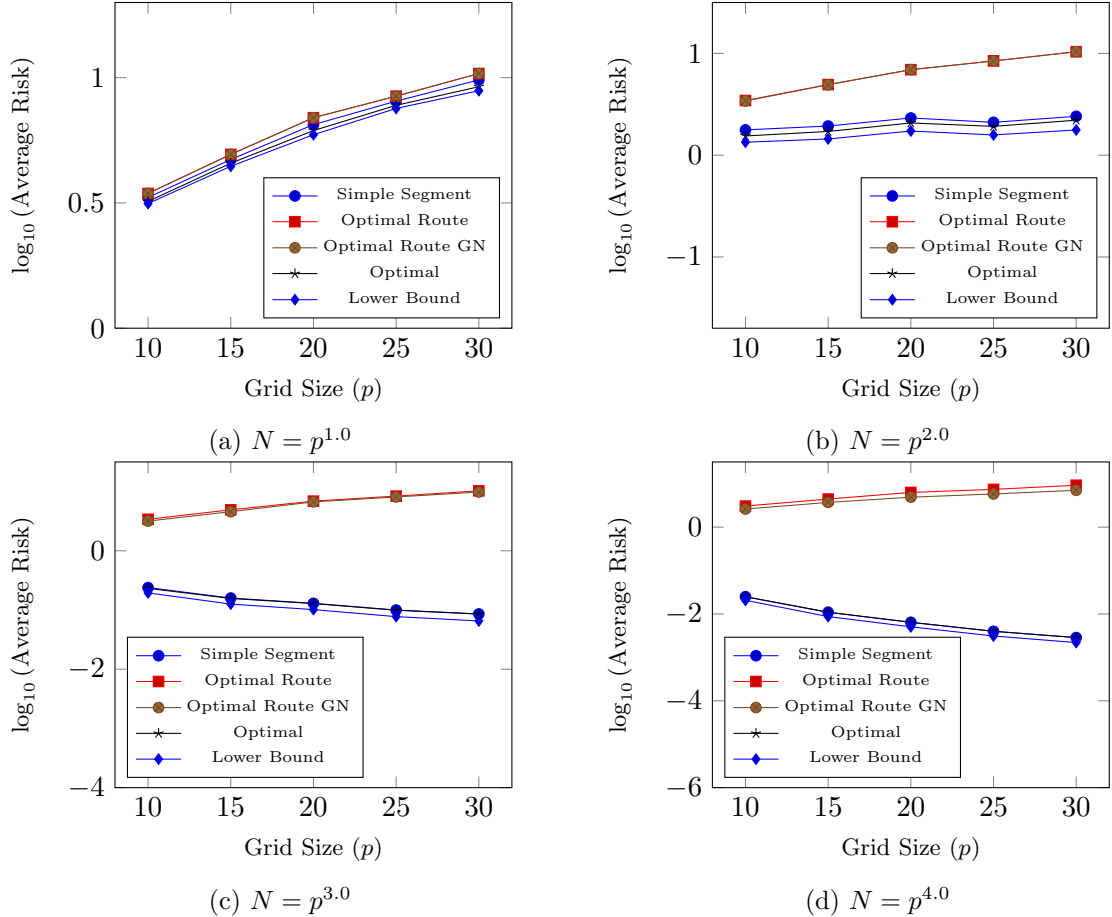
\begin{figure}[h]
		\captionsetup[subfigure]{justification=centering}
		\centering 
		\begin{subfigure}[h]{0.45\textwidth}
			\begin{tikzpicture}
				\begin{axis}[
					width=2.7in,
					xlabel = Grid Size ($p$),
					ylabel = $\log_{10}\textrm{(Average Risk)}$,
					legend style={legend pos=south east, font=\tiny},
					label style={font=\footnotesize},
					ymin=0,
					ymax=1.3,
					xtick={10,15,20,25,30}
					]
					
					\addplot+ [
					discard if not={alpha}{1.0},
					on layer=background
					] table [
					x=grid_size,
					y=seg_simple,
					]{data/revision/results_1.0_cov_6_1_3_1.txt};\addlegendentry{Simple Segment};
					
					\addplot+ [
					discard if not={alpha}{1.0},
					on layer=main
					] table [
					x=grid_size,
					y=route,
					]{data/revision/results_1.0_cov_6_1_3_1.txt};\addlegendentry{Optimal Route};

     				\addplot+ [
					discard if not={alpha}{1.0},
					] table [
					x=grid_size,
					y=route_grow,
					]{data/revision/results_1.0_cov_6_1_3_1.txt};\addlegendentry{Optimal Route GN};
					
					\addplot+ [
					discard if not={alpha}{1.0},
					on layer=main
					] table [
					x=grid_size,
					y=bayes_optimal,
					]{data/revision/results_1.0_cov_6_1_3_1.txt};\addlegendentry{Optimal};
					
					\addplot+ [
					discard if not={alpha}{1.0},
					on layer=main
					] table [
					x=grid_size,
					y=lb,
					]{data/revision/results_1.0_cov_6_1_3_1.txt};\addlegendentry{Lower Bound};
				\end{axis}
			\end{tikzpicture}
			\caption{$N= p^{1.0}$}
			\label{fig:risk_1.0_1_3_1_1.0}
		\end{subfigure}
		~~
		\begin{subfigure}[h]{0.45\textwidth}
			\begin{tikzpicture}
				\begin{axis}[
					width=2.7in,
					xlabel = Grid Size ($p$),
					ylabel = $\log_{10}\textrm{(Average Risk)}$,
					legend style={legend pos=south east, font=\tiny},
					label style={font=\footnotesize},
					ymin=-1.7,
					ymax=1.5,
					xtick={10,15,20,25,30}
					]
					\addplot+ [
					discard if not={alpha}{2.0},
					] table [
					x=grid_size,
					y=seg_simple,
					]{data/revision/results_1.0_cov_6_1_3_1.txt};\addlegendentry{Simple Segment};

					\addplot+ [
					discard if not={alpha}{2.0},
					] table [
					x=grid_size,
					y=route,
					]{data/revision/results_1.0_cov_6_1_3_1.txt};\addlegendentry{Optimal Route};

     				\addplot+ [
					discard if not={alpha}{2.0},
					] table [
					x=grid_size,
					y=route_grow,
					]{data/revision/results_1.0_cov_6_1_3_1.txt};\addlegendentry{Optimal Route GN};
					
					\addplot+ [
					discard if not={alpha}{2.0},
					on layer=main
					] table [
					x=grid_size,
					y=bayes_optimal,
					]{data/revision/results_1.0_cov_6_1_3_1.txt};\addlegendentry{Optimal};
					
					\addplot+ [
					discard if not={alpha}{2.0},
					on layer=main
					] table [
					x=grid_size,
					y=lb,
					]{data/revision/results_1.0_cov_6_1_3_1.txt};\addlegendentry{Lower Bound};
				\end{axis}
			\end{tikzpicture}
			\caption{$N= p^{2.0}$}
			\label{fig:risk_1.0_1_3_1_2.0}
		\end{subfigure}
		
		\begin{subfigure}[h]{0.45\textwidth}
			\begin{tikzpicture}
				\begin{axis}[
					width=2.7in,
					xlabel = Grid Size ($p$),
					ylabel = $\log_{10}\textrm{(Average Risk)}$,
					legend style={legend pos=south west, font=\tiny},
					label style={font=\footnotesize},
					ymin=-4,
					ymax=1.5,
					xtick={10,15,20,25,30},
					]
					\addplot+ [
					discard if not={alpha}{3.0},
					] table [
					x=grid_size,
					y=seg_simple,
					]{data/revision/results_1.0_cov_6_1_3_1.txt};\addlegendentry{Simple Segment};

					\addplot+ [
					discard if not={alpha}{3.0},
					] table [
					x=grid_size,
					y=route,
					]{data/revision/results_1.0_cov_6_1_3_1.txt};\addlegendentry{Optimal Route};

					\addplot+ [
					discard if not={alpha}{3.0},
					] table [
					x=grid_size,
					y=route_grow,
					]{data/revision/results_1.0_cov_6_1_3_1.txt};\addlegendentry{Optimal Route GN};
					
					\addplot+ [
					discard if not={alpha}{3.0},
					on layer=main
					] table [
					x=grid_size,
					y=bayes_optimal,
					]{data/revision/results_1.0_cov_6_1_3_1.txt};\addlegendentry{Optimal};
					
					\addplot+ [
					discard if not={alpha}{3.0},
					on layer=main
					] table [
					x=grid_size,
					y=lb,
					]{data/revision/results_1.0_cov_6_1_3_1.txt};\addlegendentry{Lower Bound};
				\end{axis}
			\end{tikzpicture}
			\caption{$N= p^{3.0}$}
			\label{fig:risk_1.0_1_3_1_3.0}
		\end{subfigure}
		~~
		\begin{subfigure}[h]{0.45\textwidth}
			\begin{tikzpicture}
				\begin{axis}[
					width=2.7in,
					xlabel = Grid Size ($p$),
					ylabel = $\log_{10}\textrm{(Average Risk)}$,
					legend style={legend pos=south west, font=\tiny},
					label style={font=\footnotesize},
					ymin=-6,
					ymax=1.5,
					xtick={10,15,20,25,30},
					]
					\addplot+ [
					discard if not={alpha}{4.0},
					] table [
					x=grid_size,
					y=seg_simple,
					]{data/revision/results_1.0_cov_6_1_3_1.txt};\addlegendentry{Simple Segment};

					\addplot+ [
					discard if not={alpha}{4.0},
					] table [
					x=grid_size,
					y=route,
					]{data/revision/results_1.0_cov_6_1_3_1.txt};\addlegendentry{Optimal Route};

					\addplot+ [
					discard if not={alpha}{4.0},
					] table [
					x=grid_size,
					y=route_grow,
					]{data/revision/results_1.0_cov_6_1_3_1.txt};\addlegendentry{Optimal Route GN};
					
					\addplot+ [
					discard if not={alpha}{4.0},
					on layer=main
					] table [
					x=grid_size,
					y=bayes_optimal,
					]{data/revision/results_1.0_cov_6_1_3_1.txt};\addlegendentry{Optimal};
					
					\addplot+ [
					discard if not={alpha}{4.0},
					on layer=main
					] table [
					x=grid_size,
					y=lb,
					]{data/revision/results_1.0_cov_6_1_3_1.txt};\addlegendentry{Lower Bound};
				\end{axis}
			\end{tikzpicture}
			\caption{$N= p^{4.0}$}
			\label{fig:risk_1.0_1_3_1_4.0}
		\end{subfigure}
		
		\caption{Average integrated risks of different estimators ($\alpha = 1.0$, covariance is $e^{-3\mathscr{L}} + I$).}
		\label{fig:risk_1.0_1_3_1}
	\end{figure}

	\begin{figure}[h]
		\captionsetup[subfigure]{justification=centering}
		\centering 
		\begin{subfigure}[h]{0.45\textwidth}
			\begin{tikzpicture}
				\begin{axis}[
					width=2.7in,
					xlabel = Grid Size ($p$),
					ylabel = $\log_{10}\textrm{(Average Risk)}$,
					legend style={legend pos=south east, font=\tiny},
					label style={font=\footnotesize},
					ymin=0,
					ymax=1.3,
					xtick={10,15,20,25,30}
					]
					
					\addplot+ [
					discard if not={alpha}{1.0},
					on layer=background
					] table [
					x=grid_size,
					y=seg_simple,
					]{data/revision/results_1.0_cov_6_3_1_1.txt};\addlegendentry{Simple Segment};
					
					\addplot+ [
					discard if not={alpha}{1.0},
					on layer=main
					] table [
					x=grid_size,
					y=route,
					]{data/revision/results_1.0_cov_6_3_1_1.txt};\addlegendentry{Optimal Route};

     				\addplot+ [
					discard if not={alpha}{1.0},
					] table [
					x=grid_size,
					y=route_grow,
					]{data/revision/results_1.0_cov_6_3_1_1.txt};\addlegendentry{Optimal Route GN};
					
					\addplot+ [
					discard if not={alpha}{1.0},
					on layer=main
					] table [
					x=grid_size,
					y=bayes_optimal,
					]{data/revision/results_1.0_cov_6_3_1_1.txt};\addlegendentry{Optimal};
					
					\addplot+ [
					discard if not={alpha}{1.0},
					on layer=main
					] table [
					x=grid_size,
					y=lb,
					]{data/revision/results_1.0_cov_6_3_1_1.txt};\addlegendentry{Lower Bound};
				\end{axis}
			\end{tikzpicture}
			\caption{$N= p^{1.0}$}
			\label{fig:risk_1.0_3_1_1_1.0}
		\end{subfigure}
		~~
		\begin{subfigure}[h]{0.45\textwidth}
			\begin{tikzpicture}
				\begin{axis}[
					width=2.7in,
					xlabel = Grid Size ($p$),
					ylabel = $\log_{10}\textrm{(Average Risk)}$,
					legend style={legend pos=south east, font=\tiny},
					label style={font=\footnotesize},
					ymin=-1.7,
					ymax=1.5,
					xtick={10,15,20,25,30}
					]
					\addplot+ [
					discard if not={alpha}{2.0},
					] table [
					x=grid_size,
					y=seg_simple,
					]{data/revision/results_1.0_cov_6_3_1_1.txt};\addlegendentry{Simple Segment};

					\addplot+ [
					discard if not={alpha}{2.0},
					] table [
					x=grid_size,
					y=route,
					]{data/revision/results_1.0_cov_6_3_1_1.txt};\addlegendentry{Optimal Route};

     				\addplot+ [
					discard if not={alpha}{2.0},
					] table [
					x=grid_size,
					y=route_grow,
					]{data/revision/results_1.0_cov_6_3_1_1.txt};\addlegendentry{Optimal Route GN};
					
					\addplot+ [
					discard if not={alpha}{2.0},
					on layer=main
					] table [
					x=grid_size,
					y=bayes_optimal,
					]{data/revision/results_1.0_cov_6_3_1_1.txt};\addlegendentry{Optimal};
					
					\addplot+ [
					discard if not={alpha}{2.0},
					on layer=main
					] table [
					x=grid_size,
					y=lb,
					]{data/revision/results_1.0_cov_6_3_1_1.txt};\addlegendentry{Lower Bound};
				\end{axis}
			\end{tikzpicture}
			\caption{$N= p^{2.0}$}
			\label{fig:risk_1.0_3_1_1_2.0}
		\end{subfigure}
		
		\begin{subfigure}[h]{0.45\textwidth}
			\begin{tikzpicture}
				\begin{axis}[
					width=2.7in,
					xlabel = Grid Size ($p$),
					ylabel = $\log_{10}\textrm{(Average Risk)}$,
					legend style={legend pos=south west, font=\tiny},
					label style={font=\footnotesize},
					ymin=-4,
					ymax=1.5,
					xtick={10,15,20,25,30},
					]
					\addplot+ [
					discard if not={alpha}{3.0},
					] table [
					x=grid_size,
					y=seg_simple,
					]{data/revision/results_1.0_cov_6_3_1_1.txt};\addlegendentry{Simple Segment};

					\addplot+ [
					discard if not={alpha}{3.0},
					] table [
					x=grid_size,
					y=route,
					]{data/revision/results_1.0_cov_6_3_1_1.txt};\addlegendentry{Optimal Route};

					\addplot+ [
					discard if not={alpha}{3.0},
					] table [
					x=grid_size,
					y=route_grow,
					]{data/revision/results_1.0_cov_6_3_1_1.txt};\addlegendentry{Optimal Route GN};
					
					\addplot+ [
					discard if not={alpha}{3.0},
					on layer=main
					] table [
					x=grid_size,
					y=bayes_optimal,
					]{data/revision/results_1.0_cov_6_3_1_1.txt};\addlegendentry{Optimal};
					
					\addplot+ [
					discard if not={alpha}{3.0},
					on layer=main
					] table [
					x=grid_size,
					y=lb,
					]{data/revision/results_1.0_cov_6_3_1_1.txt};\addlegendentry{Lower Bound};
				\end{axis}
			\end{tikzpicture}
			\caption{$N= p^{3.0}$}
			\label{fig:risk_1.0_3_1_1_3.0}
		\end{subfigure}
		~~
		\begin{subfigure}[h]{0.45\textwidth}
			\begin{tikzpicture}
				\begin{axis}[
					width=2.7in,
					xlabel = Grid Size ($p$),
					ylabel = $\log_{10}\textrm{(Average Risk)}$,
					legend style={legend pos=south west, font=\tiny},
					label style={font=\footnotesize},
					ymin=-6,
					ymax=1.5,
					xtick={10,15,20,25,30},
					]
					\addplot+ [
					discard if not={alpha}{4.0},
					] table [
					x=grid_size,
					y=seg_simple,
					]{data/revision/results_1.0_cov_6_3_1_1.txt};\addlegendentry{Simple Segment};

					\addplot+ [
					discard if not={alpha}{4.0},
					] table [
					x=grid_size,
					y=route,
					]{data/revision/results_1.0_cov_6_3_1_1.txt};\addlegendentry{Optimal Route};

					\addplot+ [
					discard if not={alpha}{4.0},
					] table [
					x=grid_size,
					y=route_grow,
					]{data/revision/results_1.0_cov_6_3_1_1.txt};\addlegendentry{Optimal Route GN};
					
					\addplot+ [
					discard if not={alpha}{4.0},
					on layer=main
					] table [
					x=grid_size,
					y=bayes_optimal,
					]{data/revision/results_1.0_cov_6_3_1_1.txt};\addlegendentry{Optimal};
					
					\addplot+ [
					discard if not={alpha}{4.0},
					on layer=main
					] table [
					x=grid_size,
					y=lb,
					]{data/revision/results_1.0_cov_6_3_1_1.txt};\addlegendentry{Lower Bound};
				\end{axis}
			\end{tikzpicture}
			\caption{$N= p^{4.0}$}
			\label{fig:risk_1.0_3_1_1_4.0}
		\end{subfigure}
		
		\caption{Average integrated risks of different estimators ($\alpha = 1.0$, covariance is $3e^{-\mathscr{L}} + I$).}
		\label{fig:risk_1.0_3_1_1}
	\end{figure}
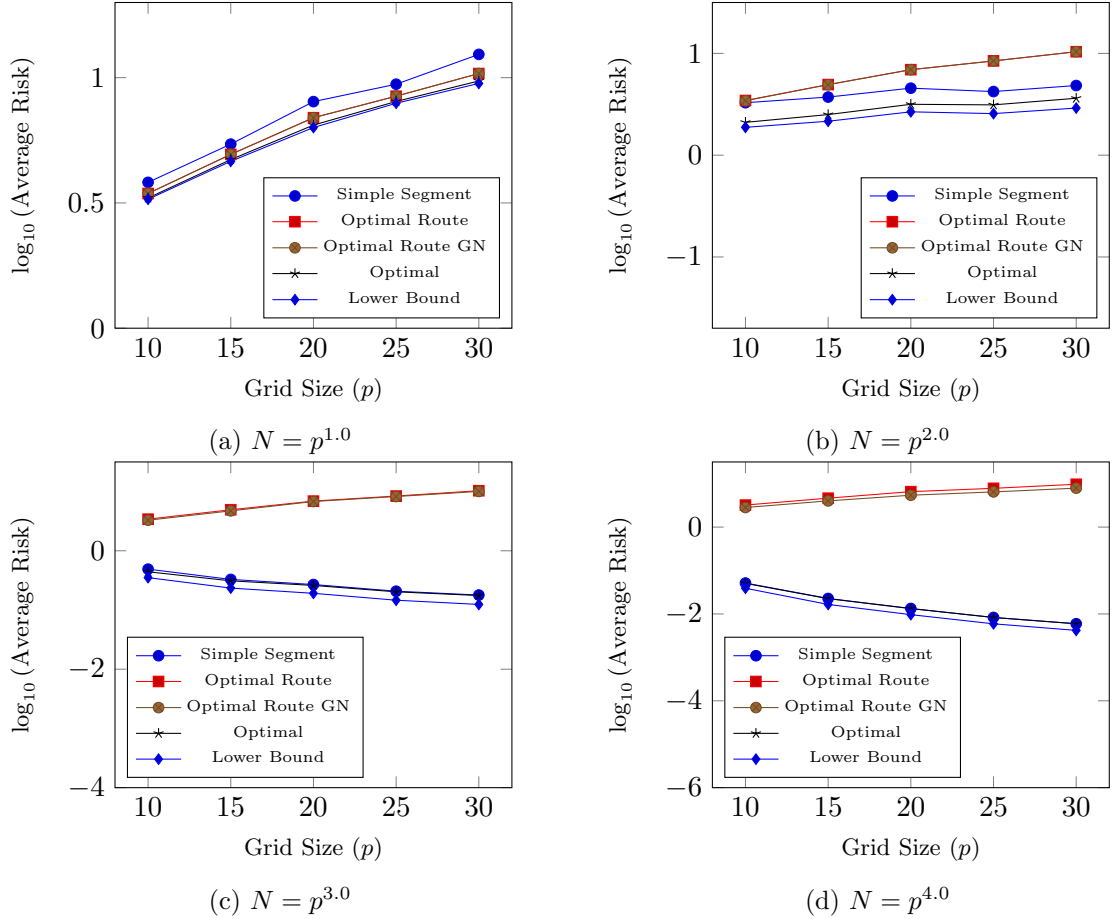

	\subsection{Entries in $K_p$ are drawn from $\mathcal{U}_{[-1,1]}$}
	\label{subsec:-1_1}
	
	In the first case, $K_p$ is a random matrix whose elements are drawn from a uniform distribution between $[-1,1]$. It can be checked that both the covariance matrix $\Sigma_p = (1/|\mathcal{S}_p|^2)K_p^\intercal K_p$ and the precision matrix $\Psi_p=\Sigma^{-1}_p$ violate the second part of Assumption \ref{assumption:3}. The variance of the means of segment travel times $\tau^2$ is set to be $0.5$ which is similar to the variance of the segment travel times $\sigma^2_s$. The rest of the experimental setups are the same as those in \Cref{sec:numerical}. %
    \Cref{fig:cov_2_risk_1.0} shows similar trends for the integrated risks of the simple segment-based estimators and the optimal route-based estimators, as those in \Cref{fig:risk_1.0_1_1_1}. This is not too much out of expectation as the proof of \Cref{thm:asymptotic_general_seg_better_route} and
	\Cref{cor:grid} do not require spatial decay of precision matrix $\Psi_p$. Moreover, the entries in the covariance matrix are mostly dominated by the diagonal and the off-diagonal entries sum up to zero in expectation. This roughly gives $\sum_{s,t\in\mathcal{S}_p}\sigma_{s,t} = \mathcal{O}(|\mathcal{S}_p|)=\mathcal{O}(p^2)$. %

	\begin{figure}[htbp]
		\captionsetup[subfigure]{justification=centering}
		\centering 
		
		\begin{subfigure}[h]{0.45\textwidth}
			\begin{tikzpicture}
				\begin{axis}[
					width=2.7in,
					xlabel = Grid Size ($p$),
					ylabel = $\log_{10}\textrm{(Average Risk)}$,
					legend style={legend pos=south east, font=\tiny},
					label style={font=\footnotesize},
					ymin=-0.5,
					ymax=1.2,
					xtick={10,15,20,25,30}
					]
					\addplot+ [
					discard if not={alpha}{1.0},
					] table [
					x=grid_size,
					y=seg_simple,
					]{data/results_beta_1.0_cov_4.txt};\addlegendentry{Simple Segment};
					
					\addplot+ [
					discard if not={alpha}{1.0},
					] table [
					x=grid_size,
					y=route,
					]{data/results_beta_1.0_cov_4.txt};\addlegendentry{Optimal Route};
					
					\addplot+ [
					discard if not={alpha}{1.0},
					] table [
					x=grid_size,
					y=bayes_optimal,
					]{data/results_beta_1.0_cov_4.txt};\addlegendentry{Optimal};

					\addplot+ [
					discard if not={alpha}{1.0},
					] table [
					x=grid_size,
					y=lb,
					]{data/results_beta_1.0_cov_4.txt};\addlegendentry{Lower Bound};
				\end{axis}
			\end{tikzpicture}
			\caption{$N= p^{1.0}$}
			\label{fig:cov_2_risk_1.0_1.0}
		\end{subfigure}
		~~
		\begin{subfigure}[h]{0.45\textwidth}
			\begin{tikzpicture}
				\begin{axis}[
					width=2.7in,
					xlabel = Grid Size ($p$),
					ylabel = $\log_{10}\textrm{(Average Risk)}$,
					legend style={legend pos=south east, font=\tiny},
					label style={font=\footnotesize},
					ymin=-10,
					ymax=1.5,
					xtick={10,15,20,25,30}
					]
					\addplot+ [
					discard if not={alpha}{2.0},
					] table [
					x=grid_size,
					y=seg_simple,
					]{data/results_beta_1.0_cov_4.txt};\addlegendentry{Simple Segment};
					
					\addplot+ [
					discard if not={alpha}{2.0},
					] table [
					x=grid_size,
					y=route,
					]{data/results_beta_1.0_cov_4.txt};\addlegendentry{Optimal Route};
					
					\addplot+ [
					discard if not={alpha}{2.0},
					] table [
					x=grid_size,
					y=bayes_optimal,
					]{data/results_beta_1.0_cov_4.txt};\addlegendentry{Optimal};
					
					\addplot+ [
					discard if not={alpha}{2.0},
					] table [
					x=grid_size,
					y=lb,
					]{data/results_beta_1.0_cov_4.txt};\addlegendentry{Lower Bound};
				\end{axis}
			\end{tikzpicture}
			\caption{$N= p^{2.0}$}
			\label{fig:cov_2_risk_1.0_2.0}
		\end{subfigure}
		
		\begin{subfigure}[h]{0.45\textwidth}
			\begin{tikzpicture}
				\begin{axis}[
					width=2.7in,
					xlabel = Grid Size ($p$),
					ylabel = $\log_{10}\textrm{(Average Risk)}$,
					legend style={legend pos=south west, font=\tiny},
					label style={font=\footnotesize},
					ymin=-11.0,
					ymax=1.5,
					xtick={10,15,20,25,30},
					]
					\addplot+ [
					discard if not={alpha}{3.0},
					] table [
					x=grid_size,
					y=seg_simple,
					]{data/results_beta_1.0_cov_4.txt};\addlegendentry{Simple Segment};
					
					\addplot+ [
					discard if not={alpha}{3.0},
					] table [
					x=grid_size,
					y=route,
					]{data/results_beta_1.0_cov_4.txt};\addlegendentry{Optimal Route};
					
					\addplot+ [
					discard if not={alpha}{3.0},
					] table [
					x=grid_size,
					y=bayes_optimal,
					]{data/results_beta_1.0_cov_4.txt};\addlegendentry{Optimal};
					
					\addplot+ [
					discard if not={alpha}{3.0},
					] table [
					x=grid_size,
					y=lb,
					]{data/results_beta_1.0_cov_4.txt};\addlegendentry{Lower Bound};
				\end{axis}
			\end{tikzpicture}
			\caption{$N= p^{3.0}$}
			\label{fig:cov_2_risk_1.0_3.0}
		\end{subfigure}
		~~
		\begin{subfigure}[h]{0.45\textwidth}
			\begin{tikzpicture}
				\begin{axis}[
					width=2.7in,
					xlabel = Grid Size ($p$),
					ylabel = $\log_{10}\textrm{(Average Risk)}$,
					legend style={legend pos=south west, font=\tiny},
					label style={font=\footnotesize},
					ymin=-11,
					ymax=1.2,
					xtick={10,15,20,25,30},
					]
					\addplot+ [
					discard if not={alpha}{4.0},
					] table [
					x=grid_size,
					y=seg_simple,
					]{data/results_beta_1.0_cov_4.txt};\addlegendentry{Simple Segment};
					
					\addplot+ [
					discard if not={alpha}{4.0},
					] table [
					x=grid_size,
					y=route,
					]{data/results_beta_1.0_cov_4.txt};\addlegendentry{Optimal Route};
					
					\addplot+ [
					discard if not={alpha}{4.0},
					] table [
					x=grid_size,
					y=bayes_optimal,
					]{data/results_beta_1.0_cov_4.txt};\addlegendentry{Optimal};
					
					\addplot+ [
					discard if not={alpha}{4.0},
					] table [
					x=grid_size,
					y=lb,
					]{data/results_beta_1.0_cov_4.txt};\addlegendentry{Lower Bound};
				\end{axis}
			\end{tikzpicture}
			\caption{$N= p^{4.0}$}
			\label{fig:cov_2_risk_1.0_4.0}
		\end{subfigure}
		\caption{Average integrated risks of different estimators ($\alpha = 1.0$). Covariance of segment travel times is $\Sigma_p = K_p^\intercal K_p/|\mathcal{S}_p|^2$ where $K_p$ is a random matrix whose entries are drawn from $\mathcal{U}_{[-1,1]}$.}
		\label{fig:cov_2_risk_1.0}
	\end{figure}
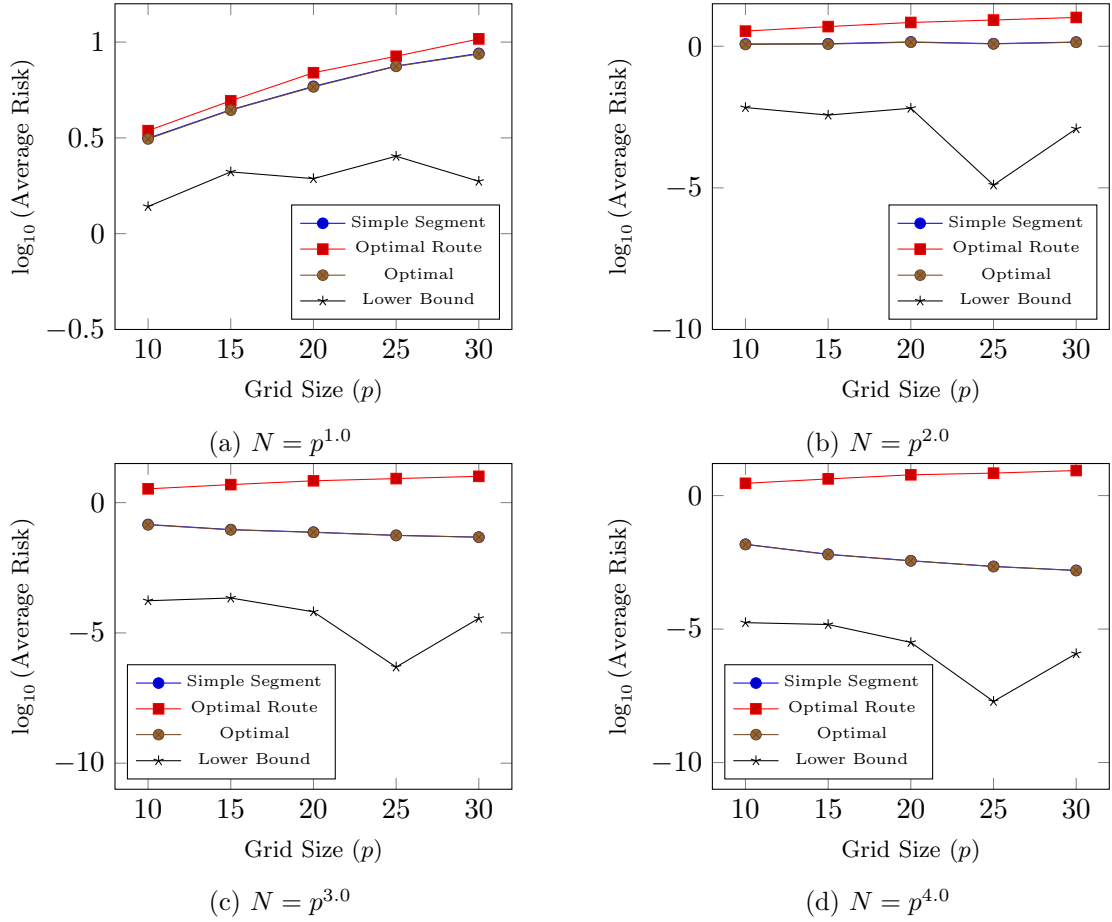
	
	On the other hand, the information-theoretic lower bounds seem to be of lower order compared to the integrated risks of the simple segment-based estimator. This is expected as the proof of \Cref{thm:asymptotic_opt} critically uses the spatial decay of the precision matrix $\Psi_p$. What is surprising is that the simple segment-based estimator is still highly competitive compared to the optimal estimator --- its risk is very close to the optimal risk. %

	\subsection{Entries in $K_p$ are drawn from $\mathcal{U}_{[0,1]}$}
	\label{subsec:0_1}
	
	In the second case, $K_p$ is a random matrix whose elements are drawn from a uniform distribution between $[0,1]$. It can be checked that both the covariance matrix $\Sigma_p = (1/|\mathcal{S}_p|^2)K_p^\intercal K_p$ and the precision matrix $\Psi_p=\Sigma^{-1}_p$ violate the second part of Assumption \ref{assumption:3}. In particular, we have $\sum_{s,t\in\mathcal{S}_p}|\sigma_{s,t}| = \sum_{s,t\in\mathcal{S}_p}\sigma_{s,t} = \mathcal{O}(|\mathcal{S}_p|^2)$. The variance of the means of segment travel times $\tau^2$ is set to be $0.5$ which is similar to the variance of the segment travel times $\sigma^2_s$. %
 \Cref{fig:cov_2_risk_1.0} now shows quite different trends for the integrated risks of the simple segment-based estimators and the optimal route-based estimators, compared to those in \Cref{fig:risk_1.0_1_1_1}. Following the same proof of \Cref{thm:asymptotic_general_seg_better_route}, one can show that $R\big(\ThetasegP\big)=\mathcal{O}(|\mathcal{S}_p|^2/N)=\mathcal{O}(p^4/N)$ which is reflected by Figures \ref{fig:cov_3_risk_1.0}. %
	The simple segment-based estimator with $\phi_s(N_s)=N_s/(N_s + 1),\:\forall s\in y$ weighs too much on the training data, and since segment travel times in this setting have much higher variance, the optimal segment-based estimator places more weights on the prior mean when the sample size is small. This results in poor performance of the simple segment-based estimators (outperformed by the optimal route-based estimator) when sample size is small. The performance of the simple segment-based estimator improves significantly as the sample size gets larger, and eventually dominates the optimal route-based estimator. %

	\begin{figure}[htbp]
		\captionsetup[subfigure]{justification=centering}
		\centering 
		
		\begin{subfigure}[h]{0.45\textwidth}
			\begin{tikzpicture}
				\begin{axis}[
					width=2.7in,
					xlabel = Grid Size ($p$),
					ylabel = $\log_{10}\textrm{(Average Risk)}$,
					legend style={legend pos=south east, font=\tiny},
					label style={font=\footnotesize},
					ymin=-0.75,
					ymax=1.5,
					xtick={10,15,20,25,30}
					]
					\addplot+ [
					discard if not={alpha}{1.0},
					] table [
					x=grid_size,
					y=seg_simple,
					]{data/results_beta_1.0_cov_5.txt};\addlegendentry{Simple Segment};
					
					\addplot+ [
					discard if not={alpha}{1.0},
					] table [
					x=grid_size,
					y=route,
					]{data/results_beta_1.0_cov_5.txt};\addlegendentry{Optimal Route};
					
					\addplot+ [
					discard if not={alpha}{1.0},
					on layer=main
					] table [
					x=grid_size,
					y=bayes_optimal,
					]{data/results_beta_1.0_cov_5.txt};\addlegendentry{Optimal};
					
					\addplot+ [
					discard if not={alpha}{1.0},
					on layer=main
					] table [
					x=grid_size,
					y=lb,
					]{data/results_beta_1.0_cov_5.txt};\addlegendentry{Lower Bound};
					
				\end{axis}
			\end{tikzpicture}
			\caption{$N= p^{1.0}$}
			\label{fig:cov_3_risk_1.0_1.0}
		\end{subfigure}
		~~
		\begin{subfigure}[h]{0.45\textwidth}
			\begin{tikzpicture}
				\begin{axis}[
					width=2.7in,
					xlabel = Grid Size ($p$),
					ylabel = $\log_{10}\textrm{(Average Risk)}$,
					legend style={legend pos=south east, font=\tiny},
					label style={font=\footnotesize},
					ymin=-10,
					ymax=1.2,
					xtick={10,15,20,25,30}
					]
					\addplot+ [
					discard if not={alpha}{2.0},
					] table [
					x=grid_size,
					y=seg_simple,
					]{data/results_beta_1.0_cov_5.txt};\addlegendentry{Simple Segment};
					
					\addplot+ [
					discard if not={alpha}{2.0},
					] table [
					x=grid_size,
					y=route,
					]{data/results_beta_1.0_cov_5.txt};\addlegendentry{Optimal Route};
					
					\addplot+ [
					discard if not={alpha}{2.0},
					on layer=main
					] table [
					x=grid_size,
					y=bayes_optimal,
					]{data/results_beta_1.0_cov_5.txt};\addlegendentry{Optimal};
					
					\addplot+ [
					discard if not={alpha}{2.0},
					on layer=main
					] table [
					x=grid_size,
					y=lb,
					]{data/results_beta_1.0_cov_5.txt};\addlegendentry{Lower Bound};
					
				\end{axis}
			\end{tikzpicture}
			\caption{$N= p^{2.0}$}
			\label{fig:cov_3_risk_1.0_2.0}
		\end{subfigure}
		
		\begin{subfigure}[h]{0.45\textwidth}
			\begin{tikzpicture}
				\begin{axis}[
					width=2.7in,
					xlabel = Grid Size ($p$),
					ylabel = $\log_{10}\textrm{(Average Risk)}$,
					legend style={legend pos=south west, font=\tiny},
					label style={font=\footnotesize},
					ymin=-13,
					ymax=1.5,
					xtick={10,15,20,25,30},
					]
					\addplot+ [
					discard if not={alpha}{3.0},
					] table [
					x=grid_size,
					y=seg_simple,
					]{data/results_beta_1.0_cov_5.txt};\addlegendentry{Simple Segment};
					
					\addplot+ [
					discard if not={alpha}{3.0},
					] table [
					x=grid_size,
					y=route,
					]{data/results_beta_1.0_cov_5.txt};\addlegendentry{Optimal Route};
					
					\addplot+ [
					discard if not={alpha}{3.0},
					on layer=main
					] table [
					x=grid_size,
					y=bayes_optimal,
					]{data/results_beta_1.0_cov_5.txt};\addlegendentry{Optimal};
					
					\addplot+ [
					discard if not={alpha}{3.0},
					on layer=main
					] table [
					x=grid_size,
					y=lb,
					]{data/results_beta_1.0_cov_5.txt};\addlegendentry{Lower Bound};
					
				\end{axis}
			\end{tikzpicture}
			\caption{$N= p^{3.0}$}
			\label{fig:cov_3_risk_1.0_3.0}
		\end{subfigure}
		~~
		\begin{subfigure}[h]{0.45\textwidth}
			\begin{tikzpicture}
				\begin{axis}[
					width=2.7in,
					xlabel = Grid Size ($p$),
					ylabel = $\log_{10}\textrm{(Average Risk)}$,
					legend style={legend pos=south west, font=\tiny},
					label style={font=\footnotesize},
					ymin=-15,
					ymax=1.5,
					xtick={10,15,20,25,30},
					]
					\addplot+ [
					discard if not={alpha}{4.0},
					] table [
					x=grid_size,
					y=seg_simple,
					]{data/results_beta_1.0_cov_5.txt};\addlegendentry{Simple Segment};
					
					\addplot+ [
					discard if not={alpha}{4.0},
					] table [
					x=grid_size,
					y=route,
					]{data/results_beta_1.0_cov_5.txt};\addlegendentry{Optimal Route};
					
					\addplot+ [
					discard if not={alpha}{4.0},
					on layer=main
					] table [
					x=grid_size,
					y=bayes_optimal,
					]{data/results_beta_1.0_cov_5.txt};\addlegendentry{Optimal};
					
					\addplot+ [
					discard if not={alpha}{4.0},
					on layer=main
					] table [
					x=grid_size,
					y=lb,
					]{data/results_beta_1.0_cov_5.txt};\addlegendentry{Lower Bound};

				\end{axis}
			\end{tikzpicture}
			\caption{$N= p^{4.0}$}
			\label{fig:cov_3_risk_1.0_4.0}
		\end{subfigure}
		\caption{Average integrated risks of different estimators ($\alpha = 1.0$). Covariance of segment travel times is $\Sigma_p = K_p^\intercal K_p/|\mathcal{S}_p|^2$ where $K_p$ is a random matrix whose entries are drawn from $\mathcal{U}_{[0,1]}$.}
		\label{fig:cov_3_risk_1.0}
	\end{figure}
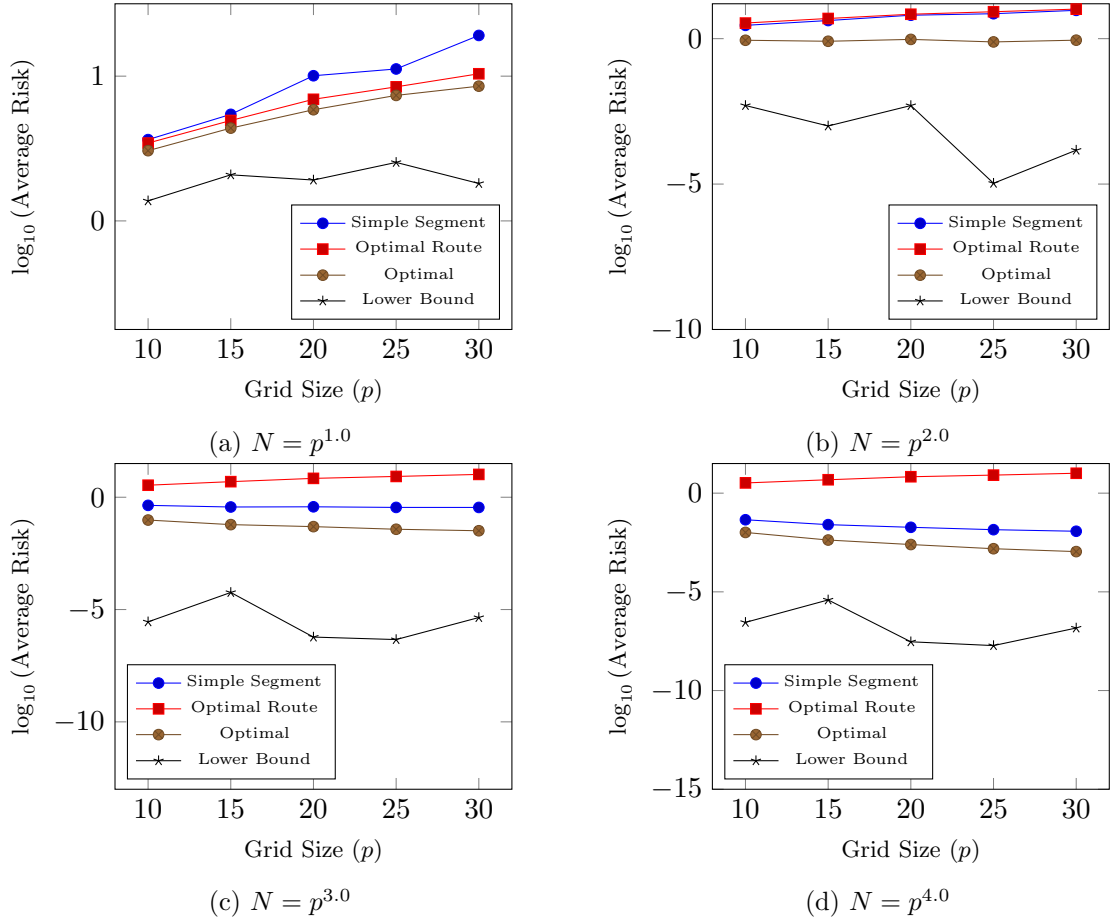
	
	Similarly to \Cref{fig:cov_2_risk_1.0}, the information-theoretic lower bounds are of lower order compared to the risks of the simple segment-based estimator. The difference is that the gap between the simple segment-based estimator and the optimal estimator gets larger.
	These observations along with those in Section \ref{subsec:0_1} suggest that controlling the covariance matrix of the segment travel times is likely much more important than controlling the precision matrix in maintaining the optimality of the simple segment-based estimator.

\end{document}